\def\confversion{0}
\def\ifconf{\ifnum\confversion=1}
\def\ifnotconf{\ifnum\confversion=0}

\documentclass[11 pt]{article}
\def\showauthornotes{0}
\def\showkeys{0}
\def\showdraftbox{0}

\usepackage{xspace,xcolor,enumerate}
\usepackage{amsmath,amssymb}
\usepackage{amsthm}
\usepackage[toc,page]{appendix}
\usepackage{thmtools}
\usepackage{thm-restate}
\usepackage{color,graphicx}
\usepackage{boxedminipage}
\usepackage{makecell}

\ifnum\showkeys=1
\usepackage[color]{showkeys}
\fi

\definecolor{darkred}{rgb}{0.5,0,0}
\definecolor{darkgreen}{rgb}{0,0.35,0}
\definecolor{darkblue}{rgb}{0,0,0.55}

\usepackage[pdfstartview=FitH,pdfpagemode=UseNone,colorlinks,linkcolor=darkblue,filecolor=darkred,citecolor=darkgreen,urlcolor=darkred,pagebackref]{hyperref}

\usepackage[capitalise,nameinlink]{cleveref}
\usepackage[T1]{fontenc}
\usepackage{mathtools,dsfont,bbm}
\usepackage{mathpazo}
\ifnum\showauthornotes=1
\newcommand{\Authornote}[2]{{\sf\small\color{red}{[#1: #2]}}}
\newcommand{\Authorcomment}[2]{{\sf \small\color{gray}{[#1: #2]}}}
\newcommand{\Authorfnote}[2]{\footnote{\color{red}{#1: #2}}}
\else
\newcommand{\Authornote}[2]{}
\newcommand{\Authorcomment}[2]{}
\newcommand{\Authorfnote}[2]{}
\fi

\ifnum\showdraftbox=1

\else

\fi

\newtheorem{theorem}{Theorem}[section]
\newtheorem{conjecture}[theorem]{Conjecture}

\newtheorem{definition}[theorem]{Definition}

\newtheorem{lemma}[theorem]{Lemma}
\newtheorem{remark}[theorem]{Remark}
\newtheorem{proposition}[theorem]{Proposition}
\newtheorem{corollary}[theorem]{Corollary}
\newtheorem{claim}[theorem]{Claim}
\newtheorem{fact}[theorem]{Fact}

\newtheorem{remk}[theorem]{Remark}

\newtheorem{algo}[theorem]{Algorithm}

\def\FullBox{\hbox{\vrule width 6pt height 6pt depth 0pt}}

\def\qed{\ifmmode\qquad\FullBox\else{\unskip\nobreak\hfil
\penalty50\hskip1em\null\nobreak\hfil\FullBox
\parfillskip=0pt\finalhyphendemerits=0\endgraf}\fi}

\def\qedsketch{\ifmmode\Box\else{\unskip\nobreak\hfil
\penalty50\hskip1em\null\nobreak\hfil$\Box$
\parfillskip=0pt\finalhyphendemerits=0\endgraf}\fi}

\newenvironment{proofof}[1]{\begin{trivlist} \item {\bf Proof of
#1:~~}}
  {\qed\end{trivlist}}

\def\to{\rightarrow}
\def\eps{\varepsilon}
\def\epsilon{\varepsilon}

\def\eps{\epsilon}

\def\phi{\varphi}
\def\cal{\mathcal}

\def\psdgeq{\succeq} 
\newcommand{\defeq}{:=}

\newcommand{\ie}{i.e.,\xspace}

\newcommand{\R}{{\mathbb R}}
\newcommand{\E}{{\mathbb E}}

\newcommand{\N}{{\mathbb{N}}}

\newcommand{\F}{{\mathbb F}}

\newcommand{\gauss}[2]{{\cal N(#1, #2)}}

\usepackage{nicefrac}
\newcommand{\abs}[1]{\ensuremath{\left\lvert #1 \right\rvert}}

\newcommand{\norm}[1]{\ensuremath{\left\lVert #1 \right\rVert}}

\newcommand{\ip}[2] {\ensuremath{\langle #1 , #2 \rangle}}

\newcommand{\one}{{\mathbf{1}}}

\makeatletter

\def\ProbabilityRender#1#2{%fancy probability command
  \@ifnextchar\bgroup%
  {\renderwithdist{#1}{#2}}
   {\singlervrender{#1}{#2}}
}
\def\singlervrender#1#2{%
   \ensuremath{\mathchoice
       {{#1}\left[ #2 \right]}
       {{#1}[ #2 ]}
       {{#1}[ #2 ]}
       {{#1}[ #2 ]}
   }
}
\def\renderwithdist#1#2#3{%
   \@ifnextchar\bgroup
   {\superfancyrender{#1}{#2}{#3}}
   {\ensuremath{\mathchoice
      {\underset{#2}{#1}\left[ #3 \right]}
      {{#1}_{#2}[ #3 ]}
      {{#1}_{#2}[ #3 ]}
      {{#1}_{#2}[ #3 ]}
     }
   }
}
\def\superfancyrender#1#2#3#4#5{
   \ensuremath{\mathchoice
      {\underset{#1}{{#1}}\left#4 #3 \right#5}
      {{#1}_{#2}#4 #3 #5}
      {{#1}_{#2}#4 #3 #5}
      {{#1}_{#2}#4 #3 #5}
   }
}
\makeatother

\newfont{\inhead}{eufm10 scaled\magstep1}

\newcommand{\calA}{{\cal A}}
\newcommand{\calB}{{\cal B}}
\newcommand{\calC}{{\cal C}}
\newcommand{\calD}{{\cal D}}
\newcommand{\calE}{{\cal E}}
\newcommand{\calF}{{\cal F}}

\newcommand{\calI}{{\cal I}}

\newcommand{\calL}{{\cal L}}
\newcommand{\calM}{{\cal M}}
\newcommand{\calN}{{\cal N}}

\newcommand{\calP}{{\cal P}}

\newcommand{\calS}{{\cal S}}

\newcommand{\poly}{{\mathrm{poly}}}

\DeclareMathOperator{\aut}{\operatorname{Aut}}

\newcommand{\ceil}[1]{\ensuremath{\left\lceil #1 \right\rceil}}

\newcommand{\bigoh}{\operatorname{O}}

\newcommand{\bigomega}{\mathop{\Omega}}
\newcommand{\littleoh}{\operatorname{o}}

\DeclareSymbolFont{extraup}{U}{zavm}{m}{n}
\DeclareMathSymbol{\varheart}{\mathalpha}{extraup}{86}
\DeclareMathSymbol{\vardiamond}{\mathalpha}{extraup}{87}

\def\RR{\mathbb R}

\def\one{\mathbf 1}

\def\OPT{\mathsf{OPT}}

\DeclarePairedDelimiter\set{\lbrace}{\rbrace}

\def\multiset#1#2{\ensuremath{\left(\kern-.3em\left(\genfrac{}{}{0pt}{}{#1}{#2}\right)\kern-.3em\right)}}
\let\multichoose=\multiset

\newcommand{\pE}{\widetilde{\E}}

\DeclareMathOperator{\slice}{\mathcal{S}}

\usepackage{tikz}
\usetikzlibrary{shapes.geometric}
\usepackage{float}

\usepackage{thm-restate}

\newcommand{\unif}{\in_{\text{R}}}
\newcommand{\T}{\intercal}
\let\mod\relax
\DeclareMathOperator{\mod}{mod}
\DeclareMathOperator{\GOE}{GOE}

\DeclareMathOperator{\nullspace}{Null}
\DeclareMathOperator{\parents}{parents}
\renewcommand{\E}{\mathop{\mathbb{E}}}
\DeclareMathOperator{\body}{body}

\renewcommand*{\circle}[1]{\scalebox{0.85}{\footnotesize
    \tikz[baseline=(char.base)]{
        \node[shape=circle,draw,inner sep=1.5pt](char) {   \ifx&#1&
        \color{white} $i$
        \else
        $#1$
        \fi};
    }}}
\renewcommand*{\square}[1]{\scalebox{0.85}{\footnotesize
    \tikz[baseline=(char.base), square/.style={regular polygon,regular polygon sides=4}]{
        \node[draw,square, inner sep=0.5pt](char) {
        \ifx&#1&
        \color{white} $t$
        \else
        $#1$
        \fi};
    }}}
\usepackage[top=1in, bottom=1in, left=1in, right=1in]{geometry}

\usepackage{amsmath}
\usepackage{amssymb}
\usepackage{amsthm}
\usepackage{ytableau}
\usepackage{mathtools}
\usepackage{blkarray}

\begin{document}

\title{Sum-of-Squares Lower Bounds for Sherrington-Kirkpatrick via Planted Affine Planes}

\author{
  Mrinalkanti Ghosh\thanks{{\tt TTIC}. {\tt mkghosh@ttic.edu}. Supported in part by NSF grant CCF-1816372.}
  \and
  Fernando Granha Jeronimo\thanks{{\tt University of Chicago}. {\tt granha@uchicago.edu}. Supported in part by NSF grant CCF-1816372.}
  \and
  Chris Jones\thanks{{\tt University of Chicago}. {\tt csj@uchicago.edu}. Supported in part by NSF grant CCF-2008920.}
  \and
  Aaron Potechin\thanks{{\tt University of Chicago}. {\tt potechin@uchicago.edu}. Supported in part by NSF grant CCF-2008920.}
  \and
  Goutham Rajendran\thanks{{\tt University of Chicago}. {\tt goutham@uchicago.edu}. Supported in part by NSF grant CCF-1816372.}
}

\date{\today}

\maketitle
\thispagestyle{empty}

{\abstract{The Sum-of-Squares (SoS) hierarchy is a semi-definite
programming meta-algorithm that captures state-of-the-art polynomial
time guarantees for many optimization problems such as Max-$k$-CSPs
and Tensor PCA. On the flip side, a SoS lower bound provides evidence
of hardness, which is particularly relevant to average-case problems
for which NP-hardness may not be available.

In this paper, we consider the following average case problem, which
we call the \emph{Planted Affine Planes} (PAP) problem: Given $m$
random vectors $d_1,\ldots,d_m$ in $\mathbb{R}^n$, can we prove that
there is no vector $v \in \mathbb{R}^n$ such that for all $u \in [m]$,
$\langle v, d_u\rangle^2 = 1$? In other words, can we prove that $m$
random vectors are not all contained in two parallel hyperplanes at
equal distance from the origin? We prove that for $m \leq
n^{3/2-\eps}$, with high probability, degree-$n^{\Omega(\epsilon)}$
SoS fails to refute the existence of such a vector $v$.

When the vectors $d_1,\ldots,d_m$ are chosen from the multivariate
normal distribution, the PAP problem is equivalent to the problem of
proving that a random $n$-dimensional subspace of $\mathbb{R}^m$ does
not contain a boolean vector. As shown by Mohanty--Raghavendra--Xu
[STOC 2020], a lower bound for this problem implies a lower bound for
the problem of certifying energy upper bounds on the
Sherrington-Kirkpatrick Hamiltonian, and so our lower bound implies a
degree-$n^{\Omega(\epsilon)}$ SoS lower bound for the certification
version of the Sherrington-Kirkpatrick problem.

}}

\newpage
\ifnotconf
\pagenumbering{roman}
\tableofcontents
\clearpage
\fi

\pagenumbering{arabic}
\setcounter{page}{1}

\section{Introduction}\label{sec:intro}

The Sum-of-Squares (SoS) hierarchy is a semi-definite programming
(SDP) hierarchy which provides a meta-algorithm for polynomial
optimization~\cite{Lasserre15}. Given a polynomial objective function
and a system of polynomial equalities and inequalities as constraints,
the SoS framework specifies a family of increasingly ``larger'' SDP
programs, where each program provides a convex relaxation to the
polynomial optimization problem. The family is indexed by a size
parameter $D$ called the SoS degree. Roughly speaking, the larger the
SoS degree $D$, the tighter the relaxation, but also, the greater the
computational time required to solve the convex program, with $D=O(1)$
corresponding to polynomial time and $D = n$ able to exactly solve an
optimization problem on $n$ boolean variables. Due to the versatility
of polynomials in modeling computational problems, the SoS hierarchy
can be applied to a vast range of optimization problems. It has been
shown to be quite successful in this regard, as it captures
state-of-the-art approximation guarantees for many problems such as
Sparsest Cut~\cite{AroraRV04}, MaxCut~\cite{GW94}, Tensor
PCA~\cite{HopSS15} and all Max-$k$-CSPs~\cite{Raghavendra08}.

The success of SoS for optimization confers on it an important role as
an algorithmic tool. For this reason, on the flip side, understanding
the degree range for which SoS \textit{fails} to provide a desired
guarantee to a computational problem can be useful to the algorithm
designer in two ways. Firstly and more concretely, since SoS is a
proof system capturing a broad class of algorithmic
reasoning~\cite{FKP19}, an SoS lower bound can inform the algorithm
designer not only of the minimum degree required within the SoS
hierarchy, but also to avoid methods of proof that are captured by
low-degree SoS reasoning. Secondly, an SoS lower bound can serve as
strong evidence for computational hardness~\cite{hop17,hop18}, even
though it is not a formal guarantee against all types of
algorithms. This hardness evidence is particularly relevant to
average-case problems for which we do not have NP-hardness (see, e.g.,
the SoS lower bound on the Planted Clique problem~\cite{BHKKMP16}).

Our main results concern the performance of SoS on the following basic optimization problem
\begin{equation}\label{eq:general_opt}
  \OPT(W) \defeq \max_{x \in \{\pm 1\}^n} x^\T W x,
\end{equation}
where $W$ is a symmetric matrix in $\RR^{n\times n}$. This problem
arises in the fields of computer science and statistical physics,
though the terminology can sometimes differ. Computer scientists might
regard $x \in \{\pm 1\}^n$ as encoding a bipartition of $[n] = \{1,
2, \ldots, n\}$. Note that by taking $W$ to be a graph
Laplacian~\cite[Section 4]{HooryLW06} the problem is equivalent to the
MaxCut problem, a well-known NP-hard problem in the worst
case~\cite{K72}. A statistical physicist might regard $x$ as encoding
spin values in a spin-glass model. The matrix $-W$ is regarded as
the \textit{Hamiltonian} of the underlying physical system, where
entry $-W_{i,j}$ models the interaction between spin $x_i$ and $x_j$
(with $-W_{i,j} \ge 0$ being ferromagnetic and $-W_{i,j} < 0$ being
anti-ferromagnetic). Then, the optimized $x$ corresponds to the
minimum-energy, or ground, state of the system.

Instead of considering $\OPT(W)$ for a worst-case $W$, one can consider the
average-case problem in which $W$ is sampled according to some
distribution. One of the simplest models of random matrices is the
Gaussian Orthogonal Ensemble, denoted $\GOE(n)$ for $n$-by-$n$
matrices and defined as follows.
\begin{definition}
  The Gaussian Orthogonal Ensemble, denoted $\GOE(n)$, is the distribution of $\frac{1}{\sqrt{2}}(A + A^\T)$
  where $A$ is a random $n\times n$ matrix with i.i.d. standard Gaussian entries.
\end{definition}
Taking $W \sim \GOE(n)$ for the optimization problem $\OPT(W)$
of~\cref{eq:general_opt} gives rise to the so-called
Sherrington--Kirkpatrick (SK) Hamiltonian~\cite{SK76}. Note that $\GOE(n)$ is a particular kind
of Wigner matrix ensemble, thereby satisfying the semicircle law, which
in this case establishes that the largest eigenvalue of $W$ is
$(2+\littleoh_n(1)) \cdot \sqrt{n}$ with probability
$1-\littleoh_n(1)$. Thus, a trivial spectral bound establishes
$\OPT(W) \le (2+\littleoh_n(1)) \cdot n^{3/2}$ with probability
$1-\littleoh_n(1)$. However, in a foundational work based on a
variational argument~\cite{P79}, Parisi conjectured that
$$
\E_{W \sim \GOE(n)}\left[ \OPT(W) \right] \approx 2 \cdot P^* \cdot n^{3/2},
$$
where $P^* \approx 0.7632$ is now referred to as the Parisi constant.
In a breakthrough result, Talagrand~\cite{Tal06} gave a rigorous proof
of Parisi's conjecture~\footnote{The results of Talagrand~\cite{Tal06}
were for the Sherrington--Kirkpatrick and mixed $p$-spin systems with
$p$ even. In~\cite{Panchenko2014}, Panchenko generalized these results
to arbitrary mixed $p$-spin systems (also including odd $p$).}. The
question then became, ``is there a polynomial-time algorithm that
given $W \sim \GOE(n)$ computes an $x$ achieving close to $\OPT(W)$?"
As it turns out, the answer was essentially shown to be yes by
Montanari~\cite{Montanari19}!

The natural question we study is that of certification: ``is there an
efficient algorithm to certify an upper bound on $\OPT(W)$ for any
input $W$, that improves upon the trivial spectral bound?" In
particular, we can ask how well SoS does as a certification
algorithm. The natural upper bound of $(2+\littleoh_n(1)) \cdot
n^{3/2}$ obtained via the spectral norm of $W$ is also the value of
the degree-$2$ SoS relaxation~\cite{MS16}. Two independent recent
works of Mohanty--Raghavendra--Xu~\cite{MRX20} and
Kunisky--Bandeira~\cite{KuniskyBandeira19} show that degree-4 SoS does
not perform much better, and a heuristic argument from~\cite{bkw19} suggests that even degree-$(n/\log n)$ SoS cannot certify anything stronger than the trivial spectral bound. Thus we ask,

\begin{center}
  \emph{Can higher-degree SoS certify better upper bounds for the Sherrington--Kirkpatrick problem, \\
        hopefully closer to the true bound $2 \cdot P^* \cdot n^{3/2}$?}
\end{center}

\noindent \textbf{Our Results.} \enspace We answer the question above negatively by showing that even at degree as large as
$n^\delta$, SoS cannot improve upon the basic spectral
algorithm. More precisely, we have the following theorem which is our
first main result and our most important contribution.

\begin{restatable}{theorem}{SKbounds}[Main I]\label{theo:sk-bounds}
  There exists a constant $\delta > 0$ such that, w.h.p. for $W \sim \GOE(n)$, there is a degree-$n^\delta$ SoS solution
  for the Sherrington--Kirkpatrick problem with value at least $(2-\littleoh_n(1)) \cdot n^{3/2}$.
\end{restatable}

In light of the result of Montanari~\cite{Montanari19}, the situation
is intriguing. Montanari showed that for all $\eps > 0$, there is a
$\bigoh_\eps(n^2)$ time randomized algorithm that given a random $W$
drawn from the Gaussian Orthogonal Ensemble, outputs an $x$ such that
${x^T}Wx \geq (1-\eps)\OPT(W)$.  The correctness of the algorithm
assumes a widely-believed conjecture from statistical physics known as
the full replica symmetry breaking assumption. However, we show an
integrality gap for SoS.

Based on this, it is an interesting question whether SoS, together
with an appropriate rounding scheme, is optimal for the
Sherrington-Kirkpatrick problem. On the one hand, the situation could
be similar to the Feige-Schechtman integrality gap instance for
MaxCut~\cite{fs02}. For the Feige-Schechtman integrality gap instance,
SoS fails to certify the value of the optimal solution. However,
applying hyperplane rounding to the SoS solution gives an
almost-optimal solution for these instances. It could be the case that
there is a rounding scheme which takes an SoS solution for the
Sherrington-Kirkpatrick problem on a random $W$ and returns an almost
optimal solution $x$. On the other hand, we currently don't know what
this rounding scheme would be.

In order to prove~\cref{theo:sk-bounds}, we first introduce a new
average-case problem we call Planted Affine Planes (PAP) for which we directly prove a SoS lower bound. We then use
the PAP lower bound to prove a lower bound on the
Sherrington--Kirkpatrick problem. The PAP problem can be informally
described as follows (see~\cref{def:prob:pap} for the formal definition).
\begin{definition}[Informal statement of PAP]
  Given $m$ random vectors $d_1,\ldots,d_m$ in $\mathbb{R}^n$, can we
  prove that there is no vector $v \in \RR^n$ such that for all
  $u \in [m]$, $\langle v, d_u\rangle^2 = 1$? In other words, can we
  prove that $m$ random vectors are not all contained in two parallel
  hyperplanes at equal distance from the origin?
\end{definition}
This problem, when we restrict $v$ to a Boolean vector in $\set{\pm \frac{1}{\sqrt{n}}}^n$,
can be encoded as the feasibility of the polynomial system
\begin{align*}
\exists v \in \R^n~\text{s.t.} \qquad & \forall i \in [n], \; v_i^2 = \frac{1}{n},\\
& \forall u \in [m], \; \ip{v}{d_u}^2 = 1.
\end{align*}
Hence it is a ripe candidate for SoS. However, we show that SoS fails
to refute a random instance. The Boolean restriction on $v$ actually
makes the lower bound result stronger since SoS cannot refute even a
smaller subset of vectors in $\RR^n$. In this work, we will consider two
different random distributions, namely when $d_1, \ldots, d_m$ are
independent samples from the multivariate normal distribution and when
they are independent samples from the uniform distribution on the
boolean hypercube.
\begin{theorem}[Main II]\label{theo:sos-bounds}
  For both the Gaussian and Boolean settings, there exists a constant $c > 0$ such that for all $\eps > 0$ and $\delta \le c\eps$, for $m \leq n^{3/2 - \eps}$, w.h.p. there is a feasible degree-$n^\delta$ SoS solution for Planted Affine Planes.
\end{theorem}

It turns out that the Planted Affine Plane problem introduced above is
closely related to the following ``Boolean vector in a random
subspace'' problem, which we call the Planted Boolean Vector problem,
introduced by~\cite{MRX20} in the context of studying the performance
of SoS on computing the Sherrington--Kirkpatrick Hamiltonian.

The Planted Boolean Vector problem is to certify that a random subspace of $\R^n$ is far from containing a boolean vector.
Specifically, we want to certify an upper bound for
\[
\OPT(V) \defeq  \frac{1}{n}\max_{b \in \{\pm 1\}^n} b^\T \Pi_V b,
\]
where $V$ is a uniformly random $p$-dimensional subspace\footnote{$V$
can be specified by a basis, which consists of $p$ i.i.d.  samples
from $\calN(0, I)$.} of $\RR^n$, and $\Pi_V$ is the projector onto
$V$. In brief, the relationship to the Planted Affine Plane problem is
that the PAP vector $v$ represents the coefficients on a linear 
combination for the vector $b$ in the span of a basis of 
$V$.

An argument of~\cite{MRX20} shows that, when $p \ll n$, w.h.p.,
$\OPT(V) \approx \frac{2}{\pi}$, whereas they also show that
w.h.p. assuming $p \geq n^{0.99}$, there is a degree-4 SoS solution
with value $1-\littleoh_n(1)$. They ask whether or not there is a
polynomial time algorithm that can certify a tighter bound; we rule
out SoS-based algorithms for a larger regime both in terms of SoS
degree and the dimension $p$ of the random subspace.

\begin{restatable}{theorem}{booleanSubspace}[Main III]\label{theo:boolean-subspace}
	There exists a constant $c > 0$ such that, for all $\eps > 0$ and $\delta \le c\eps$, for $p \geq n^{2/3 + \eps}$, w.h.p. over $V$ there is a
  degree-$n^\delta$ SoS solution for Planted Boolean Vector of value $1$.
\end{restatable}

\noindent \textbf{Our Approach.} \enspace We now provide a brief
high-level description of our approach (see~\cref{sec:strategy} for a
more detailed overview). The bulk of our technical contribution lies
in the SoS lower bound for the Planted Affine Planes
problem,~\cref{theo:sos-bounds}. We then show that Planted Affine
Planes in the Gaussian setting is equivalent to the Planted Boolean
Vector problem. The reduction from Sherrington-Kirkpatrick to the
Planted Boolean Vector problem is due to
Mohanty--Raghavendra--Xu~\cite{MRX20}.

As a starting point to the PAP lower bound, we employ the general
techniques introduced by Barak et al.~\cite{BHKKMP16} for SoS lower
bounds. We use their pseudocalibration machinery to produce a good
candidate SoS solution $\pE$. The operator $\pE$ unfortunately does
not exactly satisfy the PAP constraints ``$\ip{v}{d_u}^2 = 1$'', it
only satisfies them up to a tiny error. We use an interesting and
rather generic approach to round $\pE$ to a nearby pseudoexpectation
operator $\pE'$ which does exactly satisfy the constraints.

For degree $D$, the candidate SoS solution can be viewed as a
(pseudo) moment matrix $\calM$ with rows and columns indexed by
subsets $I,J\subset [n]$ with size bounded by $D/2$ and with entries
\[\calM[I,J] \defeq \pE[v^{I} v^{J}].\]
The matrix $\calM$ is a random function of the inputs $d_1, \dots, d_m$, and the most challenging part of the 
analysis consists of showing that $\calM$ is positive
semi-definite (PSD) with high probability.

Similarly to~\cite{BHKKMP16}, we decompose
$\calM$ as a linear combination of graph matrices, i.e., $\calM~=~\sum_{\alpha}~\lambda_{\alpha}~\cdot~M_{\alpha}$, where $M_{\alpha}$
is the graph matrix associated with shape $\alpha$. In brief, each
graph matrix aggregates all terms with shape $\alpha$  in the Fourier expansions of the entries of $\calM$ -- the shape $\alpha$ is informally a graph with labeled edges
with size bounded by $\poly(D)$. A graph
matrix decomposition of $\calM$ is particularly handy in the PSD
analysis since the operator norm of individual graph matrices $M_{\alpha}$ is (with high probability)
determined by simple combinatorial properties of the graph
$\alpha$. One technical difference from~\cite{BHKKMP16} is that our
graph matrices have two types of vertices $\square{}$ and $\circle{}$; these graph matrices fall into the general framework developed by Ahn et al. in~\cite{AMP20}.

To show that the matrix $\calM$ is PSD, we need to study the graph matrices that appear with nonzero coefficients in the decomposition. The matrix $\calM$ can be split into blocks and each diagonal block contains in the decomposition a (scaled) identity matrix. From the graph matrix perspective, this means that certain ``trivial'' shapes appear in the decomposition, with appropriate coefficients. If we could bound the norms of all other graph matrices that appear against these trivial shapes and show that, together, they have negligible norm compared to the sum of these scaled identity blocks, then we would be in good shape.

Unfortunately, this approach will not work. The kernel of the matrix $\calM$ is nontrivial, as a consequence of satisfying the PAP constraints ``$\ip{v}{d_u}^2 = 1$", and hence there is no hope of showing that the contribution of all nontrivial shapes in the decomposition of $\calM$ has small norm. Indeed, certain shapes $\alpha$ appearing in the 
decomposition of $\calM$ are such that $\norm{\lambda_{\alpha} \cdot M_{\alpha}}$ is large. As it turns out, all such shapes have a simple graphical substructure, and so we call these shapes \textit{spiders}. 

To get around the null space issue, we restrict ourselves to $\nullspace(\calM)^\perp$, which is the complement of the nullspace of $\calM$. 
We show that the substructure present in a spider implies that the spider is close to the zero matrix in $\nullspace(\calM)^\perp$. Because of this, we can almost freely 
add and subtract $M_\alpha$ for spiders $\alpha$ while preserving the action of $\calM$ on $\nullspace(\calM)^\perp$. Our strategy is to ``kill'' the spiders 
by subtracting off $\lambda_\alpha \cdot M_\alpha$ for each spider $\alpha$. But because $M_{\alpha}$ is only approximately in $\nullspace(\calM)^\perp$, this 
strategy could potentially introduce new graph matrix terms, and in particular it could introduce new spiders. To handle this, 
we recursively kill them while carefully analyzing how the coefficients of all the graph matrices change. After all spiders
are killed, the resulting moment matrix becomes
$$
\sum_{0 \le k \le D/2} \frac{1}{n^{k}} \cdot I_k + \sum_{\gamma \colon \textup{non-spiders}} \lambda_{\gamma}' \cdot M_{\gamma},
$$
for some new coefficients $\lambda_{\gamma}'$. Here, $I_k$ is the
matrix which has an identity in the $k$th block and the remaining
entries $0$. Using a novel charging argument, we finally show that the
latter term is negligible compared to the former term, thus
establishing $\calM \succeq 0$.

\noindent \textbf{Summary of Related Work and Our Contributions.} \enspace
We now summarize the existing work on these problems and our
contributions.  Degree-$4$ SoS lower bounds on the
Sherrington-Kirkpatrick Hamiltonian problem were proved independently
by Mohanty--Raghavendra--Xu~\cite{MRX20} and
Kunisky--Bandeira~\cite{KuniskyBandeira19} whereas we prove an
improved degree-$n^{\delta}$ SoS lower bound for some constant $\delta
> 0$.  Our result is obtained by reducing the Sherrington-Kirkpatrick
problem to the ``Boolean Vector in a Random Subspace'' problem which
is equivalent to our new Planted Affine Planes problem on the normal
distribution. The reduction from Sherrington-Kirkpatrick problem to
the ``Boolean Vector in a Random Subspace'' is due to
Mohanty--Raghavendra--Xu~\cite{MRX20}. The results of
Mohanty--Raghavendra--Xu~\cite{MRX20} and
Kunisky--Bandeira~\cite{KuniskyBandeira19} build on a degree-$2$ SoS
lower bounds of Montanari and Sen~\cite{MS16}. Regarding upper bounds,
Montanari~\cite{Montanari19} gave an efficient randomized message
passing algorithm to estimate $\OPT(W)$ in the SK problem within a
$(1-\eps)$ factor under the full replica symmetry breaking assumption.

Degree-$4$ SoS lower bounds on the ``Boolean Vector in a Random
Subspace'' problem for $p~\ge~n^{0.99}$ were proved by
Mohanty--Raghavendra--Xu in~\cite{MRX20} where this problem was
introduced. We improve the dependence on $p$ to $p \ge n^{2/3
+ \epsilon}$ for any $\epsilon > 0$ and obtain a stronger
degree-$n^{c\eps}$ SoS lower bound for some absolute constant $c > 0$.

\section{Technical Preliminaries}\label{sec:prelim}

In this section we record problem statements, then define and discuss the main objects in our SoS
lower bound: pseudoexpectation operators, the moment matrix, and graph matrices.

For a vector or variable
$v \in \R^n$, and $I \subseteq [n]$, we use the notation
$v^I \defeq \prod_{i \in I}v_i$. When a statement holds with high
probability (w.h.p.), it means it holds with probability $1 - o_n(1)$. In
particular, there is no requirement for small $n$.

\subsection{Problem statements}

We introduce the Planted Affine Planes problem over a distribution $\calD$.
\begin{definition}[Planted Affine Planes (PAP) problem]\label{def:prob:pap}
  Given $d_1, \dots, d_m \sim \calD$ where each $d_u$ is a vector in $\RR^n$,
  determine whether there exists $v \in \set{\pm \frac{1}{\sqrt{n}}}^n$ such that
  \[
  \ip{v}{d_u}^2 = 1,
  \]
  for every $u \in [m]$.
\end{definition}
Our results hold for the Gaussian setting $\mathcal{D} = \calN(0, I)$ and the boolean setting where $\calD$ is uniformly sampled from $\{\pm 1\}^n$, though we conjecture (\cref{sec:open-problems}) that similar SoS bounds hold under more general conditions on $\calD$.

Observe that in both settings the solution vector $v$ is restricted to be Boolean (in the sense that the entries are either $\frac{1}{\sqrt{n}}$ or $\frac{-1}{\sqrt{n}}$) and an SoS lower bound for this restricted version of the problem is
stronger than when $v$ can be an arbitrary vector from $\RR^n$.

The Sherrington--Kirkpatrick (SK) problem comes from the spin-glass model
in statistical physics~\cite{SK76}.

\begin{definition}[Sherrington-Kirkpatrick problem]\label{def:prob:sk}
  Given $W \sim \GOE(n)$, compute
  \[
    \OPT(W) \defeq \max_{x \in \{\pm 1\}^n} x^\T W x.
  \]
\end{definition}

The Planted Boolean Vector problem was introduced by
Mohanty--Raghavendra--Xu~\cite{MRX20}, where it was called the
``Boolean Vector in a Random Subspace''.

\begin{definition}[Planted Boolean Vector problem]\label{def:prob:pbv}
  Given a uniformly random $p$-dimensional subspace $V$ of $\mathbb{R}^n$ in the form of
  a projector $\Pi_V$ onto $V$, compute
  \[
  \OPT(V) \defeq  \frac{1}{n}\max_{b \in \{\pm 1\}^n} b^\T \Pi_V b.
  \]
\end{definition}

\subsection{Sum-of-Squares solutions}

We will work with two equivalent definitions of a degree-$D$ SoS
solution: a pseudoexpectation operator and a moment matrix. We tailor these
definitions to our setting of feasibility of systems of polynomial
equality constraints given by the common zero set of a collection of
polynomials $\calP$ on $\pm \frac{1}{\sqrt{n}}$ Boolean variables
$v_1,\dots,v_n$.  For a degree-$D$ solution to be well defined, we
need $D$ to be at least the maximum degree of a polynomial in
$\calP$. Let $\RR^{\le D}(v_1,\dots,v_n)$ be the subset of
polynomials of degree at most $D$ from the polynomial ring
$\RR(v_1,\dots,v_n)$. We denote the degree of a polynomial
$f \in \RR(v_1,\dots,v_n)$ by $\deg(f)$.

\subsubsection{Pseudoexpectation operator}

We formally define the pseudoexpectation operators used in our setting.
\begin{definition}[Pseudoexpectation]\label{def:pseudoexpectation}
  Given a finite collection of ``constraint'' polynomials $\calP$ of degree at most $D$ on $\pm \frac{1}{\sqrt{n}}$ Boolean variables $v_1,\dots,v_n$,
  a degree-$D$ pseudoexpectation operator $\pE$ is an operator
  $\pE \colon \RR^{\le D}(v_1,\dots,v_n) \to \mathbb{R}$ satisfying:
  \begin{enumerate}
    \item $\pE[1] = 1$, \label{pe:normalized}
    \item $\pE$ is an $\RR$-linear operator, i.e., $\pE[f+g] = \pE[f] + \pE[g]$ for every $f,g \in \RR^{\le D}(v_1,\dots,v_n)$, \label{pe:linear}
    \item $\pE[f^2] \ge 0$ for every $f \in \RR^{\le D}(v_1,\dots,v_n)$ with $\deg(f^2) \le D$. \label{pe:psdness}
    \item $\pE[(v_i^2-\frac{1}{n}) \cdot f] = 0$ for all $i \in [n]$ and for every $f \in \RR^{\le D}(v_1,\dots,v_n)$ with $\deg(f) \le D-2$, and \label{pe:boolean}   
    \item $\pE[g \cdot f] = 0$ for every $g \in \calP, f \in \RR^{\le D}(v_1,\dots,v_n)$ with $\deg(f \cdot g) \le D$. \label{pe:feasible}
  \end{enumerate}
\end{definition}
Note that $\pE$ behaves similarly to an expectation operator
restricted to $\RR^{\le D}(v_1,\dots,v_n)$ with the caveat that $\pE$
is only guaranteed to be non-negative on sum-of-squares polynomials.

The degree-$D$ SoS algorithm checks feasibility of a polynomial system by 
checking whether or not a degree-$D$ pseudoexpectation operator exists. To 
show an SoS lower bound, one must construct a pseudoexpectation 
operator.

\subsubsection{Moment matrix}\label{sec:moment-mtx}

We define
the moment matrix associated with a degree-$D$
pseudoexpectation $\pE$.
\begin{definition}[Moment Matrix of $\pE$]
  The moment matrix $\calM=\calM(\pE)$ associated to a pseudoexpectation $\pE$ is a
  $\binom{[n]}{\leq D/2} \times \binom{[n]}{\leq D/2}$ matrix with rows and columns indexed
  by subsets of $I, J \subseteq [n]$ of size at most $D/2$ and defined as
  \[
  \calM[I, J] \defeq \pE\left[ v^I \cdot v^J \right].
  \]
\end{definition}

To show that a candidate pseudoexpectation satisfies~\cref{pe:psdness} in~\cref{def:pseudoexpectation}, we will rely on the following standard fact.
\begin{fact}
  In the definition of pseudoexpectation, \cref{def:pseudoexpectation}, the condition in \cref{pe:psdness} is equivalent to $\calM \succeq 0$.
\end{fact}

\subsection{Graph matrices}
To study $\calM$, we decompose it using the framework of \textit{graph matrices}. Originally developed in the context of the planted clique problem, graph matrices are random matrices whose entries are symmetric functions of an underlying random object -- in our case, the set of vectors $d_1, \dots, d_m$. We take the general presentation and results from~\cite{AMP20}. For our purposes, the following definitions are sufficient.

The graphs that we study have two types of vertices, circles $\circle{}$ and squares $\square{}$. We let $\calC_m$ be a set of $m$ circles labeled 1 through $m$, which we denote by $\circle{1}, \circle{2}, \dots, \circle{m}$, and let $\calS_n$ be a set of $n$ squares labeled 1 through $n$, which we denote by $\square{1}, \square{2}, \dots, \square{n}$. We will work with bipartite graphs with edges between circles and squares, which have positive integer labels on the edges. When there are no multiedges (the graph is simple), such graphs are in one-to-one correspondence with Fourier characters on the vectors $d_u$. An edge between $\circle{u}$ and $\square{i}$ with label $l$ represents $h_{l}(d_{u,i})$ where $\{h_k\}$ is the Fourier basis (e.g. Hermite polynomials).

\[ \text{simple graph with labeled edges} \qquad \Longleftrightarrow \qquad \displaystyle\prod_{\substack{\circle{u} \in \calC_m,\\ \square{i} \in \calS_n}} h_{l(\circle{u}, \square{i})}(d_{u,i}) \]

An example of a Fourier polynomial as a graph with labeled edges is given in~\cref{fig:fourier_graph}. Unlabeled edges are implicitly labeled 1.
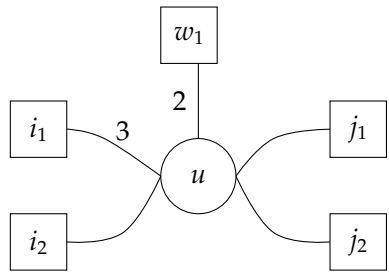
\begin{figure}[h!]
\centering      
\begin{tikzpicture}[scale=0.5,every node/.style={scale=0.5}]
\draw  (-2,3) rectangle node {\huge $i_1$}(-0.5,1.5) node (v5) {};
\draw  (3,1) ellipse (1 and 1) node {\huge $u$};
\draw  (6.5,3) rectangle node (v10) {\huge $j_1$} (8,1.5);
\draw  (-2,0) rectangle node {\huge $i_2$} (-0.5,-1.5);
\draw  (6.5,0) rectangle node {\huge $j_2$} (8,-1.5);
\node (v1) at (-0.5,3) {};
\node (v4) at (-0.5,2.25) {};
\node (v6) at (6.5,2.25) {};
\node (v8) at (-0.5,-0.75) {};
\node (v9) at (6.5,-0.75) {};
\node at (-0.5,0) {};
\node (v2) at (2,1) {};
\node at (2,1.5) {};
\node (v7) at (6.5,1.5) {};
\node (v3) at (4,1) {};
%\draw  plot[smooth, tension=.7] coordinates {(v1) (1,2.5) (v2)};
\draw  plot[smooth, tension=.7] coordinates {(v3)};
\draw  plot[smooth, tension=.7] coordinates {(v3)};
\draw  plot[smooth, tension=.7] coordinates {(v2) (0.5,2) (v4)};
\node at (1,2.2) {\huge $3$};
%\draw  plot[smooth, tension=.7] coordinates {(v2) (0.5,1) (v5)};
\draw  plot[smooth, tension=.7] coordinates {(v3) (5,2) (v6)};
\draw  plot[smooth, tension=.7] coordinates {(v3)};
\draw  plot[smooth, tension=.7] coordinates {(v2) (1,-0.5) (v8)};
\draw  plot[smooth, tension=.7] coordinates {(v3) (5,-0.5) (v9)};
\node at (6.5,3) {};
\draw  plot[smooth, tension=.7] coordinates {(v3)};
\draw  plot[smooth, tension=.7] coordinates {(v3)};
\draw  plot[smooth, tension=.7] coordinates {(v10)};
\node at (6.5,3) {};
\draw  (2,5.5) rectangle node {\huge $w_1$} (3.5,4);
\node (v11) at (2.5,4) {};
\node (v13) at (3,4) {};
\node (v12) at (3,2) {};
%\draw  plot[smooth, tension=.7] coordinates {(v12) (2.5,3) (2.5,4)};
\draw  plot[smooth, tension=.7] coordinates {(v12) (3,3) (3,4)};
\node at (2.5,3) {\huge $2$};
\end{tikzpicture}
\caption{The Fourier polynomial $h_3(d_{u,i_1})h_1(d_{u,i_2})h_2(d_{u,w_1})h_1(d_{u,j_1})h_1(d_{u,j_2})$ represented as a graph.}
\label{fig:fourier_graph}
\end{figure}

Define the degree of a vertex $v$,  denoted $\deg(v)$, to be the sum of the labels incident to $v$, and $\abs{E}$ to be the sum of all labels. For 
intuition it is mostly enough to work with simple graphs, in which case these quantities make sense as the edge multiplicities in an implicit multigraph.

\begin{definition}[Proper]
We say an edge-labeled graph is \textit{proper} if it has no multiedges.
\end{definition}
The definitions allow for ``improper'' edge-labeled multigraphs which simplify multiplying graph matrices (\cref{sec:single-spider} and \cref{sec:exact-constraints}).

\begin{definition}[Matrix indices]
A \textit{matrix index} is a set $A$ of elements from $\calC_m \cup \calS_n$.
\end{definition}
We let $A(\square{i})$ or $A(\circle{u})$ be 0 or 1 to indicate if the vertex is in $A$.

\begin{definition}[Ribbons]\label{def:ribbon}
A \textit{ribbon} is an undirected, edge-labeled graph $R = (V(R), E(R), A_R, B_R)$, where $V(R) \subseteq \calC_m\cup \calS_n$ and $A_R, B_R$ are two matrix indices (possibly not disjoint) with $A_R, B_R \subseteq V(R)$, representing two distinguished sets of vertices. Furthermore, all edges in $E(R)$ go between squares and circles.
\end{definition}
We think of $A_R$ and $B_R$ as being the ``left'' and ``right'' sides of $R$, respectively. We also define the set of ``middle vertices'' $C_R \defeq V(R) \setminus (A_R \cup B_R)$. If $e \not\in E(R)$, then we define its label $l(e) = 0$. We also abuse notation and write $l(\square{i}, \circle{u})$ instead of $l(\{\square{i}, \circle{u}\})$.

Akin to the picture above, each ribbon corresponds to a Fourier polynomial.
This Fourier polynomial lives inside a single entry of the matrix $M_R$.
In the definition below, the $h_k(x)$ are the Fourier basis corresponding to the respective setting. In the Gaussian case, they are the (unnormalized) Hermite polynomials, and in the boolean case, they are just the parity function, represented by
\[h_0(x) = 1, \qquad h_1(x) = x, \qquad h_k(x) = 0 \;\; (k \geq 2) \]

\begin{definition}[Matrix for a ribbon]\label{def:ribbon-matrix}
The matrix $M_R$ has rows and columns indexed by subsets of $\calC_m~\cup~\calS_n$, with a single nonzero entry defined by
\[M_R[I, J] = \left\{\begin{array}{lr}
    \displaystyle\prod_{\substack{e \in E(R), \\ e = \{\square{i}, \circle{u}\}}} h_{l(e)}(d_{u,i}) &  I = A_R, J = B_R\\
    0 & \text{Otherwise}
\end{array}\right. \]
\end{definition}

Next we describe the shape of a ribbon, which is essentially the ribbon when we have forgotten all the vertex labels and retained only the graph structure and the distinguished sets of vertices.
\begin{definition}[Index shapes]
An \textit{index shape} is a set $U$ of formal variables. Furthermore, each variable is labeled as either a ``circle'' or a ``square''.
\end{definition}
We let $U(\square{i})$ and $U(\circle{u})$ be either 0 or 1 for whether $\square{i}$ or $\circle{u}$, respectively, is in $U$.

\begin{definition}[Shapes]\label{def:shape}
A \textit{shape} is an undirected, edge-labeled graph $\alpha = (V(\alpha), E(\alpha), U_\alpha, V_\alpha)$ where $V(\alpha)$ is a set of formal variables, each of which is labeled as either a ``circle'' or a ``square''. $U_\alpha$ and $V_\alpha$ are index shapes (possibly with variables in common) such that $U_\alpha, V_\alpha \subseteq V(\alpha)$. The edge set $E(\alpha)$ must only contain edges between the circle variables and the square variables.
\end{definition}

We'll also use $W_\alpha \defeq V(\alpha) \setminus (U_\alpha \cup V_\alpha)$ to denote the ``middle vertices'' of the shape.

\begin{remk}
	We will abuse notation and use $\square{i}, \square{j}, \circle{u}, \circle{v}, \ldots$ for both the vertices of ribbons and the vertices of shapes. If they are ribbon vertices, then the vertices are elements of $\calC_m\cup\calS_n$ and if they are shape vertices, then they correspond to formal variables with the appropriate type.
\end{remk}

\begin{definition}[Trivial shape]
	Define a shape $\alpha$ to be trivial if $U_\alpha = V_\alpha$, $W_\alpha = \emptyset$ and $E(\alpha) = \emptyset$.
\end{definition}

\begin{definition}[Transpose of a shape]
  The transpose of a shape  $\alpha = (V(\alpha), E(\alpha), U_\alpha, V_\alpha)$ is defined
  to be the shape $\alpha^{\T} = (V(\alpha), E(\alpha), V_\alpha, U_\alpha)$.
\end{definition}

For a shape $\alpha$ and an injective map $\sigma : 
V(\alpha) \to \calC_m \cup \calS_n$, we define the 
realization $\sigma(\alpha)$ as a ribbon in the natural
way, by labeling all the variables using the map 
$\sigma$. We also require $\sigma$ to be 
type-preserving i.e. it takes square variables to $\calS_n$ and circle variables to $\calC_m$. 
The ribbons that result are referred to as \textit{ribbons of shape $\alpha$}; notice that this partitions the set of all ribbons according to their shape\footnote{Partitions up to equality of shapes, where two shapes are equal if there is a type-preserving bijection between their variables that converts one shape to the other. When we operate on sets of shapes below, we implicitly use each distinct shape only once.}\footnote{Note that in our definition two realizations of a shape may give the same ribbon.}.

Finally, given a shape $\alpha$, the graph matrix $M_\alpha$ consists of all Fourier characters for ribbons of shape $\alpha$.
\begin{definition}[Graph matrices]\label{def:graph-matrix}
Given a shape $\alpha = (V(\alpha), E(\alpha), U_\alpha, V_\alpha)$, the \textit{graph matrix} $M_\alpha$ is
\[M_\alpha = \displaystyle\sum_{R \text{ is a ribbon of shape }\alpha} M_R\]
\end{definition}

The moment matrix for PAP will turn out to be defined using graph matrices $M_\alpha$ whose left and right sides only have square vertices, and no circles. However, in the course of the analysis we will factor and multiply graph matrices with circle vertices in the left or right.

\subsection{Norm bounds}
The spectral norm of a graph matrix is determined, up to logarithmic factors, by relatively simple combinatorial properties of the graph. For a subset $S \subseteq \calC_m \cup \calS_n$, we define the weight $w(S)~\defeq~(\#\text{ circles in }S)\cdot \log_n(m)+ (\#\text{ squares in }S)$. Observe that $n^{w(S)} = m^{\# \text{ circles in }S}\cdot n^{\#\text{ squares in }S}$.

\begin{definition}[Minimum vertex separator]
For a shape $\alpha$, a set $S_{\min}$ is a minimum vertex separator if all paths from $U_\alpha$ to $V_\alpha$ pass through $S_{\min}$ and $w(S_{\min})$ is minimized over all such separating sets.
\end{definition}

Let $W_{iso}$ denote the set of isolated vertices in $W_\alpha$. Then essentially the following norm bound holds for all shapes $\alpha$ with high probability (a formal statement can be found in~\cref{app:norm_bounds}):
\[\norm{M_\alpha} \leq  \widetilde\bigoh\left(n^{\frac{w(V(\alpha)) - w(S_{\min}) + w(W_{iso})}{2}}\right)\]

In fact, the only probabilistic property required of the inputs $d_1, \dots, d_m$ by our proof is that the above norm bounds hold for all shapes that arise in the analysis.
We henceforth assume that the norm bounds in~\cref{lem:gaussian-norm-bounds} (for the Gaussian case) and~\cref{lem:norm-bounds} (for the boolean case) hold.

\section{Proof Strategy}\label{sec:strategy}

Here we explain in more detail the ideas for the Planted Affine Planes
lower bound. Towards the proof of~\cref{theo:sos-bounds}, fix a
constant $\eps > 0$ and a random instance $d_1, \dots, d_m$ with
$n \leq m \leq n^{3/2-\eps}$. We will construct a pseudoexpectation operator
and show that it is PSD up to degree $D = 2\cdot n^\delta$
with high probability.

We start by 
pseudocalibrating to obtain a pseudoexpectation operator $\pE$. The
operator $\pE$ will exactly satisfy the ``booleanity" constraints
``$v_i^2 = \frac{1}{n}$" though it may not exactly satisfy the
constraints ``$\ip{v}{d_u}^2 = 1$" due to truncation error in the
pseudocalibration. Taking the truncation parameter $n^{\tau}$ to be larger than the degree $D$ of the SoS solution, i.e., $\delta \ll \tau$, the truncation error is small enough that we can
round $\pE$ to a nearby $\pE'$ that exactly satisfies the
constraints. This is formally accomplished by viewing
$\pE \in \RR^{\binom{[n]}{\leq D}}$ as a vector and expressing the
constraints as a matrix $Q$ such that $\pE$ satisfies the constraints
iff it lies in the null space of $Q$. The choice of $\pE'$ is then the
projection of $\pE$ to $\nullspace(Q)$. The end result is that we
construct a moment matrix $\calM_{fix} = \calM + \calE$ that exactly
satisfies the constraints such that $\norm{\calE}$ is tiny. This step is done in~\cref{sec:exact-constraints}.

After performing pseudocalibration, in both settings, we will have
essentially the graph matrix decomposition
\[
\calM = \sum_{\text{shapes }\alpha} \lambda_\alpha M_\alpha = \displaystyle\sum_{\substack{\text{shapes }\alpha:\\ \deg(\square{i}) + U(\square{i}) + V(\square{i}) \text{ even},\\ \deg(\circle{u})\text{ even}}} \frac{1}{n^{\frac{\abs{U_\alpha} + \abs{V_\alpha}}{2}}}\cdot \left(\prod_{\circle{u} \in V(\alpha)} h_{\deg(\circle{u})}(1)\right) \cdot \frac{M_\alpha}{n^{\abs{E(\alpha)}/2}}
\]
Here $h_k(1)$ is in both settings the $k$-th Hermite polynomial, evaluated on 1.

In this decomposition of $\calM$, the trivial shapes will be the
dominant terms which we will use to bound the other terms. Recall that
a shape $\alpha = (V(\alpha), E(\alpha), U_\alpha, V_\alpha)$ is
trivial if $U_\alpha = V_\alpha$, $W_\alpha = \emptyset$ and
$E(\alpha) = \emptyset$. These shapes contribute scaled identity
matrices on different blocks of the main diagonal of $\calM$, with
trivial shape $\alpha$ contributing an identity matrix with
coefficient $n^{-\abs{U_\alpha}}$. Two trivial shapes are illustrated
in~\cref{fig:trivial_shapes}.

\begin{figure}[h!]
  \centering      
  \begin{tikzpicture}[scale=0.5,every node/.style={scale=0.5}]
   \draw  (-1,1) rectangle node {\huge $u_1$} (0.5,-0.5);
   \draw  (-1.5,1.5) rectangle (1,-1);
   \node at (-0.25,-1.75) {\huge $U_{\alpha} \cap V_{\alpha}$};
   \node at (-2.5,0) {\huge $\frac{1}{n}$};
   \draw  (4.5,2.5) rectangle (7,-2);
   \draw  (5,2) rectangle node {\huge $u_1$} (6.5,0.5);
   \draw  (5,0) rectangle node {\huge $u_2$} (6.5,-1.5);
   \node at (3.5,0) {\huge $\frac{1}{n^2}$};
   \node at (5.75,-2.75) {\huge $U_{\alpha} \cap V_{\alpha}$};
  \end{tikzpicture}
  \caption{Two examples of trivial shapes.}
  \label{fig:trivial_shapes}
\end{figure}
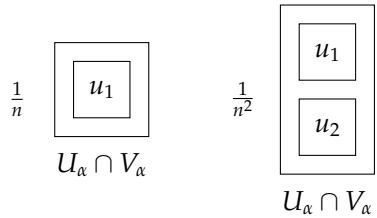

Let $\calM_{\text{triv}}$ be this diagonal matrix of trivial shapes in
the above decomposition of $\calM$. To prove that $\calM \psdgeq 0$,
we attempt the simple strategy of showing that the norm of all other
terms can be ``charged'' against this diagonal matrix
$\calM_{\text{triv}}$. For several shapes this strategy is indeed
viable. To illustrate, let's consider one such shape $\alpha$ depicted
in~\cref{fig:non_spider}.

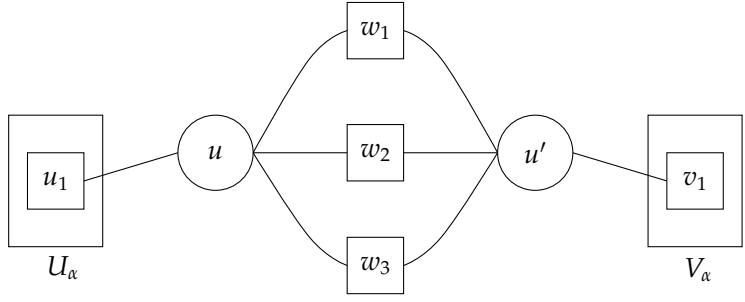
\begin{figure}[h!]
\centering      
\begin{tikzpicture}[scale=0.5,every node/.style={scale=0.5}]
  \draw  (-6.5,2) rectangle node {\huge $u_1$} (-5,0.5);
  \draw  (2,6) rectangle node {\huge $w_1$} (3.5,4.5);
  \draw  (10.5,2) rectangle node {\huge $v_1$} (12,0.5);
  \draw  (2,-0.25) rectangle node {\huge $w_3$} (3.5,-1.75);
  \draw  (2,1.25) rectangle node {\huge $w_2$} (3.5,2.75);
  \draw  (-1.5,2) ellipse (1 and 1) node {\huge $u$};
  \draw  (7,2) ellipse (1 and 1) node {\huge $u'$};
  \node (v1) at (2,2) {};
  \node (v5) at (2,5.25) {};
  \node (v3) at (3.5,5.25) {};
  \node (v6) at (2,-1) {};
  \node (v2) at (-0.5,2) {};
  \node (v4) at (-0.5,2) {};
  \draw  plot[smooth, tension=.7] coordinates {(v1)};
  \draw  plot[smooth, tension=.7] coordinates {(v1) (1,2) (v2)};
  \draw  plot[smooth, tension=.7] coordinates {(v4)};
  \draw  plot[smooth, tension=.7] coordinates {(v4) (1,4.5) (v5)};
  \draw  (-7,3) rectangle (-4.5,-0.5);
  \node at (-5.5549,-1.061) {\huge $U_{\alpha}$};
  \draw  (10,3) rectangle (12.5,-0.5);
  \node at (11.3364,-1.1181) {\huge $V_{\alpha}$};
  \node at (2,5.25) {};
  \node at (2,2) {};
  \node (v11) at (3.5,2) {};
  \node (v13) at (3.5,-1) {};
  \node (v8) at (-2.5,2) {};
  \node (v12) at (6,2) {};
  \node (v10) at (8,2) {};
  \node (v7) at (-5,1.25) {};
  \node (v9) at (10.5,1.25) {};
  \draw (v7);
  \draw  plot[smooth, tension=.7] coordinates {(v2)};
  \draw  plot[smooth, tension=.7] coordinates {(-0.5,2) (1,-0.25) (v6)};
  \draw  plot[smooth, tension=.7] coordinates {(v3) (4.5,4.5) (6,2)};
  \draw  plot[smooth, tension=.7] coordinates {(v10) (v9)};
  \draw  plot[smooth, tension=.7] coordinates {(v12) (4.5,2) (v11)};
  \draw  plot[smooth, tension=.7] coordinates {(v12) (4.5,-0.25) (v13)};
  \draw  plot[smooth, tension=.7] coordinates {(v7) (v8)};
  \end{tikzpicture}
  \caption{Picture of basic non-spider shape $\alpha$.}
  \label{fig:non_spider}
\end{figure}

This graph matrix has $\abs{\lambda_\alpha}
= \Theta(\frac{1}{n^5})$. Using the graph matrix norm bounds, with
high probability the norm of this graph matrix is $\tilde{O}({n^2}m)$:
there are four square vertices and two circle vertices which are not
in the minimum vertex separator. Thus, for this shape $\alpha$, with
high probability $\abs{\lambda_{\alpha}}\norm{M_{\alpha}}$ is
$\tilde{O}\left(\frac{m}{n^3}\right)$ and thus
$\lambda_{\alpha}M_{\alpha} \preceq \frac{1}{n}Id$ (which is the
multiple of the identity appearing in the corresponding block of
$\calM_{\text{triv}}$).

Unfortunately, as pointed out in the introduction, some shapes
$\alpha$ that appear in the decomposition have $\norm{\lambda_\alpha
M_\alpha}$ too large to be charged against
$\calM_{\textup{triv}}$. These are shapes with a certain substructure
(actually the same structure that appears in the matrix $Q$ used to
project the pseudoexpectation operator!) whose norms cannot be handled
by the preceding argument, and which we denote \textit{spiders}.  The
following graph depicts one such \textit{spider} shape (and also
motivates this terminology):
\begin{figure}[h!]
  \centering      
  \begin{tikzpicture}[scale=0.5,every node/.style={scale=0.5}]
    \draw  (-4,4) rectangle node {\huge $u_1$} (-2.5,2.5);
    \draw  (-4,-0.5) rectangle node {\huge $u_2$} (-2.5,-2);
    \draw  (5.5,4) rectangle node {\huge $v_1$} (7,2.5);
    \draw  (5.5,-0.5) rectangle node {\huge $v_2$} (7,-2);
   \draw  (1.5,1) ellipse (1 and 1) node {\huge $u$};
   \node (v1) at (-2.5,3.25) {};
   \node (v5) at (5.5,3.25) {};
  \node (v3) at (-2.5,-1.25) {};
  \node (v6) at (5.5,-1.25) {};
  \node (v2) at (0.5,1) {};
  \node (v4) at (2.5,1) {};
  \draw  plot[smooth, tension=.7] coordinates {(v1)};
  \draw  plot[smooth, tension=.7] coordinates {(v1) (-0.5,3) (v2)};
  \draw  plot[smooth, tension=.7] coordinates {(v3) (-0.5,-1) (v2)};
  \draw  plot[smooth, tension=.7] coordinates {(v4)};
  \draw  plot[smooth, tension=.7] coordinates {(v4) (3.5,3) (v5)};
  \draw  plot[smooth, tension=.7] coordinates {(v4) (3.5,-1) (v6)};
  \draw  (-4.5,5) rectangle (-2,-3.3727);
  \node at (-3.0549,-3.9337) {\huge $U_{\alpha}$};
  \draw  (4.8834,5.3694) rectangle (7.5016,-3.1943);
  \node at (6.338,-3.8124) {\huge $V_{\alpha}$};
\end{tikzpicture}
  \caption{Picture of basic spider shape $\alpha$.}
\end{figure}
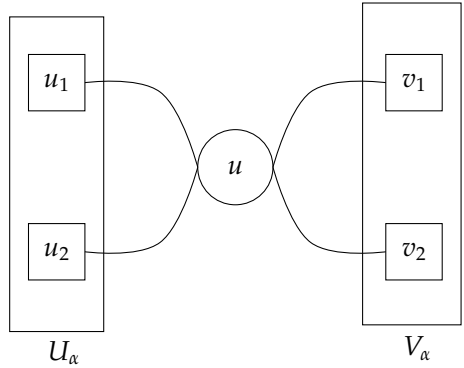

The norm $\norm{\lambda_\alpha M_\alpha}$ of this graph is
$\widetilde\bigomega(\frac{1}{n^{2}})$, as can be easily estimated through the
norm bounds (the coefficient is $\lambda_\alpha = \frac{-2}{n^4}$, the
minimum vertex separator is $\circle{u}$, and there are no isolated
vertices). This is too large to bound against $\frac{1}{n^2}Id$, which is the coefficient of $M_\text{triv}$ on this spider's block.

To skirt this and other spiders, we restrict ourselves to
vectors $x \perp \nullspace(M)$, and observe that this spider $\alpha$ satisfies $x^\T M_{\alpha} \approx 0$. To be more precise, consider the following argument. Consider the two shapes in~\cref{fig:betas}, $\beta_1$ and $\beta_2$ (take note of the label 2 on the edge in $\beta_2$).

\begin{figure}[h!]
\centering      
\begin{tikzpicture}[scale=0.5,every node/.style={scale=0.5}]
\draw  (-4,4) rectangle node {\huge $u_1$} (-2.5,2.5);
\draw  (-4,-0.5) rectangle node {\huge $u_2$} (-2.5,-2);
\draw  (1.5,1) ellipse (1 and 1) node {\huge $u$};   
\draw  (12,1) ellipse (1 and 1) node {\huge $u$};      
\node (v1) at (-2.5,3.25) {};
\node (v3) at (-2.5,-1.25) {};
\node (v2) at (0.5,1) {};
\node (v4) at (11,1) {};
\draw  plot[smooth, tension=.7] coordinates {(v1)};
\draw  plot[smooth, tension=.7] coordinates {(v1) (-0.5,3) (v2)};
\draw  plot[smooth, tension=.7] coordinates {(v3) (-0.5,-1) (v2)};
\draw  plot[smooth, tension=.7] coordinates {(v4)};
\draw  (-4.5,5) rectangle (-2,-3.3727);
\node at (-3.0549,-3.9337) {\huge $U_{\beta_1}$};
\draw  (0,3.5) rectangle (3,-1.5);
\node at (1.5,-2.5) {\huge $V_{\beta_1}$};
\draw  (10.5,3) rectangle  (13.5,-1);
\node at (12,-1.5) {\huge $V_{\beta_2}$};
\draw  (6,2.5) rectangle (7.5,-1.5);
\draw  (9.5,1.5) rectangle node {\huge $w$} (8,0);
\draw  plot[smooth, tension=.7] coordinates {(v4)};
\draw  plot[smooth, tension=.7] coordinates {(v4)};
\node (v5) at (9.5,1) {};
\draw  plot[smooth, tension=.7] coordinates {(v4) (v5)};
\node at (7,-2.1) {\huge $U_{\beta_2} = \emptyset$};
\node at (10.2,1.3) {\Large $2$};
\end{tikzpicture}
\caption{Picture of shapes $\beta_1$ and $\beta_2$.}
\label{fig:betas}
\end{figure}
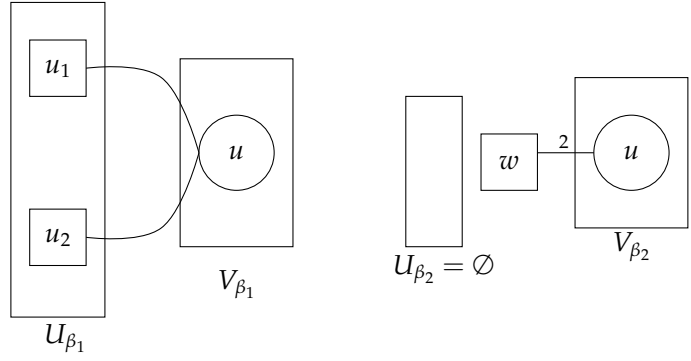

We claim that every column of the matrix $2M_{\beta_1} + \frac{1}{n}M_{\beta_2}$
is in the null space of $\calM$. There are $m$ nonzero columns indexed
by assignments to $V$, which can be a single circle
$\circle{1}, \circle{2}, \dots, \circle{m}$. The nonzero rows are
$\emptyset$ in $\beta_2$ and $\{\square{i}, \square{j}\}$ for $i \neq j$ in $\beta_1$. Fixing $I \subseteq [n]$, entry
$(I, \circle{u})$ of the product matrix $\calM(2M_{\beta_1} + \frac{1}{n}M_{\beta_2})$ is
\begin{align*}
2& \displaystyle\sum_{i < j}\pE [v^I v_i v_j] \cdot d_{ui} d_{uj} + \frac{1}{n}\pE[v^I] \cdot \sum_{i}(d_{ui}^2 - 1)\\
&= 2\displaystyle\sum_{i < j}\pE [v^I v_i v_j] \cdot d_{ui} d_{uj} + \pE[v^Iv_i^2] \cdot \sum_{i}d_{ui}^2 - \pE[v^I] && (\pE \text{ satisfies ``}v_i^2 = \frac{1}{n}")\\
&= \sum_{i,j} \pE[v^I v_i v_j] d_{ui}d_{uj} - \pE[v^I] \\
&= \pE[v^I(\ip{v}{d_u}^2 - 1)]\\
&= 0 && (\pE \text{ satisfies ``}\ip{v}{d_u}^2 = 1")
\end{align*}
In words, the constraint ``$\ip{v}{d_u}^2 = 1 $'' creates a shape
$2\beta_1 + \frac{1}{n}\beta_2$ that lies in the null space of the moment
matrix. On the other hand, we can approximately factor the spider
$\alpha$ across its central vertex, and when we do so, the shape
$\beta_1$ appears on the left side.
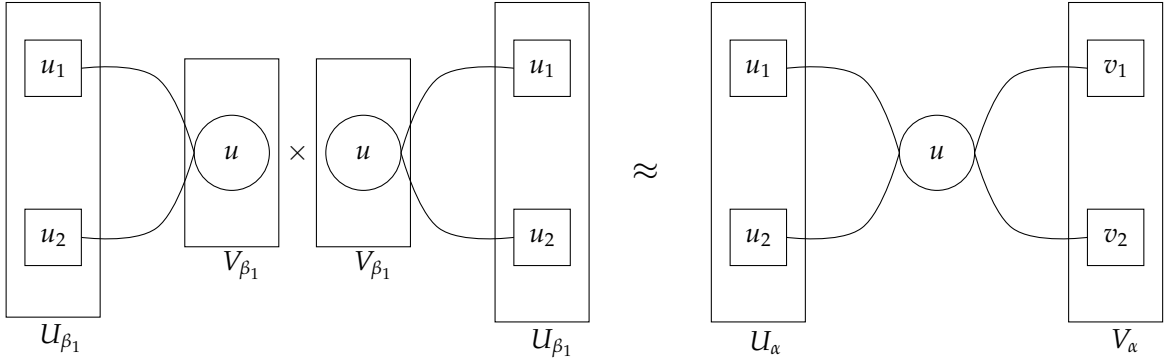
\begin{figure}[h!]
\centering      
\begin{tikzpicture}[scale=0.5,every node/.style={scale=0.5}]
\draw  (4.25,-5.5) rectangle node {\huge $u_1$} (5.75,-7);
\draw  (4.25,-10) rectangle node {\huge $u_2$} (5.75,-11.5);
\draw  (9.75,-8.5) ellipse (1 and 1) node {\huge $u$};
\draw  (13.25,-8.5) ellipse (1 and 1) node {\huge $u$};
\node (v1) at (5.75,-6.25) {};
\node (v3) at (5.75,-10.75) {};
\node (v2) at (8.75,-8.5) {};
\node (v4) at (12.25,-8.5) {};
\draw  plot[smooth, tension=.7] coordinates {(v1)};
\draw  plot[smooth, tension=.7] coordinates {(v1) (7.75,-6.5) (v2)};
\draw  plot[smooth, tension=.7] coordinates {(v3) (7.75,-10.5) (v2)};
\draw  plot[smooth, tension=.7] coordinates {(v4)};
\draw  (3.75,-4.5) rectangle (6.25,-12.8727);
\node at (5.1951,-13.4337) {\huge $U_{\beta_1}$};
\draw  (16.75,-4.5) rectangle (19.25,-13);
\draw  (17.25,-5.5) rectangle node {\huge $u_1$} (18.75,-7);
\draw  (17.25,-10) rectangle node {\huge $u_2$} (18.75,-11.5);
\node (v5) at (14.25,-8.5) {};
\node (v6) at (17.25,-6.25) {};
\node (v7) at (17.25,-10.75) {};
\draw  plot[smooth, tension=.7] coordinates {(v5)};
\draw  plot[smooth, tension=.7] coordinates {(v5) (15.25,-6.5) (v6)};
\draw  plot[smooth, tension=.7] coordinates {(v5) (15.25,-10.5) (v7)};
\node at (18.257,-13.5621) {\huge $U_{\beta_1}$};
\node at (11.5,-8.5) {\huge $\times$};
\node at (20.75,-9) {\Huge $\approx$};
\draw  (23,-5.5) rectangle node {\huge $u_1$} (24.5,-7);
\draw  (23,-10) rectangle node {\huge $u_2$} (24.5,-11.5);
\draw  (32.5,-5.5) rectangle node {\huge $v_1$} (34,-7);
\draw  (32.5,-10) rectangle node {\huge $v_2$} (34,-11.5);
\draw  (28.5,-8.5) ellipse (1 and 1) node {\huge $u$};
\node (v1) at (24.5,-6.25) {};
\node (v5) at (32.5,-6.25) {};
\node (v3) at (24.5,-10.75) {};
\node (v6) at (32.5,-10.75) {};
\node (v2) at (27.5,-8.5) {};
\node (v4) at (29.5,-8.5) {};
\draw  plot[smooth, tension=.7] coordinates {(v1)};
\draw  plot[smooth, tension=.7] coordinates {(v1) (26.5,-6.5) (v2)};
\draw  plot[smooth, tension=.7] coordinates {(v3) (26.5,-10.5) (v2)};
\draw  plot[smooth, tension=.7] coordinates {(v4)};
\draw  plot[smooth, tension=.7] coordinates {(v4) (30.5,-6.5) (v5)};
\draw  plot[smooth, tension=.7] coordinates {(v4) (30.5,-10.5) (v6)};
\draw  (22.5,-4.5) rectangle (25,-13);
\node at (23.9451,-13.561) {\huge $U_{\alpha}$};
\draw  (32,-4.5) rectangle (34.5,-13);
\node at (33.5,-13.5) {\huge $V_{\alpha}$};
\draw  (8.5,-6) rectangle (11,-11);
\node at (10,-11.5) {\huge $V_{\beta_1}$};
\draw  (12,-6) rectangle (14.5,-11);
\node at (13.5,-11.5) {\huge $V_{\beta_1}$};
\end{tikzpicture}
\caption{Approximation $\beta_1 \times \beta_1^\T \approx \alpha$.}
\end{figure}

Therefore $M_\alpha \approx M_{\beta_1} M_{\beta_1}^\T \approx
(M_{\beta_1} + \frac{1}{2n}M_{\beta_2}) M_{\beta_1}^\T$. The columns of
the matrix $M_{\beta_1} + \frac{1}{2n}M_{\beta_2}$ are in the null
space of $\calM$, so for $x \perp \nullspace(\calM)$ we have $x^\T
M_\alpha \approx 0$.

More formally, we are able to find coefficients $c_\beta$ so that all
columns of the matrix
\[A = M_\alpha + \displaystyle\sum_{\beta} c_\beta M_\beta \]
are in $\nullspace(\calM)$. We then observe the following fact:
\begin{fact}\label{fact:null-space}
	If $x \perp \nullspace(\calM)$ and $\calM A = 0$, then $x^\T(AB + \calM)x = x^\T(B^\T A^\T + \calM)x= x^\T \calM x$.
\end{fact}
Using the fact, we can freely add multiples of $A$ to $\calM$ without
changing the action of $\calM$ on $\nullspace(\calM)^\perp$. A
judicious choice is to subtract $\lambda_\alpha A$ which will ``kill''
the spider from $\calM$. Doing this for all spiders, we produce a
matrix whose action is equivalent on $\nullspace(\calM)^\perp$, and
which has high minimum eigenvalue by virtue of the fact that it has no
spiders, showing that $\calM$ is PSD.

The catch is two-fold: first, the coefficients $c_\beta$ may
contribute to the coefficients on the non-spiders; second, the further
intersection terms $M_\beta$ may themselves be spiders ( though they
will always have fewer square vertices than $\alpha$). Thus we must
recursively kill these spiders, until there are no spiders remaining
in the decomposition of $\calM$. The resulting matrix has some new
coefficients on the non-spiders
\[ \calM' = \displaystyle\sum_{\text{non-spiders }\beta} \lambda_\beta' M_\beta. \]
We must bound the accumulation on the coefficients
$\lambda_\beta'$. We do this by considering the \textit{web} of
spiders and non-spiders created by each spider and using bounds on the $c_\beta$
and $\lambda_\alpha$ to argue that the contributions do not blow up, via an interesting charging scheme that exploits the structure of these graphs.

\section{Pseudocalibration}\label{sec:pseudo_calib}

To be able to apply the pseudocalibration technique of~\cite{BHKKMP16}
to an average-case feasibility problem, in our case the PAP
problem, one needs to design a planted distribution supported on
feasible instances. This is done
in~\cref{subsec:pap:dist}. In~\cref{subsec:pseudo_calib_technique}, we
recall the precise details in applying pseudocalibration. Then we pseudocalibrate in the Gaussian~(\cref{subsec:calib:gauss}) and
boolean~(\cref{subsec:calib:bool}) settings.

\subsection{PAP planted distribution}\label{subsec:pap:dist}

We formally define the random and the planted distributions for the
Planted Affine Planes problem in the Gaussian and boolean
settings. These two (families of) distributions are required by the
pseudocalibration machinery in order to define a candidate
pseudoexpectation operator $\pE$. For the Gaussian setting, we have
the following distributions.

\begin{definition}[Gaussian PAP distributions]\label{def:prob:pap:gauss:dist}
  The Gaussian PAP distributions are as follows.
  \begin{enumerate}
      \item (Random distribution) $m$ i.i.d. vectors $d_u \sim \gauss{0}{I}$.
      \item (Planted distribution) A vector $v$ is sampled uniformly from $\left\{\pm \frac{1}{\sqrt{n}}\right\}^n$, as well as signs $b_u \unif \{\pm 1\}$,
             and $m$ vectors $d_u$ are drawn from $\mathcal{N}(0, I)$ conditioned on $\ip{d_u}{v} = b_u$.
  \end{enumerate}
\end{definition}

For the boolean setting, we have the following distributions.
\begin{definition}[Boolean PAP distributions]\label{def:prob:pap:bool:dist}
  The boolean PAP distributions are as follows
  \begin{enumerate}
      \item (Random distribution) $m$ i.i.d. vectors $d_u \unif \{-1,+1\}^n$.
      \item (Planted distribution) A vector $v$ is sampled uniformly from $\left\{\pm \frac{1}{\sqrt{n}}\right\}^n$, as well as signs $b_u \unif \{\pm 1\}$, and $m$ vectors $d_u$ are drawn from $\left\{\pm 1\right\}^n$ conditioned on $\ip{d_u}{v} = b_u$.
  \end{enumerate}
\end{definition}

\subsection{Pseudocalibration technique}\label{subsec:pseudo_calib_technique}

We will use the shorthand $\E_{\text{ra}}$ and $\E_{\text{pl}}$ for
the expectation under the random and planted distributions.
Pseudocalibration gives a method for constructing a candidate
pseudoexpectation operator $\pE$.  
The idea behind pseudocalibration is that
$\E_{\text{ra}} \pE f(v)$ should match with $\E_{\text{pl}} f(v)$ for
every low-degree test of the data $t = t(d) = t(d_1, \dots, d_m)$,
\[\E_{\text{ra}} t(d) \pE f(v) = \E_{\text{pl}} t(d) f(v) .\]
When pseudocalibrating, one can freely choose the ``outer'' basis in
which to express the polynomial $f(v)$, as well as the ``inner'' basis
of low-degree tests which should agree with the planted
distribution. Though we attempted to use alternate bases to simplify
the analysis, ultimately we opted for the standard choice of bases: a
Fourier basis for the inner basis in each setting (Hermite functions
for the Gaussian setting, parity functions for the boolean setting),
and the coordinate basis $v^I$ for the outer basis.

When the inner basis is orthonormal under the random distribution (as a Fourier basis is), the pseudocalibration condition
gives a formula for the coefficients of $\pE f(v)$ in the orthonormal basis (though it only gives the coefficients of the low-degree functions $t(d)$). Concretely, letting the inner basis be indexed by $\alpha \in \calF$, as a function of $d$ the pseudocalibration condition enforces
\[ \pE f(v) = \displaystyle\sum_{\substack{\alpha \in \calF: \\ \abs{\alpha} \leq n^\tau}} \left( \E_{\text{pl}} t_\alpha(d) f(v) \right)t_\alpha (d).\]
Here we use ``$\abs{\alpha} \leq n^\tau$'' to describe the set of low-degree tests. The pseudocalibration condition does not prescribe any coefficients for functions $t_\alpha(d)$ with $\abs{\alpha} > n^\tau$ and an economical choice is to set these coefficients to zero. 

When pseudocalibrating, our pseudoexpectation operator is guaranteed to be linear, as the expression above is linear in $f$. It is guaranteed to satisfy all constraints of the form ``$f(v) =0$''. It will approximately satisfy constraints of the form ``$f(v, d) = 0$'', though only up to truncation error. 

\begin{fact}[Proof in \cref{lem:boolean-approximate-constraints}]\label{lem:pE-constraints}
  If $p(v)$ is a polynomial which is uniformly zero on the planted
  distribution, then $\pE[p]$ is the zero function. If $p(v,d)$ is a polynomial which is uniformly zero on the planted distribution, then the only nonzero Fourier coefficients of $\pE[p]$ are those with size between $n^\tau \pm \deg_d(p)$.
\end{fact}

Truncation
introduces a tiny error in the constraints, which we are able to handle in a
uniform way in~\cref{sec:exact-constraints}.

For the pseudocalibration we truncate to only Fourier coefficients of
size at most $n^\tau$. The relationship between the parameters is $\delta \le c\tau \le c'\eps$ where $c' < c < 1$ are absolute constants. We will assume that they are sufficiently small for all our proofs to go through.

Pseudocalibration also by default does not enforce the condition $\pE[1] = 1$. However, this is easily fixed by dividing the operator by $\pE[1]$. As will be pointed out in~\cref{rmk:pe-one}, w.h.p. in the unnormalized pseudocalibration, $\pE[1] = 1 + \littleoh_n(1)$ and so the error introduced does not impact the statement of any lemmas.

\subsection{Gaussian setting pseudocalibration}\label{subsec:calib:gauss}

We start by computing the pseudocalibration for the Gaussian setting. Here the natural choice of Fourier basis is the Hermite polynomials. Let $\alpha \in (\N^n)^m$ denote a Hermite polynomial index. Define $\alpha! \defeq \prod_{u,i} \alpha_{u,i}!$ and $\abs{\alpha} \defeq \sum_{u, i} \alpha_{u,i}$ and $\abs{\alpha_u} \defeq \sum_i \alpha_{u,i}$. We let $h_\alpha(d_1, \dots, d_m)$ denote an unnormalized Hermite polynomial, so that $h_{\alpha}/\sqrt{\alpha!}$ forms an orthonormal basis for polynomials in the entries of the vectors $d_1, \dots, d_m$, under the inner product $\ip{p}{q} = \E_{d_1, \dots, d_m \sim \calN(0, I)} [p \cdot q]$.

We can view $\alpha$ as an $m\times n$ matrix of natural numbers, and with this view we also define $\alpha^\T \in (\N^m)^n$.
\begin{lemma}\label{lem:gaussian-pseudocal}
For any $I \subseteq [n]$, the pseudocalibration value is
\[\pE v^I = \displaystyle\sum_{\substack{\alpha: \abs{\alpha} \leq n^\tau,\\ \abs{\alpha_u} \text{ even}, \\ \abs{(\alpha^\T)_i} \equiv I_i \; (\mod 2)}} \left(\prod_{u = 1}^m h_{\abs{\alpha_u}}(1) \right)\cdot\frac{1}{n^{\abs{I}/2 + \abs{\alpha}/2}} \cdot\frac{h_{\alpha}(d_1, \dots, d_m)}{\alpha!}. \]
\end{lemma}
In words, the nonzero Fourier coefficients are those which have even row sums, and whose column sums match the parity of $I$.  
\begin{proof}
The truncated pseudocalibrated value is defined to be
\[\pE v^I = \displaystyle\sum_{\alpha : \abs{\alpha} \leq n^\tau} \frac{h_{\alpha}(d_1, \dots, d_m)}{\alpha!} \cdot \E_{\text{pl}}[h_{\alpha}(d_1, \dots, d_m) \cdot v^I] \]
So we set about to compute the planted moments. For this computation, the following lemma is crucial. Here, we give a short proof of this lemma using generating functions. For a combinatorial proof, see Appendix \ref{app:combinatorialpseudocalibration}.

\begin{lemma}\label{lem:fixed-moments}
Let $\alpha \in \N^n$. When $v$ is fixed and $b$ is fixed (not necessarily +1 or -1) and $d \sim N(0, I)$ conditioned on $\ip{v}{d} = b\norm{v}$, 
\[\E_{d}[h_{\alpha}(d)] = \frac{v^\alpha}{\norm{v}^{\abs{\alpha}}} \cdot h_{\abs{\alpha}}(b).\]
\end{lemma}
\begin{proof}
It suffices to prove the claim when $\norm{v} = 1$ since the left-hand side is independent of $\norm{v}$. Express $d = bv + (I - vv^\T)x$ where $x \sim N(0,I)$ is a standard normal variable. Now we want
\[\E_{x \sim N(0,I)} h_{\alpha}\left(bv + (I - vv^\T)x\right). \]
The Hermite polynomial generating function is
\[\displaystyle\sum_{\alpha \in \N^n} \E_{x \sim N(0,I)} h_{\alpha}\left(bv + (I - vv^\T)x\right)\frac{t^{\alpha}}{\alpha!} = \E_x\exp\left(\ip{bv + (I - vv^\T)x}{t} - \frac{\norm{t}_2^2}{2}\right)\]
\[= \int_{\mathbb{R}^n} \frac{1}{(2\pi)^{\frac{n}{2}}} \cdot \exp\left(\ip{bv + (I - vv^\T)x}{t} - \frac{\norm{t}_2^2}{2} - \frac{\norm{x}_2^2}{2}\right) \; dx. \]
Completing the square,
\begin{align*}= &  \int_{\mathbb{R}^n} \frac{1}{(2\pi)^{\frac{n}{2}}} \cdot \exp\left(\ip{bv}{t} - \frac{\ip{v}{t}^2}{2} - \frac{1}{2} \cdot\norm{x- (t - \ip{v}{t}v)}_2^2\right) \; dx \\
=&  \exp\left(\ip{bv}{t} - \frac{\ip{v}{t}^2}{2}\right) \\
=& \exp\left(b\ip{v}{t} - \frac{1}{2} \cdot \ip{v}{t}^2\right).
\end{align*}
How can we Taylor expand this in terms of $t$? The Taylor expansion of $\exp(by - \frac{y^2}{2})$ is $\sum_{i=0}^\infty h_i(b) \frac{y^i}{i!}$. That is, the $i$-th derivative in $y$ of $\exp(by - \frac{y^2}{2})$, evaluated at 0, is $h_i(b)$. Using the chain rule with $y=\ip{v}{t}$, the $\alpha$-derivative in $t$ of our expression, evaluated at 0, is $v^\alpha \cdot h_{\abs{\alpha}}(b)$. This is the expression we wanted when $\norm{v} = 1$, and along with the aforementioned remark about homogeneity in $\norm{v}$ this completes the proof.
\end{proof}

Now we can finish the calculation.
To compute $\E_{\text{pl}}[h_{\alpha}(d_1, \dots, d_m) \cdot v^I]$, marginalize $v$ and the $b_u$ and factor the conditionally independent $b_u$ and $d_u$.
\begin{align*}
    \E_{\text{pl}}[h_{\alpha}(d_1, \dots, d_m) v^I] &= \E_{v, b_u} v^I \prod_{u=1}^m \E_{d}\left[h_{\alpha_u}(d_u) \mid v, b_u\right]\\
    &= \E_{v, b_u} v^I \cdot \prod_{u=1}^m \frac{v^{\alpha_u}}{\norm{v}^{\abs{\alpha_u}}} \cdot h_{\abs{\alpha_u}}(b_u) && (\text{\cref{lem:fixed-moments})}\\
    &= \left(\E_{v} \frac{v^{I + \sum_{u=1}^m \alpha_u}}{\norm{v}^{\sum_{u=1}^m\abs{\alpha_u}}} \right) \cdot \left( \prod_{u=1}^m \E_{b_u}h_{\abs{\alpha_u}}(b_u)\right)
\end{align*}
The Hermite polynomial expectations will be zero in expectation over $b_u$ if the degree is odd, and otherwise $b_u$ is raised to an even power and can be replaced by 1. This requires that $\abs{\alpha_u}$ is even for all $u$. The norm $\norm{v}$ is constantly $1$ and can be dropped. The numerator will be $\frac{1}{n^{\abs{I}/2 + \abs{\alpha}/2}}$ if the parity of every $\abs{(\alpha^\T)_i}$ matches $I_i$, and 0 otherwise. This completes the pseudocalibration calculation.
\end{proof}

We can now write $\calM$ in terms of graph matrices.

\begin{definition}\label{def:calL_valid_shapes}
	Let $\cal{L}$ be the set of all proper shapes $\alpha$ with the following properties
	\begin{itemize}
		\item $U_{\alpha}$ and $V_{\alpha}$ only contain square vertices and $|U_{\alpha}|, |V_{\alpha}| \le n^{\delta}$
		\item $W_\alpha$ has no degree $0$ vertices
		\item $\deg(\square{i}) + U_\alpha(\square{i}) + V_\alpha(\square{i})$ is even for all $\square{i} \in V(\alpha)$ 
		\item $\deg(\circle{u})$ is even and $\deg(\circle{u}) \ge 4$ for all $\circle{u} \in V(\alpha)$
		\item $|E(\alpha)| \le n^\tau$
	\end{itemize}
\end{definition}

\begin{remark}
  Note that the shapes in $\cal{L}$ can have isolated vertices in $U_{\alpha} \cap V_{\alpha}$.
\end{remark}

\begin{remark}\label{rmk:circle_deg_bound}
	$\calL$ captures all the shapes that have nonzero coefficient when we write $\calM$ in terms of graph matrices. The constraint $\deg(\circle{u}) \ge 4$ arises because pseudocalibration gives us that $\deg(\circle{u})$ is even, $\circle{u}$ cannot be isolated, and $h_2(1) = 0$.
\end{remark}

%Note that a corollary of the definition is that $\abs{U_\alpha} + \abs{V_\alpha}$ must be even for $\alpha \in \calL$. We will use the notation $n_{(k)} \defeq n(n+2)(n+4)\cdots (n+2k-2)$ to denote a rising factorial. 
% \cnote{Commented out, but also incorrect.} In this notation the moments of the uniform distribution on the sphere are
%     \[\E_{v \unif S^n} v^I = \left\{\begin{array}{lr}
%         \frac{\prod_{i=1}^n I_i!!}{n_{(\abs{I}/2)}} & \text{if all }I_i\text{ are even}\\
%         0 & \text{ otherwise}
%     \end{array}\right. \]

For a shape $\alpha$, we define 
\[\alpha! \defeq \prod_{e \in E(\alpha)} l(e)!\] Note that this equals the factorial of the corresponding index of the Hermite polynomial for this shape.

\begin{definition}
	For any shape $\alpha$, if $\alpha \in \cal{L}$, define \[\lambda_{\alpha} \defeq 
	\left( \prod_{\circle{u}\in V(\alpha)} h_{\deg(\circle{u})}(1)\right) 
	\cdot \frac{1}{ n^{(\abs{U_\alpha} + \abs{V_\alpha} + \abs{E(\alpha)})/2}}
	\cdot \frac{1}{\alpha!}  \]
	Otherwise, define $\lambda_{\alpha} \defeq 0$.
\end{definition}

\begin{corollary} Modulo the footnote\footnote{Technically, the graph matrices $M_\alpha$ have rows and columns indexed by all subsets of $\calC_m \cup \calS_n$. The submatrix with rows and columns from $\binom{\calS_n}{\leq D/2}$ equals the moment matrix for $\pE$ as defined in~\cref{sec:moment-mtx}.}, $\calM = \displaystyle\sum_{\text{shapes }\alpha} \lambda_{\alpha} M_{\alpha}$.
\end{corollary}

\subsection{Boolean setting pseudocalibration}\label{subsec:calib:bool}

We now present the pseudocalibration for the boolean setting. For the
sequel, we need notation for vectors on a slice of the boolean
cube.

\begin{definition}[Slice]
  Let $v \in \set{\pm 1}^n$ and $\theta \in \mathbb{Z}$. The slice $\slice_{v}(\theta)$ is defined as
  $$
  \slice_{v}(\theta) \coloneqq \set{d \in \set{\pm 1}^n ~\vert~ \ip{v}{d} = \theta}.
  $$
  We use $\slice_{v}(\pm \theta)$ to denote $\slice_{v}(\theta) \cup \slice_{v}(-\theta)$ and
  $\slice(\theta)$ to denote $\slice_{v}(\theta)$ when $v$ is the all-ones vector.
\end{definition}

\begin{remark}
  With our notation for the slice, the planted distribution in the boolean setting can be equivalently described as
  \begin{enumerate}
    \item Sample $v \in \set{\frac{\pm 1}{\sqrt{n}}}^n$ uniformly, and then
    \item Sample $d_1,\dots,d_m$ independently and uniformly from $\slice_{\sqrt{n} \cdot v}(\pm\sqrt{n})$.
  \end{enumerate}
\end{remark}
The planted distribution doesn't actually exist for every $n$, but this is immaterial, as we can still define the pseudoexpectation via the same formula.

We will also need the expectation of monomials over the slice
$\slice(\sqrt{n})$ since they will appear in the description of the
pseudocalibrated Fourier coefficients.

\begin{definition}
   $
   e(k) \coloneqq \E_{x \unif \mathcal{S}(\sqrt{n})}\left[x_1\cdots x_k\right].
   $   
\end{definition}

We now compute the Fourier coefficients of $\pE v^{\beta}$, where
$\beta \in \mathbb{F}_2^n$. The Fourier basis when $d_1, \dots, d_m \unif \{\pm 1\}^n$ is the set of parity functions. Thus a character can be specified by $\alpha \in (\F_2^n)^m$, where $\alpha$
is composed of $m$ vectors
$\alpha_1,\dots,\alpha_m \in \mathbb{F}_2^n$.  More precisely, the
character $\chi_{\alpha}$ associated to $\alpha$ is defined as
\[ \chi_\alpha(d_1,\dots,d_m) \defeq \prod_{u=1}^m d_u^{\alpha_u}\]
We denote by $\abs{\alpha}$ the number of non-zero entries of $\alpha$
and define $\abs{\alpha_u}$ similarly. Thinking of $\alpha$ as an $m\times n$ matrix with entries in $\F_2$, we also define $\alpha^\T \in (\F_2^n)^m$.

\begin{lemma}\label{lem:boolean-pseudocalibration}
We have
  $$
  \pE v^{\beta} = \frac{1}{n^{\abs{\beta}/2}}\sum_{\substack{\alpha \colon \abs{\alpha} \le n^\tau, \\ \abs{\alpha_u} \text{ even},\\ \abs{\alpha^\T_i} \equiv \beta_i \;(\mod 2) }} \prod_{u=1}^m e(\abs{\alpha_u}) \cdot \chi_{\alpha_u}(d_u).
  $$
\end{lemma}
The set of nonzero coefficients has a similar structure as in the Gaussian case: the rows of $\alpha$ must have an even number of entries, and the $i$-th column must have parity matching $\beta_i$.

\begin{proof}
   Given $\alpha \in (\F_2^n)^m$ with $\abs{\alpha} \le n^\tau$, the pseudocalibration equation enforces by construction that
   $$
   \E_{d_1,\dots,d_m \in \set{\pm 1}^n} (\pE v^{\beta})(d_1,\dots,d_m) \cdot \chi_{\alpha}(d_1,\dots,d_m) = \E_{\text{pl}} v^{\beta} \cdot \chi_{\alpha}(d_1,\dots,d_m).
   $$

  Computing the RHS above 
  yields
  \begin{align*}
    \E_{v \in \set{\pm 1}^n}  \E_{d_1,\dots,d_m \unif \slice_v(\pm \sqrt{n})}\left[ v^{\beta} \prod_{u=1}^m \chi_{\alpha_u}(d_u) \right] &= \E_{v \in \set{\pm 1}^n}  \E_{d_1,\dots,d_m \unif \slice(\pm \sqrt{n})} \left[ v^{\beta} \prod_{u=1}^m \chi_{\alpha_u}(v) \chi_{\alpha_u}(d_u) \right] \\
      & = \E_{v \in \set{\pm 1}^n} \chi_{\alpha_1+\cdots +\alpha_m + \beta}(v) \E_{d_1,\dots,d_m \in \slice(\pm \sqrt{n})}\left[ \prod_{i=1}^m \chi_{\alpha_i}(d_i) \right] \\
      & = \one_{\left[\alpha_1+\cdots +\alpha_m=\beta\right]} \cdot \prod_{i=1}^m \E_{d_i \in \slice(\pm \sqrt{n})} \left[ \chi_{\alpha_i}(d_i) \right]\\
      & = \one_{\left[\alpha_1+\cdots +\alpha_m=\beta\right]} \cdot \prod_{i=1}^m \one_{\left[\abs{\alpha_i} \equiv 0 \pmod{2} \right]} \cdot \prod_{i=1}^m e(\abs{\alpha_i}).
\end{align*}
  Since we have a general expression for the Fourier coefficient of each character, 
  applying Fourier inversion  concludes the proof.
\end{proof}

We can now express the moment matrix in terms of graph matrices.

\begin{definition}
    Let $\calL_{bool}$ be the set of shapes in $\calL$ from~\cref{def:calL_valid_shapes} in which the edge labels are all 1.
\end{definition}

\begin{remark}
	$\calL_{bool}$ captures all the shapes that have nonzero coefficient when we write $\calM$ in terms of graph matrices. Similar to~\cref{rmk:circle_deg_bound}, since $e(2) = 0$ (see~\cref{claim:e2}), we have the same condition $deg(\circle{u}) \ge 4$ for shapes in $\calL_{bool}$.
\end{remark}

\begin{definition}
For all shapes $\alpha$, if $\alpha \in \calL_{bool}$ define
    \[\lambda_\alpha \defeq   \frac{1}{n^{(\abs{U_\alpha} + \abs{V_\alpha})/2}}\prod_{\circle{u} \in V(\alpha)} e(\deg(\circle{u}))\]
Otherwise, let $\lambda_\alpha \defeq 0$.
\end{definition}

\begin{corollary}$\calM = \displaystyle\sum_{\text{shapes }\alpha} \lambda_\alpha M_\alpha$
\end{corollary}

\subsubsection{Unifying the analysis}

It turns out that the analysis of the boolean setting mostly
follows from the analysis in the Gaussian setting.
Initially, the boolean pseudocalibration is essentially equal to
the Gaussian pseudocalibration in which we have removed all shapes
containing at least one edge with a label $k \ge 2$. The coefficients
on the graph matrices will actually be slightly different, but they
both admit an upper bound that is sufficient for our purposes
(see~\cref{prop:coefficient-bound} for the precise statement).

To unify the notation in our analysis, we conveniently set the edge
functions of the graphs in the boolean case to be
\[h_k(x) = \left\{\begin{array}{lr}
    1 & \text{if }k = 0\\
    x & \text{if }k = 1\\
    0 & \text{if }k \geq 2
\end{array}
\right. \]
This choice of $h_k(x)$ preserves the fact that
$\{h_0(x)=1,h_1(x)=x\}$ is an orthogonal polynomial basis in
the boolean setting, while zeroing out graphs with larger labels.

During the course of the analysis, we may multiply two graph matrices
and produce graph matrices with improper parallel edges (so-called
``intersections terms"). For a fixed pair ${u,i}$ of vertices, parallel
edges between $u$ and $i$ with labels $l_1,\dots,l_s$ correspond to
the product of orthogonal polynomials $\prod_{j=1}^s
h_{l_j}(d_{u,i}) \eqqcolon q(d_{u,i})$. We will re-express this product
as a linear combination of polynomials in the orthogonal family, \ie
$q(d_{u,i}) = \sum_{i=0}^{\textup{deg}(q)} \lambda_i \cdot
h_i(d_{u,i})$ for some coefficients $\lambda_i \in \mathbb{R}$. For
the boolean case, the polynomial $q(d_{u,i})$ will be either
$h_0(d_{u,i})=1$ or $h_1(d_{u,i})=d_{u,i}$. However, for the Gaussian
setting there may be up to $\textup{deg}(q)$ non-zero, potentially larger coefficients
$\lambda_i$ for the corresponding Hermite polynomials $h_i$.
For the graphs that arise in this way, we will always bound their 
contributions to $\calM$ by applying the triangle inequality and norm 
bounds. Since we show bounds using the larger coefficients $\lambda_i$ from the Gaussian case,
the same bounds apply when using the 0/1 coefficients in the boolean case.

We will consider separate cases at any point where the analysis differs between the two settings.

\section{Proving PSD-ness}\label{sec:psd}

Looking at the shapes that make up $\calM$, the trivial shape with $k$ square vertices contributes an identity matrix on the degree-$2k$ submatrix of $\calM$. Our ultimate goal will be to bound all shapes against these identity matrices. 

\begin{definition}[Block]
    For $k,l \in \{0,1, \dots, D/2\}$, the $(k,l)$ block of $\calM$ is the submatrix with rows from $\binom{[n]}{k}$ and columns from $\binom{[n]}{l}$. Note that when $\calM$ is expressed as a sum of graph matrices, this exactly restricts $\calM$ to shapes $\alpha$ with $\abs{U_\alpha} = k$ and $\abs{V_\alpha} = l$.
\end{definition}

We define the parameter $\eta \defeq 1/\sqrt{n}$. The trivial shapes
live in the diagonal blocks of $\calM$, and on the $(k,k)$ block
contribute a factor of $\frac{1}{n^k} = \eta^{2k}$ on the diagonal.
In principle, we could make $\eta$ as small as we like\footnote{Though
pseudocalibration truncation errors may become nonnegligible for
extremely tiny $\eta$.} by considering the moments of a rescaling of
$v$ rather than $v$ itself. Counterintuitively, it will turn out that
the scaling helps us prove PSD-ness (see~\cref{app:comment_on_scaling}
for more details). It turns out that pseudocalibrating $v$ as a unit
vector (equivalently, using $\eta = 1/\sqrt{n}$) is sufficient for our
analysis.

Towards the goal of bounding $\calM$ by the identity terms, we will bound the norm of matrices on each block of $\calM$, and invoke the following lemma to conclude PSD-ness.

\begin{lemma}\label{lem:block-psd}
    Suppose a symmetric matrix $\calA \in \R^{\binom{[n]}{\leq D} \times \binom{[n]}{\leq D}}$ satisfies, for some parameter $\eta \in (0, 1)$,
    \begin{enumerate}
        \item For each $k \in \{0, 1,\dots, D\}$, the $(k,k)$ block has minimum singular value at least $\eta^{2k}(1-\frac{1}{D+1})$
        \item For each $k,l \in \{0, 1,\dots, D\}$ such that $k \neq l$, the $(k,l)$ block has norm at most $\frac{\eta^{k+l}}{D+1}$.
    \end{enumerate}
    Then $\calA \psdgeq 0$.
\end{lemma}
\begin{proof}
We need to show that for all vectors $x$, ${x^\T}{\calA}x \geq 0$. Given a vector $x$, let $x_0,\dots,x_D$ be its components in blocks $0,\dots,D$. Observe that 
\begin{align*}
&{x^\T}{\calA}x \geq \sum_{k \in [0,D]}{\eta^{2k}\left(1 - \frac{1}{D+1}\right)\norm{x_k}^2} - \sum_{k \neq l \in [0,D]}{\frac{\eta^{k+l}}{D+1}\norm{x_k}\norm{x_l}} \\
&= (\norm{x_0},{\eta}\norm{x_1},\dots,{\eta}^D\norm{x_D}) 
\begin{pmatrix}
1-\frac{1}{D+1} & -\frac{1}{D+1} & \cdots & -\frac{1}{D+1} \\
-\frac{1}{D+1} & 1-\frac{1}{D+1} & \cdots & -\frac{1}{D+1} \\
\vdots  & \vdots  & \ddots & \vdots  \\
-\frac{1}{D+1} & -\frac{1}{D+1} & \cdots & 1-\frac{1}{D+1}
\end{pmatrix}
\begin{pmatrix}
\norm{x_0}\\
{\eta}\norm{x_1}\\
\vdots\\
{\eta}^D\norm{x_D}
\end{pmatrix} 
\geq 0.
\end{align*}
\end{proof}

We start by defining spiders, which are special shapes $\alpha$ that we will handle separately in the decomposition of $\calM$. Informally, these contain special substructures which allow their norm bounds not to be negligible with respect to the identity matrix. We then show that shapes which are not spiders have bounded norms.

\begin{definition}[Left Spider]
	A left spider is a proper shape $\alpha = (V(\alpha), E(\alpha), U_{\alpha}, V_{\alpha})$ with the property that there exist two distinct square vertices $\square{i}, \square{j} \in U_{\alpha}$ of degree $1$ and a circle vertex $\circle{u} \in V(\alpha)$ such that $E(\alpha)$ contains the edges $(\square{i}, \circle{u})$ and $(\square{j}, \circle{u})$ (these are necessarily the only edges incident to $\square{i}$ and $\square{j}$).
\end{definition}
The vertices $\square{i}$ and $\square{j}$ are called the \textit{end vertices} of $\alpha$. Because of degree parity, the end vertices must lie in $U_\alpha \setminus (U_\alpha \cap V_\alpha)$.

\begin{definition}[Right spider]
	A shape $\alpha = (V(\alpha), E(\alpha), U_{\alpha}, V_{\alpha})$ is a right spider if $\alpha^\T = (V(\alpha), E(\alpha), V_{\alpha}, U_{\alpha})$ is a left spider. The end vertices of $\alpha^\T$ are also called the end vertices of $\alpha$.
\end{definition}

\begin{definition}[Spider]
	A shape $\alpha$ is a spider if it is either a left spider or a right spider.
\end{definition}
\begin{remark}
	A spider can have many pairs of end vertices. For each possible spider shape, we single out a pair of end vertices, so that in what follows we can discuss ``the'' end vertices of the spider.
\end{remark}

\subsection{Non-spiders are negligible}

For non-spiders, we will now show that their norm is small. We point out that this norm bound on non-spiders critically relies on the assumption $m \leq n^{3/2 - \eps}$.

\begin{lemma}\label{lem:charging}
	If $\alpha \in \calL$ is not a trivial shape and not a spider, then
	\[\frac{1}{n^{|E(\alpha)|/2}} n^{\frac{w(V(\alpha)) - w(S_{\min})}{2}} \le \frac{1}{n^{\Omega(\eps |E(\alpha)|)}}\]
	where $S_{min}$ is the minimum vertex separator of $\alpha$.
\end{lemma}

\begin{proof}
The idea behind the proof is as follows. Each square vertex which is not in the minimum vertex separator contributes $\sqrt{n}$ to the norm bound while each circle vertex which is not in the minimum vertex separator contributes $\sqrt{m}$. To compensate for this, we will try and take the factor of $\frac{1}{\sqrt{n}}$ from each edge and distribute it among its two endpoints so that each square vertex which is not in the minimum vertex separator is assigned a factor of $\frac{1}{\sqrt{n}}$ or smaller and each circle vertex which is not in the minimum vertex separator is assigned a factor of $\frac{1}{\sqrt{m}}$ or smaller.
\begin{remark}
Instead of using the minimum vertex separator, we will actually use a set $S$ of square vertices such that $w(S) \leq w(S_{\min})$. For details, see the actual distribution scheme below.
\end{remark}
To motivate the distribution scheme which we use, we first give two attempts which don't quite work. For simplicity, for these first two attempts we assume that $U_{\alpha} \cap V_{\alpha} = \emptyset$ as vertices in $U_{\alpha} \cap V_{\alpha}$ can essentially be ignored.
\begin{enumerate}
\item[] Attempt 1: Take each edge and assign a factor of $\frac{1}{\sqrt[4]{n}}$ to its square endpoint and a factor of $\frac{1}{\sqrt[8]{m}}$ to its circle endpoint.

With this distribution scheme, since each circle vertex has degree at least $4$, each circle vertex is assigned a factor of $\frac{1}{\sqrt{m}}$ or smaller. Since each square vertex in $W_{\alpha}$ has degree at least $2$, each square vertex in $W_{\alpha}$ is assigned a factor of $\frac{1}{\sqrt{n}}$ or smaller. However, square vertices in $U_{\alpha} \cup V_{\alpha}$ may only have degree $1$ in which case they are assigned a factor of $\frac{1}{\sqrt[4]{n}}$ which is not small enough.

To fix this issue, we can have all of the edges which are incident to a square vertex in $U_{\alpha} \cup V_{\alpha}$ give their entire factor of $\frac{1}{\sqrt{n}}$ to the square vertex.
\begin{remark}\label{rmk:pe-one}
For analyzing $\pE[1]$, this first attempt works as $U_{\alpha} = V_{\alpha} = \emptyset$. Thus, as long as $m \leq n^{2 - \epsilon}$, with high probability $\pE[1] = 1 \pm \littleoh_n(1)$ .
\end{remark}
\item[] Attempt 2: For each edge which is between a square vertex in $U_{\alpha} \cup V_{\alpha}$ and a circle vertex, we assign a factor of $\frac{1}{\sqrt{n}}$ to the square vertex and nothing to the circle vertex. For all other edges, we assign a factor of $\frac{1}{\sqrt[4]{n}}$ to its square endpoint and a factor of $\frac{1}{\sqrt[6]{m}}$ to its circle endpoint (which we can do because $m \leq n^{\frac{3}{2} - \epsilon}$).

With this distribution scheme, each square vertex is assigned a factor of $\frac{1}{\sqrt{n}}$. Since $\alpha$ is not a spider, no circle vertex is adjacent to two vertices in $U_{\alpha}$ or $V_{\alpha}$. Thus, any circle vertex which is not adjacent to both a square vertex in $U_{\alpha}$ and a square vertex in $V_{\alpha}$ must be adjacent to at least $3$ square vertices in $W_{\alpha}$ and is thus assigned a factor of $\frac{1}{\sqrt{m}}$ or smaller. However, we can have circle vertices which are adjacent to both a square vertex in $U_{\alpha}$ and a square vertex in $V_{\alpha}$. These circle vertices may be assigned a factor of $\frac{1}{\sqrt[3]{m}}$, which is not small enough.

To fix this, observe that whenever we have a circle vertex which is adjacent to both a square vertex in $U_{\alpha}$ and a square vertex in $V_{\alpha}$, this gives a path of length $2$ from $U_{\alpha}$ to $V_{\alpha}$. Any vertex separator must contain one of the vertices in this path, so we can put one of these two square vertices in $S$ and not assign it a factor of $\frac{1}{\sqrt{n}}$.
\item[] Actual distribution scheme: Based on these observations, we use the following distribution scheme. Here we are no longer assuming that $U_{\alpha} \cap V_{\alpha}$ is empty.
\begin{enumerate}
\item[1.] Choose a set of square vertices $S \subseteq U_{\alpha} \cup V_{\alpha}$ as follows. Start with $S = U_{\alpha} \cap V_{\alpha}$. Whenever we have a circle vertex which is adjacent to both a square vertex in $U_{\alpha} \setminus V_{\alpha}$ and a square vertex in $V_{\alpha} \setminus U_{\alpha}$, put one of these two square vertices in $S$ (this choice is arbitrary). Observe that $w(S) \leq w(S_{\min})$
\item[2.] For each edge which is incident to a square vertex in $S$, assign a factor of $\frac{1}{\sqrt[3]{m}}$ to its circle endpoint and nothing to this square. 
\item[3.] For each edge which is incident to a square vertex in $(U_{\alpha} \cup V_{\alpha}) \setminus S$, assign a factor of $\frac{1}{\sqrt{n}}$ to the square vertex and nothing to the circle vertex.
\item[4.] For all other edges, assign a factor of $\frac{1}{\sqrt[4]{n}}$ to its square endpoint and a factor of $\frac{1}{\sqrt[6]{m}}$ to its circle endpoint.
\end{enumerate}
Now each square vertex which is not in $S$ is assigned a factor of $\frac{1}{\sqrt{n}}$ and since $\alpha$ is not a spider, all circle vertices are assigned a factor of $\frac{1}{\sqrt{m}}$ or smaller. 
\end{enumerate}
We now make this argument formal. 

	Let $\calC_{\alpha}$ and $\calS_{\alpha}$ be the set of circle vertices and the set of square vertices in $\alpha$ respectively. We have $ n^{\frac{w(V(\alpha)) - w(S_{\min})}{2}} \leq n^{0.5|\calS_{\alpha} \setminus S_{min}| + (0.75 - \frac{\eps}{2})|\calC_{\alpha} \setminus S_{min}|}$. So, it suffices to prove that
	\[|E(\alpha)| - |\calS_{\alpha} \setminus S_{min}| - (1.5 - \eps)|\calC_{\alpha} \setminus S_{min}| \ge \Omega(\epsilon |E(\alpha)|)\]
	
	Let $Q = U_{\alpha} \cap V_{\alpha}, P = (U_{\alpha} \cup V_{\alpha}) \setminus Q$ and let $P'$ be the set of vertices of $P$ that have degree $1$ and are not in $S_{min}$. Let $E_1$ be the set of edges incident to $P'$ and let $E_2 = E(\alpha) \setminus E_1$. 
	
	For each vertex $\square{i}$ (resp. $\circle{u}$), let the number of edges of $E_2$ incident to it be $\deg'(\square{i})$ (resp. $\deg'(\circle{u})$). Since $\alpha$ is bipartite, we have that $|E_2| = \sum_{\square{i} \in \calS_{\alpha}} \deg'(\square{i}) = \sum_{\circle{u} \in \calC_{\alpha}} \deg'(\circle{u})$. We get that
	\[|E(\alpha)| = |E_1| + |E_2| = |P'| + \frac{1}{2}(\sum_{\square{i} \in \calS_{\alpha}} \deg'(\square{i}) + \sum_{\circle{u} \in \calC_{\alpha}} \deg'(\circle{u}))\]
	
	We also have $|S_{\alpha} \setminus S_{min}| \le |P'| + |\calS_{\alpha} \cap W_{\alpha}| + |\calS_{\alpha} \cap (P \setminus P')| \le |P'| + \frac{1}{2} \sum_{\square{i} \in \calS_{\alpha}} \deg'(\square{i})$ because each square vertex outside $P' \cup Q$ has degree at least $2$ and is not incident to any edge in $E_1$. So, it suffices to prove
	\[\frac{1}{2}\sum_{\circle{u} \in \calC_{\alpha}} \deg'(\circle{u}) - (1.5 - \eps)|\calC_{\alpha} \setminus S_{min}| \ge \Omega(\epsilon |E(\alpha)|)\]
	
	Now, observe that each $\circle{u} \in \calC_{\alpha}$ is incident to at most two edges in $E_1$. This is because if it were adjacent to at least $3$ edges in $E_1$, then either $\circle{u}$ is adjacent to at least two vertices of degree $1$ in $U_{\alpha}$ or $\circle{u}$ is adjacent to at least two vertices of degree $1$ in $V_{\alpha}$. However, this cannot happen since $\alpha$ is not a spider. This implies that $\deg'(\circle{u}) \ge \deg(\circle{u}) - 2$.
	
	Note moreover that if $\circle{u} \in \calC_{\alpha} \setminus S_{min}$, we have that $\deg'(\circle{u}) \ge \deg(\circle{u}) - 1$. This is because, building on the preceding argument, $\deg'(\circle{u}) = \deg(\circle{u}) - 2$ can only happen if there exist $\square{i} \in U_{\alpha}, \square{j} \in V_{\alpha}$ such that $(\square{i}, \circle{u}), (\square{j}, \circle{u}) \in E_1$. But then, note that we have $\square{i}, \square{j} \not\in S_{min}$ by definition of $P'$ and also, $\circle{u} \not\in S_{min}$ by assumption. This means that there is a path from $U_{\alpha}$ to $V_{\alpha}$ which does not pass through $S_{min}$, which is a contradiction.
	
	Finally, we set $\epsilon$ small enough such that the following inequalities are true, both of which follow from the fact that $\deg(\circle{u}) \ge 4$ for all $\circle{u} \in \calC_{\alpha}$.
	\begin{enumerate}
		\item For any $\circle{u} \in \calC_{\alpha} \cap S_{min}$, we have $\frac{\deg(\circle{u}) - 2}{2} \ge \frac{\eps}{10}\deg(\circle{u})$.
		\item For any $\circle{u} \in \calC_{\alpha} \setminus S_{min}$, we have $\frac{\deg(\circle{u}) - 1}{2} - 1.5 + \eps \ge \frac{\eps}{10}\deg(\circle{u})$.
	\end{enumerate}
	Using this, we get
	\begin{align*}
	\frac{1}{2}\sum_{\circle{u} \in \calC_{\alpha}} \deg'(\circle{u}) - (1.5 - \eps)|\calC_{\alpha} \setminus S_{min}| &\ge \sum_{\circle{u} \in \calC_{\alpha} \cap S_{min}} 
	\frac{\deg(\circle{u}) - 2}{2} + \sum_{\circle{u} \in \calC_{\alpha} \setminus S_{min}} 
	\frac{\deg(\circle{u}) - 1}{2} - (1.5 - \eps)|\calC_{\alpha} \setminus S_{min}|\\
	&\ge \sum_{\circle{u} \in \calC_{\alpha} \cap S_{min}} 
	\frac{\eps}{10}\deg(\circle{u}) + \sum_{\circle{u} \in \calC_{\alpha} \setminus S_{min}} \left(\frac{\deg(\circle{u}) - 1}{2} - 1.5 + \eps\right)\\
	&\ge \sum_{\circle{u} \in \calC_{\alpha} \cap S_{min}} 
	\frac{\eps}{10}\deg(\circle{u}) + \sum_{\circle{u} \in \calC_{\alpha} \setminus S_{min}} 
	\frac{\eps}{10}\deg(\circle{u})\\
	&= \sum_{\circle{u} \in \calC_{\alpha}} \frac{\eps}{10}\deg(\circle{u}) = \Omega(\eps|E(\alpha)|)
	\end{align*}
\end{proof}

Since $\calL_{bool} \subseteq \calL$, the above result extends to non-trivial non spider shapes in $\calL_{bool}$ too.

\begin{corollary}
	If $\alpha \in \calL_{bool}$ is not a trivial shape and not a spider, then
	\[\frac{1}{n^{|E(\alpha)|/2}} n^{\frac{w(V(\alpha)) - w(S_{\min})}{2}} \le \frac{1}{n^{\Omega(\eps |E(\alpha)|)}}\]
\end{corollary}

\begin{corollary}\label{cor:non_spider_killing}
	If $\alpha \in \calL$ is not a trivial shape and not a spider, then w.h.p. \[\frac{1}{n^{|E(\alpha)|/2}}\norm{M_{\alpha}} \le \frac{1}{n^{\Omega(\epsilon |E(\alpha)|)}}\]
\end{corollary}

\begin{proof}
	Using the norm bounds in~\cref{lem:gaussian-norm-bounds}, we have
	\[ \norm{M_\alpha} \leq 2\cdot\left(\abs{V(\alpha)} \cdot (1+\abs{E(\alpha)}) \cdot \log(n)\right)^{C\cdot (\abs{V_{rel}(\alpha)} + \abs{E(\alpha)})} \cdot n^q{\frac{w(V(\alpha)) - w(S_{\min}) + w(W_{iso})}{2}}\]
	We have $W_{iso} = \emptyset$. Observe that since there are no degree $0$ vertices in $V_{rel}(\alpha)$, we have that $|V_{rel}(\alpha)| \le 2|E(\alpha)|$ and since we also have $|V(\alpha)|\cdot (1+\abs{E(\alpha)})\cdot \log n \le n^{O(\tau)}$, the factor $2\cdot(\abs{V(\alpha)} \cdot (1+\abs{E(\alpha)}) \cdot \log(n))^{C\cdot (\abs{V_{rel}(\alpha)} + \abs{E(\alpha)})}$ can be absorbed into $\frac{1}{n^{\Omega(\eps |E(\alpha)|)}}$. The result follows from~\cref{lem:charging}.
\end{proof}

This says that nontrivial non-spider shapes have $\littleoh_n(1)$ norm (ignoring the extra factor $\eta$ for the moment). We now demonstrate how to use this norm bound to control the total norm of all non-spiders in a block of $\calM$,~\cref{cor:non-spider-sum}. We will first need a couple propositions which will also be of use to us later after we kill the spiders.

\begin{proposition}\label{prop:edge-shape-count}
	The number of proper shapes with at most $L$ vertices and exactly $k$ edges is at most $L^{8(k+1)}$.
\end{proposition}

\begin{proof}
    The following process captures all shapes (though many will be constructed multiple times):
  \begin{itemize}
      \item Choose the number of square and circle variables in each of the four sets $U \cap V, U \setminus (U \cap V), V \setminus (U \cap V), W$. This contributes a factor of $L^{8}$.
      \item Place each edge between two of the vertices. This contributes a factor of $L^{2 k}$.
  \end{itemize}
\end{proof}

\begin{proposition}\label{prop:coefficient-bound}
$\abs{\lambda_\alpha} \leq \eta^{\abs{U_\alpha} + \abs{V_\alpha}} \cdot  \frac{\abs{E(\alpha)}^{3\cdot \abs{E(\alpha)}}}{n^{\abs{E(\alpha)}/2}}$ where we assume by convention that $0^0 = 1$.
\end{proposition}
\begin{proof}
\noindent\textbf{(Gaussian setting)} Recall that the coefficients $\lambda_\alpha$ are either zero or are defined by the formula
\[\lambda_\alpha  = \eta^{\abs{U_\alpha} + \abs{V_\alpha}}\cdot \left( \prod_{\circle{u}\in V(\alpha)} h_{\deg(\circle{u})}(1)\right) 
	\cdot \frac{1}{ n^{\abs{E(\alpha)}/2}}
	\cdot \frac{1}{\alpha!}\]
	
	The sequence $h_k(1)$ satisfies the recurrence $h_0(1) = h_1(1) = 1, h_{k + 1}(1) = h_k(1) - kh_{k - 1}(1)$. We can prove by induction that $\abs{h_k(1)} \le k^k$ and hence, 
	\[\prod_{\circle{u}\in V(\alpha)} \abs{h_{\deg(\circle{u})}(1)} \le \prod_{\circle{u}\in V(\alpha)} (\deg(\circle{u}))^{\deg(\circle{u})} \le \abs{E(\alpha)}^{\abs{E(\alpha)}}.\]

\noindent\textbf{(Boolean setting)} In the boolean setting the coefficients $\lambda_\alpha$ are defined by
    \[\lambda_\alpha =  \eta^{\abs{U_\alpha} + \abs{V_\alpha}} \cdot \left(\prod_{\circle{u} \in V(\alpha)} e(\deg(\circle{u})) \right)\]
    Using~\cref{cor:bound_on_coeff_e_k}, we have that $\abs{e(k)} \le k^{3k} \cdot n^{-k/2}$. Thus,
    \[
    \abs{\lambda_\alpha} =  \eta^{\abs{U_\alpha} + \abs{V_\alpha}} \cdot \prod_{\circle{u} \in V(\alpha)} \abs{e(\deg(\circle{u}))} \le  \eta^{\abs{U_\alpha} + \abs{V_\alpha}} \cdot \frac{\abs{E(\alpha)}^{3\abs{E(\alpha)}}}{n^{\abs{E(\alpha)}/2}}.
    \]    
\end{proof}

\begin{corollary}\label{cor:non-spider-sum}
For $k, l \in \{0, 1, \dots , D/2\}$, let $\calB_{k,l} \subseteq \calL$ denote the set of nontrivial, non-spiders $\alpha \in \calL$ on the $(k,l)$ block i.e. $\abs{U_\alpha} = k, \abs{V_\alpha} = l$. The total norm of the non-spiders in $\calB_{k, l}$ satisfies
\[\sum_{\alpha \in \calB_{k, l}} \abs{\lambda_\alpha} \norm{M_\alpha} = \eta^{k + l} \cdot \frac{1}{n^{\Omega(\eps)}} \]
\end{corollary}
\begin{proof}
\begin{align*}
    \sum_{\alpha \in \calB_{k, l}} \abs{\lambda_\alpha} \norm{M_\alpha} & \leq \sum_{\alpha \in \calB_{k, l}}\eta^{k+l} \cdot \frac{\abs{E(\alpha)}^{3\abs{E(\alpha)}}}{n^{\abs{E(\alpha)}/2}} \norm{M_\alpha} && \text{(\cref{prop:coefficient-bound})}\\
    & \leq \eta^{k+l} \cdot\sum_{\alpha \in \calB_{k, l}}\left(\frac{\abs{E(\alpha)}^3}{n^{\Omega(\eps)}}\right)^{\abs{E(\alpha)}} && \text{(\cref{cor:non_spider_killing})}\\
    & \leq\eta^{k+l} \cdot \sum_{\alpha \in \calB_{k, l}}\left(\frac{n^{3\tau}}{n^{\Omega(\eps)}}\right)^{\abs{E(\alpha)}} && (\alpha \in \calL)\\
    & \leq \eta^{k+l} \cdot \sum_{\alpha \in \calB_{k, l}}\frac{1}{n^{\Omega(\eps\abs{E(\alpha)})}}\\
    & \leq \eta^{k+l} \cdot\sum_{i=1}^\infty \frac{n^{O(\tau i)}}{n^{\Omega(\eps i)}} && \text{(\cref{prop:edge-shape-count} and  }|E(\alpha)| \ge 1\text{ for }\alpha \in \calB_{k, l})\\
    & = \eta^{k+l} \cdot \frac{1}{n^{\Omega(\eps)}} \qedhere
\end{align*}
\end{proof}

\subsection{Killing a single spider}
\label{sec:single-spider}

We saw in the Proof Strategy section that the shape $2\beta_1 + \frac{1}{n}\beta_2$ lies in the nullspace of a moment matrix which 
satisfies the constraints ``$\ip{v}{d_u}^2 = 1$". The shape $\beta_1$ is 
exactly the kind of substructure that appears in a spider! Therefore it 
is natural to hope that if $\alpha$ is a left spider, then
$\calM_{fix}M_{\alpha} = 0$. This
doesn't quite hold because $\ip{v}{d_u}^2$ is ``missing" 
some terms: in realizations of $\alpha,$ the end vertices are required to be 
distinct from the other squares in $\alpha$, which prevents terms 
for all pairs $i,j$ from appearing in the product 
$\calM_{fix}M_\alpha$. There are smaller ``intersection terms" 
(which we call
collapses of $\alpha$) that we can add so that the end vertices are permitted to take
on all pairs $i, j$. After adding in these terms, we will produce a matrix $L$ with $\calM_{fix}L =
0$.

We first define what it means to collapse a shape into another shape
by merging two vertices. Here, we only define it for merging two
square vertices, since these are the only kind of merges that will
happen in our analysis of intersection terms.

\begin{definition}[Improper collapse]
    Let $\alpha$ be a shape and let $\square{i}, \square{j}$ be two distinct square vertices in $V(\alpha)$. We define the improper collapse of $\square{i}, \square{j}$ by:
    \begin{itemize}
        \item Remove \square{i}, \square{j} from $V(\alpha)$ and replace them by a single new vertex \square{k}.
        \item Replace each edge $\{\square{i}, \circle{u}\}$ and $\{\square{j}, \circle{u}\}$, if present, by $\{\square{k}, \circle{u}\}$, keeping the same labels (note that there may be multiedges and so the new shape may not be proper).
        \item Set $U(\square{k}) = U(\square{i}) + U(\square{j}) (\mod 2)$ and $V(\square{k}) = V(\square{i}) + V(\square{j}) (\mod 2)$.
    \end{itemize}
\end{definition}

Improper collapses have parallel edges, but we can convert them back to a sum
of proper shapes.
This is done by, for each set of parallel edges, expanding the product of Fourier characters in the Fourier basis. For example, two parallel edges with label 1 should be expanded as
\[h_1(z)^2 = (z^2-1) + 1 = h_2(z) + h_0(z)\]
\begin{definition}[Collapsing a shape]
    Let $\alpha$ be a shape with two distinct square vertices $\square{i}, \square{j}$. We say that $\beta$ is a (proper) collapse of $\square{i}, \square{j}$ if $\beta$ appears in the expansion of the improper collapse of $\square{i},\square{j}$. 
\end{definition}

\begin{remark}
    If $l_1, \dots, l_k$ are the labels of a set of parallel edges, then the product $h_{l_1}(z) \cdots h_{l_k}(z)$ is even/odd depending on the parity of $l_1 + \cdots + l_k$. Thus the nonzero Fourier coefficients will be the terms of matching parity. Therefore, in both the boolean and Gaussian cases, the shapes that are proper collapses of a given improper collapse are formed by replacing each set of parallel edges by a single edge $e$ such that $l(e) \le l_1 + \ldots + l_k$ and $l(e)~\equiv~l_1 + \cdots + l_k\pmod 2$.
\end{remark}

\begin{remark}\label{rmk:parity}
	Looking at the definition and in light of the previous remark, we have the following.
	\begin{enumerate}
		\item The number of circle vertices does not change by collapsing a shape but the number of square vertices decreases by $1$.
		\item $\alpha \in \calL$ has the property that the vertices have odd degree if and only if they are in $(U_{\alpha} \cup V_{\alpha}) \setminus (U_{\alpha} \cap V_{\alpha})$. When $\alpha$ collapses, this property is preserved.
	\end{enumerate}
\end{remark}

We now define the desired shapes $L_k$ which lie in the null space of $\calM_{fix}$.

\begin{definition}
For $k \geq 2$ define the shape $\ell_k$ on $\{\square{1}, \dots, \square{k}, \circle{1} \}$ with two edges $\{\{\square{1}, \circle{1}\}, \{\square{2}, \circle{1}\}\}$. The left side of $\ell_k$ consists of $U_{\ell_k} = \{\square{1},\dots,\square{k}\}$. The right side consists of $V_{\ell_k} =\{\square{3}, \dots, \square{k}, \circle{1}\}$.
\end{definition}

\begin{definition}\label{def:lk}
Define the ``completed'' version $L_k$ of $\ell_k$ to be the matrix which is the sum of $c_\beta M_{\beta}$ for $\beta$ being the following shapes with coefficients:
\begin{itemize}
    \item ($L_{k,1}$): $\ell_k$, with coefficient 2.
    \item ($L_{k,2}$): If $k \geq 3$, collapse $\square{1}$ and $\square{3}$ in $\ell_k$ with coefficient $\frac{2}{n}$
    \item ($L_{k,3}$): If $k \geq 4$, collapse $\square{1}$ and $\square{3}$, and collapse $\square{2}$ and $\square{4}$ in $\ell_k$ with coefficient $\frac{2}{n^2}$
    \item ($L_{k,4}$): Collapse $\square{1}$ and $\square{2}$, replacing the edges by an edge with label 2, with coefficient $\frac{1}{n}$
    \item ($L_{k,5}$): If $k \geq 3$, collapse $\square{1}, \square{2}$, and $\square{3}$, replacing the edges by an edge with label 2, with coefficient $\frac{1}{n}$.
\end{itemize}
\end{definition}

For a pictorial representation of the ribbons/shapes, see ~\cref{fig:Lk} below.

\begin{lemma}\label{lem:completed-left-side}
    $\calM_{fix} L_k = 0$
\end{lemma}
\begin{proof}
    These shapes are constructed so that if we fix a partial realization 
    of the vertices $\circle{1}$ and $\square{3}, \dots, \square{k}$ as $\circle{u} \in \calC_m$ and $S \in \binom{\calS_n}{k-2}$, the squares $\square{1}$ and $\square{2}$ can still be realized as any $j_1,j_2 \in [n]$. That is, exactly the following equality holds,
    \begin{align*}
        (\calM_{fix} L_k)_I &= \displaystyle\sum_{\substack{\circle{u} \in \calC_m,\\ S \in \binom{\calS_n}{k-2}} }\left(\sum_{\substack{j_1, j_2 \in [n]:\\ j_1 \neq j_2}} \pE[v^I v^S v_{j_1}v_{j_2}] d_{uj_1}d_{uj_2} + \sum_{j_1 \in [n]} \pE[v^Iv^Sv_{j_1}^2](d_{uj_1}^2 - 1)\right)\\
        &= \displaystyle\sum_{\substack{\circle{u} \in \calC_m,\\ S \in \binom{\calS_n}{k-2}}} \pE[v^Iv^S(\ip{v}{d_u}^2 - 1)]\\
        &= 0
    \end{align*}
    
    To demonstrate how the coefficients arise, we analyze the ribbons $R$ which $L_k$ is composed of and see how they contribute to the output.
    For pictures of the ribbons/shapes, see~\cref{fig:Lk} below. 
    Let the ribbon be partially realized as $\circle{u}$ and $S = \{\square{j_3},\dots, \square{j_k}\}$. Let $(M_{fix}L_k)_{I(u, S)}$ denote the terms in $(M_{fix}L_k)_I$ with this partial realization. In this notation we want to show
    \[(\calM_{fix}L_k)_{I(u, S)} = \sum_{\substack{j_1, j_2 \in [n]:\\ j_1 \neq j_2}} \pE[v^I v^S v_{j_1}v_{j_2}] d_{uj_1}d_{uj_2} + \sum_{j_1 \in [n]} \pE[v^Iv^Sv_{j_1}^2](d_{uj_1}^2 - 1).\]
    
    \begin{figure}[h!]
        \centering
        \includegraphics[height=10cm]{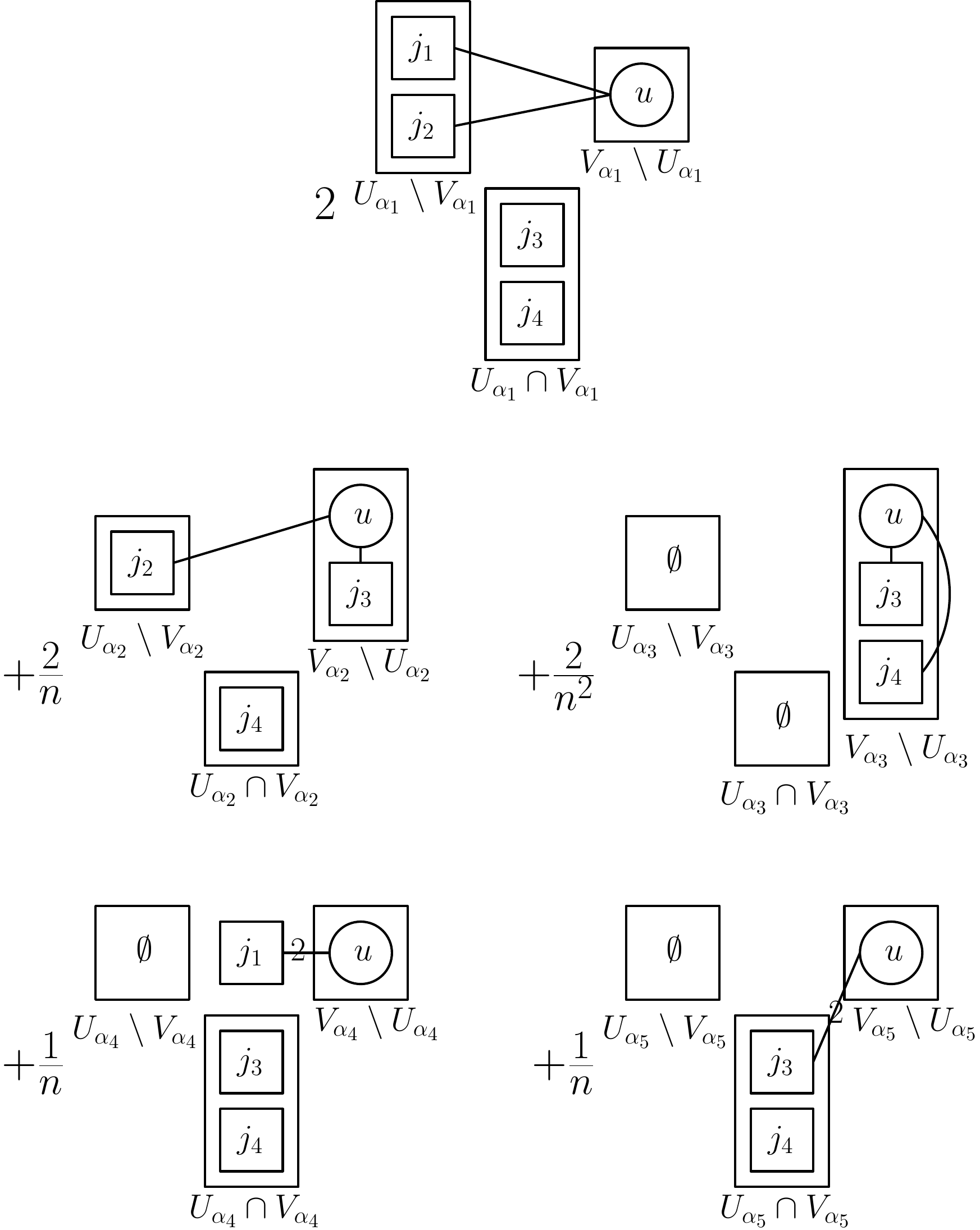}
        \caption{The five shapes that make up $L_4$.}
        \label{fig:Lk}
    \end{figure}
    
    \begin{enumerate}
        \item If we take a ribbon $R$ with $A_R = \{\square{j_1}, \dots, \square{j_k}\}$, $B_R = \{\square{j_3}, \dots, \square{j_k}\} \cup \{\circle{u}\}$ and $E(R) = \{\{\square{j_1}, \circle{u}\}, \{\square{j_2}, \circle{u}\}\}$ where $j_1 \neq j_2$ and $j_1, j_2 \notin S$ then
        \[ (\calM_{fix}M_R)_{I(u, S)} = \pE[v^Iv^Sv_{j_1}v_{j_2}]d_{uj_1}d_{uj_2}.\]
        This ribbon must ``cover'' both ordered pairs $(j_1, j_2)$ and $(j_2, j_1)$, so we want each such ribbon $R$ to appear with a coefficient of 2 in $L_k$.
        \item If we take a ribbon $R$ with $A_R = \{\square{j_1}, \dots, \square{j_k}\} \setminus \{\square{j_1}, \square{j_3}\}$, $B_R = \{\square{j_3}, \dots, \square{j_k}\} \cup \{\circle{u}\}$ and $E(R) = \{\{\square{j_3}, \circle{u}\}, \{\square{j_2}, \circle{u}\}\}$ where $j_1 = j_3 \in S$ then
        \[ (\calM_{fix}M_R)_{I(u, S)} = \pE[v^Iv^{S\setminus \{j_3\}}v_{j_2}]d_{uj_3}d_{uj_2} = n\pE[v^Iv^Sv_{j_1}v_{j_2}]d_{uj_1}d_{uj_2}.\]
        Taking a coefficient of $\frac{2}{n}$ in $L_k$ covers the two pairs $(j_1, j_2)$ and $(j_2, j_1)$ for this case of overlap with $S$.
        \item If we take a ribbon $R$ with $A_R = \{\square{j_1}, \dots, \square{j_k}\} \setminus \{\square{j_1}, \square{j_2}, \square{j_3}, \square{j_4}\}$, $B_R = \{\square{j_3}, \dots, \square{j_k}\} \cup \{\circle{u}\}$ and $E(R) = \{\{\square{j_3}, \circle{u}\}, \{\square{j_4}, \circle{u}\}\}$ where $j_1 = j_3 \in S$ and $j_2 = j_4 \in S$ then
        \[ (\calM_{fix}M_R)_{I(u, S)} = \pE[v^Iv^{S\setminus \{j_3, j_4\}}]d_{uj_3}d_{uj_4} = n^2\pE[v^Iv^Sv_{j_1}v_{j_2}]d_{uj_1}d_{uj_2}.\]
        Taking a coefficient of $\frac{2}{n^2}$ in $L_k$ covers the two pairs $(j_1, j_2)$ and $(j_2, j_1)$ for this case of overlap with $S$.
        \item If we take a ribbon $R$ with $A_R = \{\square{j_1}, \dots, \square{j_k}\}\setminus \{\square{j_1},\square{j_2}\}$, $B_R = \{\square{j_3}, \dots, \square{j_k}\} \cup \{\circle{u}\}$ and $E(R) = \{\{\square{j_1}, \circle{u}\}_2\}$ where $j_1 = j_2 \notin S$ then
        \[ (\calM_{fix}M_R)_{I(u, S)} = \pE[v^Iv^{S}](d_{uj_1}^2-1) = n\pE[v^Iv^Sv_{j_1}^2](d_{uj_1}^2-1).\]
        Taking a coefficient of $\frac{1}{n}$ in $L_k$ covers these terms.
        \item If we take a ribbon $R$ with $A_R = \{\square{j_1}, \dots, \square{j_k}\} \setminus \{\square{j_1}, \square{j_2}\}$, $B_R = \{\square{j_3}, \dots, \square{j_k}\} \cup \{\circle{u}\}$ and $E(R) = \{\{\square{j_3}, \circle{u}\}_2\}$ where $j_1 = j_2 =j_3\in S$ then
        \[ (\calM_{fix}M_R)_{I(u, S)} = \pE[v^Iv^{S}](d_{uj_3}^2-1) = n\pE[v^Iv^Sv_{j_1}^2](d_{uj_1}^2-1).\]
        Taking a coefficient of $\frac{1}{n}$ in $L_k$ covers these terms.
    \end{enumerate}
\end{proof}

One of the key facts about graph matrices is that multiplication of graph matrices approximately equals a new graph matrix, $M_\alpha \cdot M_\beta \approx M_{\gamma}$, where $\gamma$ is the result of gluing $V_\alpha$ with $U_\beta$ (and if $V_\alpha, U_\beta$ do not have the same number of vertices of each type, the product is zero). The error terms in the approximation are intersection terms (collapses) between the variables in $\alpha$ and $\beta$.
\begin{definition}
    Say that shapes $\alpha$ and $\beta$ are composable if $V_\alpha$ and $U_\beta$ have the same number of square and circle vertices. We say a shape $\gamma$ is a gluing of $\alpha$ and $\beta$, if the graph of $\gamma$ is the disjoint union of the graphs of $\alpha$ and $\beta$, followed by identifying $V_\alpha$ and $U_\beta$ under some type-preserving bijection, and if $U_\gamma = U_\alpha$ and $V_\gamma = V_\beta$.
\end{definition}

\begin{proposition}\label{prop:graph-matrix-multiplication}
    Let $\alpha, \beta$ be composable shapes. Assume that $V(\alpha) \setminus V_\alpha$ has only square vertices. Let $\{\gamma_i\}$ be the distinct gluings of $\alpha$ and $\beta$, and let $\widetilde{\calI}$ be the set of improper collapses of any number of squares (possibly zero) in $V(\alpha) \setminus V_\alpha$ with distinct squares in $V(\beta) \setminus U_\beta$ in any gluing $\gamma_i$. Then there are coefficients $c_\gamma$ for $\gamma \in \widetilde{\calI}$ such that
    \[ M_\alpha\cdot M_\beta = \displaystyle\sum_{\gamma \in \widetilde{\calI}} c_\gamma M_\gamma.\]
    Furthermore, the coefficients satisfy $\abs{c_\gamma} \leq
    2^{\abs{V(\alpha) \setminus V_\alpha}}\abs{V(\gamma)}^{\abs{V(\alpha) \setminus U_\alpha}}$.
\end{proposition}
\begin{proof}
    The product $M_\alpha \cdot M_\beta$ is a matrix which is a symmetric function of the inputs $(d_1, \dots, d_m)$, the space of which is spanned by the $M_\gamma$ over all possible shapes $\gamma$ (not restricted to $\widetilde{\calI}$), so there exist coefficients $c_\gamma$ if we allow all shapes $\gamma$. We need to check that $M_\alpha \cdot M_\beta$ actually lies in the span of shapes in $\widetilde{\calI}$ by showing that all ribbons in $M_\alpha \cdot M_\beta$ have shapes in $\widetilde{\calI}$. Expanding the definition,
    \[ M_\alpha \cdot M_\beta = \left(\displaystyle\sum_{R \text{ is a ribbon of shape }\alpha} M_R\right)\left(\sum_{S\text{ is a ribbon of shape }\beta} M_S\right) = \displaystyle\sum_{\substack{R \text{ is a ribbon of shape }\alpha,\\ S \text{ is a ribbon of shape }\beta}} M_R M_S.\]
    In order for $M_RM_S$ to be nonzero, we require $B_R = A_S$ as sets; $R$ may assign the labels arbitrarily inside $B_R$, resulting in different gluings of $\alpha$ and $\beta$. Fix $R$ and $S$, and let $\gamma$ be the corresponding gluing of $\alpha$ and $\beta$ for this $R$ and $S$.
    
    The matrix $M_RM_S$ has one nonzero entry; we claim that it is a Fourier character for a ribbon $T$ which is a collapse of $\gamma$. The labels of $R$ outside of $B_R$ can possibly overlap with the labels of $S$ outside of $A_S$, and naturally the shape of $T$ is the result of collapsing vertices in $\gamma$ with the same label.
    
    To bound the coefficients $c_\gamma$ that appear, it suffices to bound the coefficient on a ribbon $M_T$, which is bounded by the number of contributing ribbons $R, S$, where we say ribbons $R$ of shape $\alpha$ and $S$ of shape $\beta$ contribute to $T$ if $M_RM_S = M_T$. From $T$, we can completely recover the sets $A_R$ and $B_S$. The labels of $V(R) \setminus A_R$ must be among the labels of $T$; choose them in at most $\abs{V(\gamma)}^{\abs{V(\alpha) \setminus U_\alpha}}$ ways. This also determines $B_R =A_S$. All that remains is to determine the graph structure of $S$. Since improper collapsing doesn't lose any edges, knowing the labels of $R$ we know exactly which edges of $T$ must come from $R$ and $S$. The vertices $V(T) \setminus V(R)$ must come from $S$, as must $B_R$; pick a subset of $V(R) \setminus B_R$ to include in $2^{\abs{V(\alpha) \setminus V_\alpha}}$ ways.
\end{proof}

Let $\alpha$ be a left spider with end vertices $\square{i}, \square{j}$ which are adjacent to a circle $\circle{u}$. Recall that our goal is to argue that $\calM M_\alpha \approx 0$. To get there, we can try and factor $M_\alpha$ across the vertex separator $S = U_\alpha \cup \{\circle{u}\} \setminus \{\square{i},\square{j}\}$ which separates $\alpha$ into
\[ M_\alpha \approx L_{\abs{U_\alpha}} \cdot M_{\body(\alpha)}\]
where we have defined,
\begin{definition}
    Let $\alpha$ be a left spider with end vertices $\square{i}, \square{j}$. 
    Define $\body(\alpha)$ as the shape whose graph is $\alpha$ with $\square{i}$ and $\square{j}$ deleted and with $U_{\body(\alpha)} = U_\alpha \cup \{\circle{u}\} \setminus \{\square{i},\square{j}\}, V_{\body(\alpha)} = V_\alpha$. The definition is analogous for right spiders.
\end{definition}
Due to~\cref{lem:completed-left-side}, the right-hand side of the approximation is in the null space of $\calM$. We now formalize this approximate factorization.

\begin{definition}
	Let $\alpha$ be a spider with end vertices $\square{i}, \square{j}$. Define $\widetilde{\calI}_{\alpha}$ to be the set of shapes that can be obtained from $\alpha$ by performing at least one of the following steps:
	\begin{itemize}
	    \item Improperly collapse $\square{i}$ with a square vertex in $\alpha$
	    \item Improperly collapse $\square{j}$ with a square vertex in $\alpha$
	\end{itemize}
	Let $\calI_\alpha$ be the set of proper shapes that can be obtained via the same process but using proper collapses.
\end{definition}
In the above definition, we allow $\square{i}, \square{j}$ to collapse with two distinct squares, or to collapse together, or to both collapse with a common third vertex. For technical reasons we need to work with a refinement of $\calI_\alpha$ into two sets of shapes and use tighter bounds on coefficients of one set.
\begin{definition}
    Let $\calI_{\alpha}^{(1)}$ be the set of shapes that can be obtained from $\alpha$ by performing at least one of the following steps:
	\begin{itemize}
	    \item Collapse $\square{i}$ with a square vertex in $\body(\alpha) \setminus U_\alpha$
	    \item Collapse $\square{j}$ with a square vertex in $\body(\alpha) \setminus U_\alpha$ (distinct from $\square{i}$'s collapse if it happened)
	\end{itemize}        
    Let $\calI_\alpha^{(2)}\defeq \calI_\alpha \setminus \calI_\alpha^{(1)}$
    and define the improper versions $\widetilde{\calI}_\alpha^{(1)}, \widetilde{\calI}_\alpha^{(2)}$ analogously.
\end{definition}

\begin{lemma}\label{lem:improper-collapse}
	Let $\alpha$ be a left spider with end vertices \square{i}, \square{j}. There are coefficients $c_\beta$ for $\beta \in \widetilde{\calI}_\alpha$ such that 
	\[L_{\abs{U_\alpha}} \cdot M_{\body(\alpha)} = 2M_{\alpha} + \sum_{\beta \in \widetilde{\calI}_\alpha}c_\beta M_\beta,\]
	\[\abs{c_\beta} \leq 
	\begin{cases}
	 40\abs{V(\alpha)}^3 & \beta \in \widetilde{\calI}_\alpha^{(1)}\\
	 \frac{40\abs{V(\alpha)}^3}{n} & \beta \in \widetilde{\calI}_\alpha^{(2)}
	\end{cases}.\]
\end{lemma}
\begin{proof}
    First, we can check that the coefficient of $M_\alpha$ is 2. Only the $\ell_k$ term of $L_k$ has the full number of squares, and it has a factor of 2 in $L_k$.

    The shapes in $\widetilde{\calI}_\alpha$ are definitionally the intersection 
    terms that appear in this graph matrix product, and furthermore the shapes in
    $\widetilde{\calI}_\alpha$ are definitionally the intersection terms for the $\ell_k$ term. 
    Using~\cref{prop:graph-matrix-multiplication}, for each of the five shapes
    in $L_{\abs{U_\alpha}}$ the coefficient it contributes is bounded by 
    $4\abs{V(\alpha)}^3$. The coefficient on $\ell_k$ is 2, so the coefficients 
    for $\widetilde{\calI}_\alpha^{(1)}$ are at most $8 \abs{V(\alpha)}^3$. The 
    maximum coefficient of the other four shapes in $L_{\abs{U_\alpha}}$ is 
    $\frac{2}{n}$, so their total contribution to coefficients on 
    $\widetilde{\calI}_\alpha^{(2)}$ is at most $\frac{32\abs{V(\alpha)}^3}{n}$.
\end{proof}

We now want to turn our improper shapes into proper ones from $\calI_\alpha$. Unfortunately it is not quite true that to expand an improper shape, one can just expand each edge individually
(though this is true for improper ribbons).
There is an additional difficulty that arises due to ribbon symmetries. To see the difficulty, consider the example given in \cref{fig:ribbon-symmetry} below.

\begin{figure}[h!]
  \centering      
  \begin{tikzpicture}[scale=0.5,every node/.style={scale=0.5}]
    % improper shape
    \draw  (-6.5,1.5) rectangle node {\huge $u_1$} (-5,0);
    \draw  (5,1.5) rectangle node {\huge $v_2$} (6.5,0);
    \draw  (0,2.5) ellipse (1 and 1) node {\huge $w_1$};   
    \draw  (0,-1) ellipse (1 and 1) node {\huge $w_2$};      
    \node (v1) at (-5,0.75) {};
    \node (v3) at (1,2.5) {};
    \node (v2) at (-1,2.5) {};
    \node (v4) at (1,-1) {};
    \draw  (-7,2) rectangle (-4.5,-0.5);
    \node at (-5.5,-1.5) {\huge $U_{\alpha}$};
    \draw  (4.5,2) rectangle (7,-0.5);
    \node at (6,-1.5) {\huge $V_{\alpha}$};
    \node (v6) at (-1,-1) {};
    \node (v5) at (4.95,0.75) {};
    \draw (v3);
    \draw  plot[smooth, tension=.7] coordinates {(v3) (v5)};
    \draw  plot[smooth, tension=.7] coordinates {(v5) (v4)};
    \draw  plot[smooth, tension=.7] coordinates {(v1) (v6)};
    \draw  plot[smooth, tension=.7] coordinates {(v1) (-3,2.5) (v2)};
    \draw  plot[smooth, tension=.7] coordinates {(v1) (-2.5,1) (v2)};
    \node at (-3.5,3) {\Large $1$};
    \node at (-2,0.5) {\Large $1$};
    \node at (-3,-0.5) {\Large $2$};
    \node at (3,2) {\Large $2$};
    \node at (3,-0.5) {\Large $2$};

    % proper shape 1
    \draw  (12.5,-3.5) rectangle node {\huge $u_1$} (14,-5);
    \draw  (24,-3.5) rectangle node {\huge $v_2$} (25.5,-5);
    \draw  (19,-2.5) ellipse (1 and 1) node {\huge $w_1$};   
    \draw  (19,-6) ellipse (1 and 1) node {\huge $w_2$};      
    \node (v11) at (14,-4.25) {};
    \node (v13) at (20,-2.5) {};
    \node (v12) at (18,-2.5) {};
    \node (v14) at (20,-6) {};
    \draw  (12,-3) rectangle (14.5,-5.5);
    \node at (13.5,-6.5) {\huge $U_{\gamma_2}$};
    \draw  (23.5,-3) rectangle (26,-5.5);
    \node at (25,-6.5) {\huge $V_{\gamma_2}$};
    \node (v16) at (18,-6) {};
    \node (v15) at (23.95,-4.25) {};
    \draw (v13);
    \draw  plot[smooth, tension=.7] coordinates {(v13) (v15)};
    \draw  plot[smooth, tension=.7] coordinates {(v15) (v14)};
    \draw  plot[smooth, tension=.7] coordinates {(v11) (v16)};

    \node at (16,5.5) {\Large $2$};
    \node at (16,-5.5) {\Large $2$};
    \node at (22,-3) {\Large $2$};
    \node at (22,-5.5) {\Large $2$};

    % proper shape 2
    \draw  (12.5,5) rectangle node {\huge $u_1$} (14,3.5);
    \draw  (24,5) rectangle node {\huge $v_2$} (25.5,3.5);
    \draw  (19,6) ellipse (1 and 1) node {\huge $w_1$};   
    \draw  (19,2.5) ellipse (1 and 1) node {\huge $w_2$};      
    \node (v11) at (14,4.25) {};
    \node (v13) at (20,6) {};
    \node (v12) at (18,6) {};
    \node (v14) at (20,2.5) {};
    \draw  (12,5.5) rectangle (14.5,3);
    \node at (13.5,2) {\huge $U_{\gamma_1}$};
    \draw  (23.5,5.5) rectangle (26,3);
    \node at (25,2) {\huge $V_{\gamma_1}$};
    \node (v16) at (18,2.5) {};
    \node (v15) at (23.95,4.25) {};
    \draw (v13);
    \draw  plot[smooth, tension=.7] coordinates {(v13) (v15)};
    \draw  plot[smooth, tension=.7] coordinates {(v15) (v14)};
    \draw  plot[smooth, tension=.7] coordinates {(v11) (v16)};
    \node at (16,3) {\Large $2$};
    \node at (22,5.5) {\Large $2$};
    \node at (22,3) {\Large $2$};

    \draw  plot[smooth, tension=.7] coordinates {(v11)};
    \draw  plot[smooth, tension=.7] coordinates {(v11) (v12)};
    \node at (8.5,0.5) {\Huge $=$};
    \node at (19,0) {\Huge $+$};
    \node at (10.5,4) {\Huge \bf $2 \times$};
  \end{tikzpicture}
  \caption{A surprising equality of graph matrices.}
  \label{fig:ribbon-symmetry}
\end{figure}
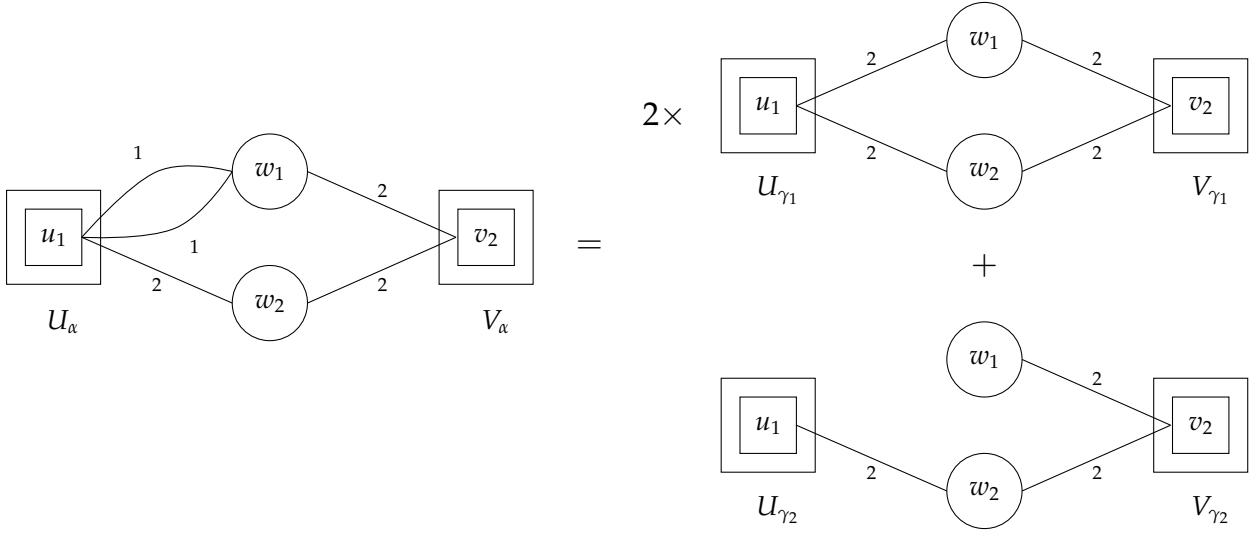

One would expect both coefficients on the right shapes to be 1 since $h_1(z)^2 = h_2(z) + h_0(z)$. However, in the left shape, the two circles are distinguishable, hence summing over all ribbons includes one with $w_1 = i, w_2 =  j$ and a second with $w_1 = j, w_2 = i$. On the top right shape, the circles are indistinguishable, hence the graph/ribbon where the circles are assigned $\{i, j\}$ is counted twice. On the bottom right shape, the circles are distinguishable, so all ribbons are summed once. To bound the new coefficients, we use the concept of shape automorphisms.

\begin{definition}
    An automorphism of a shape $\alpha$ is a function $\phi:V(\alpha) \to V(\alpha)$ that preserves the sets $U_\alpha, V_\alpha$ and is an automorphism of the underlying edge-labeled graph. Let $\aut(\alpha)$ denote the automorphism group of $\alpha$.
\end{definition}

\begin{proposition}\label{prop:expand-improper}
    Let $\alpha$ be an improper shape, and let $\calP$ be the set of proper shapes that can be obtained by expanding $\alpha$. Then there are coefficients $\abs{c_\gamma} \leq C_{Fourier}\cdot C_{Aut}$ such that
    \[M_\alpha = \displaystyle\sum_{\gamma \in \calP} c_\gamma M_\gamma\]
    where $C_{Fourier}$ is a bound on the magnitude of Fourier coefficients in the expansion and $C_{Aut} = \max_{\gamma \in \calP} \frac{\abs{\aut(\gamma)}}{\abs{\aut(\alpha)}}$.
\end{proposition}
\begin{proof}
The number of realizations of a graph matrix giving a particular ribbon is exactly the number of automorphisms, therefore
\begin{align*}
    M_\alpha &= \frac{1}{\abs{\aut(\alpha)}}\displaystyle\sum_{\text{realizations }\sigma} M_{\sigma(\alpha)}
\end{align*}
Expand each improper ribbon $M_{\sigma(\alpha)}$ into proper ribbons with coefficients at most $C_{Fourier}$. 
Because the realizations of $\alpha$ and any $\gamma$ are the same, this exactly sums over all $\gamma$ and all realizations of $\gamma$. The 
Fourier coefficient on each realization of $\gamma$ is the same; let it be 
$c_\gamma'$ with $\abs{c_\gamma'} \leq C_{Fourier}$. Continuing, 
\begin{align*}
    &= \displaystyle\frac{1}{\abs{\aut(\alpha)}} \sum_{\gamma \in \calP} c_\gamma'\sum_{\text{realizations }\sigma} M_{\sigma(\gamma)}\\
    &= \sum_{\gamma \in \calP} c_\gamma' \frac{\abs{\aut(\gamma)}}{\abs{\aut(\alpha)}} M_{\gamma}
\end{align*}
\end{proof}

\begin{proposition}\label{prop:hermite-product-coefficients}
Let $l_1 \leq \cdots \leq l_k \in \N$ and let $L = l_1 + \cdots + l_k$. Assume $L \ge 1$. In the Fourier expansion of $h_{l_1}(z)\cdots h_{l_k}(z)$, the maximum coefficient is bounded in magnitude by $(2L)^{L-l_k}$.
\begin{proof}
In the boolean case, the coefficient is 1. In the Gaussian case, the ``linearization coefficient'' of $h_p(z)$ in this product is given by orthogonality to be
\[\frac{\E_{z \sim \calN(0,1)}[h_{l_1}(z) \cdots h_{l_k}(z) \cdot h_p(z)]}{\E_{z \sim \calN(0,1)}[h_p^2(z)]}  = \frac{\E_{z \sim \calN(0,1)}[h_{l_1}(z) \cdots h_{l_k}(z) \cdot h_p(z)]}{p!}\]
A formula from, e.g.,~\cite[Example G (Continued)]{RotaWallstrom97} shows that $\E[h_{l_1} \cdots h_{l_k} \cdot h_p]$ equals the number of ``block perfect matchings'': perfect matchings on $l_1 + \cdots + l_k + p$ elements
divided into blocks of size $l_i$ or $p$ such that no two elements from the same block are matched. Bound the number of block perfect matchings by:
\begin{itemize}
    \item Pick a partial function from blocks $l_1, \dots, l_{k-1}$ to $[L]$ in at most $(L+1)^{L-l_k}$ ways.
    \item If this forms a valid partial matching and there are $p$ unmatched elements remaining, match them with the elements from the block of size $p$ in $p!$ ways.
\end{itemize}
Therefore the coefficient is bounded by $(L+1)^{L - l_k} \leq (2L)^{L-l_k}$.
\end{proof}
\end{proposition}

\begin{proposition}\label{prop:automorphism-ratio}
For a shape $\alpha$, let $\alpha\pm e$ denote the shape with edge $e$ added or deleted. 
Then
\[ \frac{\abs{\aut(\alpha \pm e)}}{\abs{\aut(\alpha)}}~\leq~\abs{V(\alpha)}^2.\]
\end{proposition}
\begin{proof}
We show that the two groups have a large subgroup which are equal. Consider $\aut(\alpha\pm e)$ and $\aut(\alpha)$ as group actions on the set $\binom{V(\alpha)}{2}$. Letting $G^e$ denote the stabilizer of edge $e$, observe that $\aut(\alpha\pm e)^e = \aut(\alpha)^e$. By the orbit-stabilizer lemma, the index $\abs{G : G^e}$ is equal to the size of the orbit of $e$, which is at least 1 and at most $\abs{V(\alpha)}^2$. So,
\[\frac{\abs{\aut(\alpha \pm e)}}{\abs{\aut(\alpha)}} = \frac{\abs{\aut(\alpha \pm e) : \aut(\alpha \pm e)^e}}{\abs{\aut(\alpha) : \aut(\alpha)^e}} \leq \abs{V(\alpha)}^2.\qedhere \] 
\end{proof}

\begin{lemma}\label{lem:collapse-lemma}
    If $\alpha$ is a left spider, there are coefficients ${c_\beta}$ for each $\beta \in \calI_\alpha$ such that 
	\[L_{\abs{U_\alpha}} \cdot M_{\body(\alpha)} = 2M_{\alpha} + \sum_{\beta \in {\calI}_\alpha}c_\beta M_\beta,\]
	\[\abs{c_\beta} \leq 
	\begin{cases}
	 160\abs{V(\alpha)}^7\abs{E(\alpha)}^2 & \beta \in {\calI}_\alpha^{(1)}\\
	 \frac{160\abs{V(\alpha)}^7\abs{E(\alpha)}^2}{n} & \beta \in {\calI}_\alpha^{(2)}
	\end{cases}.\]
\end{lemma}
\begin{proof}
    We express each $M_\beta, \beta \in \widetilde{\calI}_\alpha$ in~\cref{lem:improper-collapse} in terms of proper shapes. We apply~\cref{prop:expand-improper} using the following bounds on $C_{Fourier}$ and $C_{Aut}$. The only improperness in $\beta$ comes from 
    collapsing (at most) the two end vertices, which have a single incident 
    edge each. Therefore the set of labels of any parallel edges is either 
    $\{1,k\}$ or $\set{1,1,k},$ for some $k \leq \abs{E(\alpha)}$. By~\cref{prop:hermite-product-coefficients}, we have $C_{Fourier} \leq 4\abs{E(\alpha)}^2$. There are at most two extra parallel edges in $\beta$, so we have $C_{Aut} \leq \abs{V(\alpha)}^4$ using~\cref{prop:automorphism-ratio}. Therefore the coefficients increase by at most $C_{Fourier}\cdot C_{Aut} \leq 4\abs{E(\alpha)}^2\abs{V(\alpha)}^4$.
\end{proof}

\begin{corollary}\label{cor:right-spider-coefs}
    If $\alpha$ is a right spider, there are coefficients $c_\beta$ with the same bounds given in~\cref{lem:collapse-lemma} such that 
	\[M_{\body(\alpha)} \cdot L_{\abs{U_\alpha}}^\T = 2M_{\alpha} + \sum_{\beta \in {\calI}_\alpha}c_\beta M_\beta.\]
\end{corollary}

\begin{corollary}\label{cor:spider-killing}
    If $x \perp \nullspace(\calM_{fix})$ and $\alpha$ is a spider, then for some $c_\beta$ with the same bounds given in~\cref{lem:collapse-lemma},
    \[x^\top(M_\alpha - \displaystyle\sum_{\beta \in \calI_\alpha} c_\beta M_\beta) x = 0 \]
\end{corollary}
\begin{proof}
For a left spider, since
\[\calM_{fix} (2M_\alpha + \displaystyle\sum_{\beta \in \calI_\alpha} c_\beta M_\beta) = \calM_{fix} \cdot L_{\abs{U_\alpha}} \cdot M_{\alpha'} = 0\]
we are in position to use~\cref{fact:null-space}. For a right spider, the proof is analogous.
\end{proof}

\subsection{Killing all the spiders}

The strategy is to start with the moment matrix $\calM$ and apply~\cref{cor:spider-killing} repeatedly until we end up with no spiders in our decomposition. For each spider, killing it via~\cref{cor:spider-killing} leaves only intersection terms. Some of those intersection terms may themselves be smaller spiders, in which case we will apply the corollary again and again until only non-spiders remain. The difficulty during this procedure is to bound the total coefficient accumulated on each non-spider. To capture this process, we define the web of a spider $\alpha$, which will be a directed acyclic graph that will capture the spider killing process. For the sake of distinction, we will call the vertices of this graph ``nodes". 

\begin{definition}[Web of $\alpha$]
    The web $W(\alpha)$ of a spider $\alpha$ is a rooted directed acyclic graph 
    (DAG) whose nodes are shapes and whose root is $\alpha$. Each spider node 
    $\gamma$ has edges to nodes $\beta$ for each shape $\beta \in 
    \calI_{\gamma}$. 
    %For each edge $(\gamma, \beta)$, we call $\beta$ a child of 
    %$\gamma$. 
    The non-spider nodes are leaves/sinks of the DAG.
\end{definition}
\begin{remark}
  The DAG structure arises because each shape in $\calI_\gamma$ has strictly fewer square vertices than $\gamma$ for any spider $\gamma$. As a consequence, the height of a web $W(\alpha)$ is at most $\abs{V(\alpha)}$.
\end{remark}
Each node $\gamma$ of $W(\alpha)$ also has an associated value $v_\gamma$, which is defined by the following process:
\begin{itemize}
	\item Initially, set $v_\alpha = 1$ and for all other $\gamma$, set $v_\gamma = 0$.
	\item Starting from the root and in topological order, each spider node $\gamma$ adds $v_{\gamma} c_\beta$ to $v_\beta$ for each child $\beta \in \calI_{\gamma}$, where the $c_\beta$ are the coefficients from~\cref{cor:spider-killing}.
\end{itemize}

\begin{proposition}\label{prop:web-sum}
	If $x \perp \nullspace(\calM_{fix})$, then 
	\[\displaystyle x^\T(M_\alpha - \sum_{\text{leaves } \gamma\text{ of }W(\alpha)} v_\gamma M_\gamma)x = 0.\]
\end{proposition}
\begin{proof}
	Start with the equation $x^\T M_{\alpha} x = x^\T 
	v_{\alpha}M_{\alpha} x$. In each step, we take the topologically first spider $\gamma$, which in this case means the spider closest to the root of $W(\alpha)$, that is present in the right hand side of our equation and using \cref{cor:spider-killing}, we 
	replace $v_{\gamma}M_{\gamma}$ by $\sum_{\beta \in 
	\text{children}(\gamma)} v_\gamma c_\beta M_\beta$. 
	Precisely by the definition of the $v_{\gamma}$, this 
	process ends with the equation
	\[\displaystyle x^\T M_\alpha x = x^\T(\sum_{\text{leaves } \gamma\text{ of }W(\alpha)} v_\gamma M_\gamma)x\]
\end{proof}

\begin{proposition}\label{prop:web-parents}
For any node $\beta$ in $W(\alpha)$, $\abs{\parents(\beta)} \leq 4\abs{V(\alpha)}^3 \cdot \abs{E(\alpha)}^2$ where $parents(\beta)$ is the set of nodes $\gamma$ in $W(\alpha)$ such that $\beta \in \calI_{\gamma}$.
\end{proposition}
\begin{proof}
The following process covers all parent left spiders $\gamma$ which could possibly collapse their end vertices to form $\beta$. Starting from $\gamma = \beta,$
\begin{itemize}
    \item Pick a circle vertex $\circle{u} \in V(\gamma)$ to be the neighbor of the end vertices.
    \item Pick a square vertex $\square{i} \in V(\gamma)$ to be the collapse of the first end vertex. ``Uncollapse'' it by adding a new square to $U_{\gamma}$ with a single edge to $\circle{u}$ with label $1$. Flip the value of $U_{\gamma}(\square{i})$. Modify the label of $\{\square{i}, \circle{u}\}$ to any number up to $\abs{E(\alpha)}$.
    \item Pick a square vertex $\square{j} \in V(\gamma)$ to be the second end vertex. Optionally uncollapse it by adding a new square to $\gamma$ in the same way as above.
\end{itemize}
The process can be carried out in at most $\abs{V(\alpha)}^3\abs{E(\alpha)}(\abs{E(\alpha)}+1) \leq 2\abs{V(\alpha)}^3\abs{E(\alpha)}^2$ ways. We multiply by 2 to accommodate right spiders.
\end{proof}

Let us label each parent-child edge $(\gamma, \beta$) as either a ``type 1'' edge if $\beta \in \calI_{\gamma}^{(1)}$ or a ``type 2'' edge if $\beta \in \calI_{\gamma}^{(2)}$.

\begin{proposition}\label{prop:web-derivation-number}
  Let $p$ be a path in $W(\alpha)$ with $\#_1(p)$ type 1 edges and $\#_2(p)$ type 2 edges. Then $\#_1(p)~\leq~\abs{E(\alpha)}~+~2\#_2(p)$.
\end{proposition}
    
\begin{proof}
	For a shape $\gamma$, let $S_{\gamma}$ be the set of square vertices in $\gamma$. Then, $S_{\gamma} \cap W_{\gamma}$ will be the set of middle vertices of $\gamma$ which are squares.
  We claim that the quantity $\abs{\calS_\gamma \cap W_\gamma} + \abs{U_\gamma
    \setminus (U_\gamma \cap V_\gamma)} + \abs{V_\gamma \setminus
    (U_\gamma \cap V_\gamma)}$ decreases during a collapse.

Fix a pair of consecutive shapes $(\gamma, \beta)$ which form a type
1 edge. Looking at the definition of $\calI_\gamma^{(1)}$, each end vertex either 
  collapses with (1) nothing, or (2) a vertex of $W_\gamma$, or (3) a vertex from 
  $V_\gamma \setminus U_\gamma$ (if $\gamma$ is a left spider; 
  for a right spider, $U_\gamma \setminus V_\gamma$). 
  Furthermore, case (2) or (3) must occur for at least one of the end vertices and also, they do not collapse together.
  
  If case (2) occurs, then $\abs{\calS_\beta \cap W_\beta} < \abs{\calS_\gamma \cap W_\gamma}$ while $\abs{U_\beta 
  \setminus (U_\beta \cap V_\beta)} = \abs{U_\gamma \setminus (U_\gamma \cap 
  V_\gamma)}$ and $\abs{V_\beta \setminus
    (U_\beta \cap V_\beta)} = \abs{V_\gamma \setminus
    (U_\gamma \cap V_\gamma)}$. 
    If case (3) occurs, then $W_\beta = W_\gamma$ while 
  $\abs{U_\beta \setminus(U_\beta \cap V_\beta)}<\abs{U_\gamma 
  \setminus(U_\gamma \cap V_\gamma)}$ and $\abs{V_\beta \setminus(U_\beta \cap V_\beta)}<\abs{V_\gamma 
    \setminus(U_\gamma \cap V_\gamma)}$. 
    In all cases, $\abs{\calS_\beta \cap W_\beta} + \abs{U_\beta
  	\setminus (U_\beta \cap V_\beta)} + \abs{V_\beta \setminus
  	(U_\beta \cap V_\beta)} < \abs{\calS_\gamma \cap W_\gamma} + \abs{U_\gamma
  	\setminus (U_\gamma \cap V_\gamma)} + \abs{V_\gamma \setminus
  	(U_\gamma \cap V_\gamma)}$ as desired.
  
  Now we bound this expression for $\alpha$. From the definition of $\calL$, \cref{def:calL_valid_shapes}, 
  for spiders appearing in the pseudocalibration,
  the square vertices in $W_{\alpha}$, $U_\alpha \setminus (U_\alpha \cap 
  V_\alpha)$ and $V_\alpha \setminus (U_\alpha \cap 
  V_\alpha)$ have degree at least $1$ and can only be connected to circle vertices.
  Therefore their number is bounded by $\abs{E(\alpha)}$. Hence, initially 
  $\abs{\calS_\alpha \cap W_\alpha} + \abs{U_\alpha \setminus (U_\alpha \cap V_\alpha)} + \abs{V_\alpha \setminus (U_\alpha \cap V_\alpha)} \leq \abs{E(\alpha)}$.

  Finally, each type 2 edge in $p$ can only increase $\abs{\calS_\gamma \cap W_\gamma} + \abs{U_\gamma \setminus (U_\gamma \cap V_\gamma)} + \abs{V_\gamma \setminus (U_\gamma \cap V_\gamma)}$
  by at most 2. Therefore, we have the desired inequality $\#_1(p) \leq \abs{E(\alpha)} + 2\#_2(p)$.
\end{proof}
\begin{corollary}\label{cor:web-derivation-number}
    $\#_2(p) \geq \frac{\abs{p}}{3} - \frac{\abs{E(\alpha)}}{3}$.
\end{corollary}
\begin{proof}
  Plug in $\abs{p} = \#_1(p) + \#_2(p)$ and rearrange.
\end{proof}

Finally, we can bound the accumulation on each non-spider by a term which only depends on the parameters of the spider $\alpha$.
\begin{lemma}\label{lem:web-leaves}
There are absolute constants $C_1, C_2$ so that for all leaves $\gamma$ of $W(\alpha)$,
\[ \abs{v_\gamma} \leq (C_1 \cdot \abs{V(\alpha)} \cdot \abs{E(\alpha)})^{C_2 \abs{E(\alpha)}}.\]
\end{lemma}

\begin{proof}
  To bound $\abs{v_\gamma}$ we will sum the contributions of all paths $p = (\beta_0=\alpha,\dots,\beta_r=\gamma)$ in $W(\alpha)$
  starting from $\alpha$ and ending at $\gamma$. This path contributes a product of coefficients $c_\beta$ towards $v_\gamma$.
  
  \begin{remark}
  Here it is important that type 2 edges have stronger bounds on their coefficients $\abs{c_\beta} \leq  C\cdot (\abs{V(\alpha)}\abs{E(\alpha)})^{O(1)}/n \ll 1$.
  \end{remark}

  Before we proceed with the proof we establish some convenient notation and recall some facts.
  For consecutive shapes $\beta_{i-1},\beta_{i}$ (\ie $\beta_{i}$ is a child of $\beta_{i-1}$),
  we denote by $c_{\beta_i}$ the coefficient from~\cref{cor:spider-killing} applied on $\beta_{i - 1}$.
  By~\cref{prop:web-parents}, the in-degree of $W(\alpha)$ can be bounded as $B_1~\cdot~(\abs{V(\alpha)}\abs{E(\alpha)})^{B_2}$ for some constants $B_1, B_2$. Thus,
  the number of paths of length $r$ ending at $\gamma$ is at most $(B_1\abs{V(\alpha)}\abs{E(\alpha)})^{B_2 r}$. Using \cref{cor:spider-killing}, set $B_1, B_2$ large enough so that $c_{\beta_i}$ is at most $B_1 \cdot (\abs{V(\alpha)}\abs{E(\alpha)})^{B_2}$ for a type $1$ edge (resp. $B_1 \cdot (\abs{V(\alpha)}\abs{E(\alpha)})^{B_2} / n$ for a type $2$ edge).
  {\footnotesize

  \begin{align*}
      \abs{v_\gamma} &\le \sum_{r=0}^\infty \sum_{\substack{p = (\beta_0=\alpha,\dots,\beta_r=\gamma) \\ \textup{path from $\alpha$ to $\gamma$ in } W(\alpha)}} \prod_{i=1}^r \abs{c_{\beta_i}}\\
      &\le \sum_{r=0}^\infty\sum_{\substack{p = (\beta_0=\alpha,\dots,\beta_r=\gamma) \\ \textup{path from $\alpha$ to $\gamma$ in } W(\alpha)}} \left(B_1 \cdot (\abs{V(\alpha)}\abs{E(\alpha)})^{B_2} \right)^{\#_1(p)} \left(B_1 \cdot (\abs{V(\alpha)}\abs{E(\alpha)})^{B_2}/n \right)^{\#_2(p)}  & (\text{\cref{cor:spider-killing}})\\
      &\le \sum_{r=0}^\infty\sum_{\substack{p = (\beta_0=\alpha,\dots,\beta_r=\gamma) \\ \textup{path from $\alpha$ to $\gamma$ in } W(\alpha)}} \left(B_1 \cdot (\abs{V(\alpha)}\abs{E(\alpha)})^{B_2} \right)^{\abs{E(\alpha)} + 2\#_2(p)} \left(B_1 \cdot (\abs{V(\alpha)}\abs{E(\alpha)})^{B_2}/n \right)^{\#_2(p)}  & (\text{\cref{prop:web-derivation-number}})\\
      & = \sum_{r=0}^\infty\sum_{\substack{p = (\beta_0=\alpha,\dots,\beta_r=\gamma) \\ \textup{path from $\alpha$ to $\gamma$ in } W(\alpha)}} \left(B_1 \cdot (\abs{V(\alpha)}\abs{E(\alpha)})^{B_2} \right)^{\abs{E(\alpha)}} \left(B_1' \cdot (\abs{V(\alpha)}\abs{E(\alpha)})^{B_2'}/n \right)^{\#_2(p)}
  \end{align*}
  }
for some constants $B_1', B_2'$. 
We split the above sum into two sums, $r \le 3|E(\alpha)|$ and $r > 3|E(\alpha)|$. For $r \leq 3\abs{E(\alpha)}$, upper bounding the $\#_2(p)$ term by 1 and upper bounding 
  the number of paths by $(B_1\abs{V(\alpha)}\abs{E(\alpha)})^{B_2 r}$ gives a 
  bound of $(B_1''\abs{V(\alpha)}\abs{E(\alpha)})^{B_2'' \abs{E(\alpha)}}$ for some constants $B_1'', B_2''$.
  For larger $r$, we lower bound
  $\#_2(p) \geq r/9 = \abs{E(\alpha)}/3$ using~\cref{cor:web-derivation-number}. Applying the same bound on the number of paths,
  the total contribution of the terms corresponding to larger $r$ is bounded by 
  1 using the power of $n$ in the denominator (assuming $\delta, \tau$ are 
  small enough).
\end{proof}

We define the result of all this spider killing to be a new matrix $\calM^+$.
\begin{definition}
    Define the matrix $\calM^+$ as the result of killing all the spiders,
    \[\calM^+ \defeq \calM - \displaystyle\sum_{\text{spiders }\alpha} \lambda_\alpha \left( M_\alpha - \sum_{\text{leaves }\gamma \text{ of }W(\alpha)} v_\gamma M_\gamma \right)\]
\end{definition}

\subsection{Finishing the proof}\label{sec:finishing-psdness}

The final step of the proof is to argue that, after the spider killing 
process is completed, the newly created non-spider terms in $\calM^+$ also 
have small norm. Towards this, we would like to prove a statement similar 
to~\cref{cor:non_spider_killing}. In that proof, we used special
structural properties of the non-spiders in $\calL$ to 
prove that non-spiders in the pseudocalibration were negligible.
But now, the non-spiders in $\calM^+$ need not have the properties of 
$\calL$ -- for instance, there could be circle vertices of degree $2$ or 
isolated vertices. To handle the potentially larger norms, we will use 
that the coefficients of these new non-spider terms $\beta$ come with the 
coefficients $\lambda_\alpha$ of the spider terms $\alpha$ in whose web 
they lie. Since $\alpha$ has more vertices/edges than $\beta$, the power of $\frac{1}{n}$ in $\lambda_\alpha$ is larger than the ``expected pseudocalibration'' coefficient of $\eta^{\abs{U_\beta} + \abs{V_\beta}} \cdot \frac{1}{n^{\abs{E(\beta)}/2}}$.
We prove that these extra factors of $\frac{1}{n}$ are enough to overpower isolated
vertices or a smaller vertex separator using a careful
charging argument.

\begin{lemma}\label{lem:advanced-charging}
	If $\beta$ is a nontrivial non-spider and $\beta \in W(\alpha)$ for some spider $\alpha \in \calL$, then
	\[\eta^{\abs{U_\alpha} + \abs{V_\alpha}} \cdot \frac{1}{n^{\abs{E(\alpha)}/2}}\cdot n^{\frac{w(V(\beta)) - w(S_{\min}) + w(W_{iso})}{2}} \leq \eta^{\abs{U_\beta} + \abs{V_\beta}}\cdot \frac{1}{n^{\Omega(\eps\abs{E(\alpha)})}}\]
	where $S_{min}$ and $W_{iso}$ are the minimum vertex separator of $\beta$ and the set of isolated vertices of $V(\beta) \setminus (U_\beta \cup V_\beta)$ respectively.
\end{lemma}

\begin{proof}
	We start by giving the idea of the proof. Suppose we try to use the same 
	distribution scheme as in the proof of \cref{lem:charging}. It doesn't work 
	for two reasons. Firstly, the circle vertices in $\beta$ still have even 
	degree, which follows from \cref{rmk:parity}, but now, they could have 
	degrees $0$ or $2$. For the previous distribution scheme to go through, we 
	needed them to have degree at least $4$ which gave the necessary edge decay 
	to handle the norm bounds. Secondly, the square vertices can now have degree 
	$0$ hence getting no decay from the edges.
	
	The first issue is relatively easy to handle. Since $\beta$ was obtained by 
	collapsing $\alpha$, the circle vertices of degrees $0$ or $2$ in $\beta$ 
	must have had degree at least $4$ in $\alpha$ to begin with. Hence, we can 
	fix a particular sequence of collapses from $\alpha$ to $\beta$ and then 
	assume for the sake of analysis that the removed edges are still present. In 
	this case, the same charging argument as in \cref{lem:charging} would go 
	through. This is made formal by looking at the sequence of improper 
	collapses of this chain of collapses.
	
	To handle the second issue, let's analyze more carefully how degree $0$ 
	square vertices appear. Fix a sequence of collapses from $\alpha$ to $\beta$ 
	and consider a specific step where $\gamma$ collapsed to $\gamma'$ and a 
	square vertex of degree $0$ was formed. Let the two square vertices that 
	collapsed in $\gamma$ be $\square{i}, \square{j}$ and let the square vertex 
	of degree $0$ that formed in $\gamma'$ be $\square{k}$. In light of 
	\cref{rmk:parity}, since $\square{k}$ has degree $0$, it must not be in 
	$(U_{\gamma'} \cup V_{\gamma'})\setminus (U_{\gamma'} \cap V_{\gamma'})$ and 
	hence, $U_{\gamma'}(\square{k}) = V_{\gamma'}(\square{k}) = 0$ or 
	$U_{\gamma'}(\square{k}) = V_{\gamma'}(\square{k}) = 1$. But in the latter 
	case, this vertex does not contribute to norm bounds since it's in 
	$U_{\gamma'} \cap V_{\gamma'}$ so it can be safely disregarded. Note that
	it doesn't have to stay in this set since future collapses might collapse 
	this vertex, but this is not a problem as we can charge for this collapse 
	if it happens.
	
	So, assume we have $U_{\gamma'}(\square{k}) = V_{\gamma'}(\square{k}) = 0$. 
	But by the definition of collapse, at least one of $\square{i}$ or 
	$\square{j}$ must have been in $U_\gamma \setminus (U_\gamma \cap V_\gamma)$ 
	or $V_\gamma \setminus (U_\gamma \cap V_\gamma)$. Also from the definition 
	of collapse, we have $U_{\gamma'}(\square{k}) = U_{\gamma}(\square{i}) + 
	U_{\gamma}(\square{j}) (\mod 2)$ and $V_{\gamma'}(\square{k}) = 
	V_{\gamma}(\square{i}) + V_{\gamma}(\square{j}) (\mod 2)$. Putting these 
	together, we immediately get that the only way this could have happened is
	if either $\square{i}, \square{j} \in U_\gamma \setminus (U_\gamma \cap 
	V_\gamma)$ or if $\square{i}, \square{j} \in V_\gamma \setminus (U_\gamma 
	\cap V_{\gamma})$.
	
	When such a collapse happens, observe that $|U_{\gamma}| + |V_{\gamma}| \ge 
	|U_{\gamma'}| + |V_{\gamma'}| + 2$. This is precisely where the decay from 
	our normalization factor $\eta = \frac{1}{\sqrt{n}}$ kicks in. This 
	inequality means that an extra decay factor of $\eta^2 = \frac{1}{n}$ is 
	available to us when we compare to the "expected pseudocalibration" 
	coefficient of $\beta$. We will use this factor to charge the new square 
	vertex of degree $0$.
	
	We now make these ideas formal.
	
	Let $Q = U_\beta \cap V_\beta, P = (U_\beta \cup V_\beta) \setminus Q$ and 
	let $P'$ be the set of degree $1$ square vertices in $\beta$ that are not in 
	$S_{min}$. Let $s_0$ be the number of degree $0$ square vertices in 
	$V(\beta)\setminus Q$. All the square vertices outside $P' \cup Q \cup 
	S_{min}$ have degree at least $2$, let there be $s_{\ge 2}$ of them.
	
	Because of parity constraints, \cref{rmk:parity}, and because there are no 
	circle vertices in $U_{\beta} \cup V_{\beta}$, all circle vertices have even 
	degree in $\beta$. Let $c_0$ be the number of degree $0$ circle vertices in 
	$\beta$. Let $c_2, c_{\ge 4}$ be the number of degree $2$ circle vertices and
	the number of circle vertices of degree at least $4$ in $V(\beta) \setminus 
	S_{min}$ respectively. Then, we have \[n^{\frac{w(V(\beta)) - w(S_{\min}) + 
	w(W_{iso})}{2}} \le n^{\frac{|P'| + s_{\ge 2} + (1.5 - \eps)(c_2 + c_{\ge 
	4})}{2}} \cdot n^{s_0 + (1.5 - \epsilon)c_0}\]
	
	Using $\eta = \frac{1}{\sqrt{n}}$, it suffices to show
	\[\abs{E(\alpha)} + (|U_\alpha| + |V_\alpha| - |U_\beta| - |V_\beta|) \ge |P'| + s_{\ge 2} + (1.5 - \eps)(c_2 + c_{\ge 4}) + 2s_0 + 2(1.5 - \epsilon)c_0 + \Omega(\eps\abs{E(\alpha)})\]
	
	There can be many ways to collapse $\alpha$ to $\beta$, fix any one. We first use a charging argument for the degree $0$ square vertices.
	\begin{lemma}\label{lem:phantom_vertex}
		$|U_\alpha| + |V_\alpha| - |U_\beta| - |V_\beta| \ge 2s_0$
	\end{lemma}
	\begin{proof}
		In the collapse process, in each step, a vertex $\square{i} \in U_\gamma \setminus (U_\gamma \cap V_\gamma)$ or $\square{i} \in V_\gamma \setminus (U_\gamma \cap V_\gamma)$ of degree $1$ in an intermediate shape $\gamma$ collapses with another square vertex $\square{k}$. We have that $|U_\gamma| + |V_\gamma|$ decreases precisely when $\square{i}$ collapses with $\square{k} \in U_\gamma$ (resp. $\square{k} \in V_\gamma$). In either case, the quantity decreases by exactly $2$ which we allocate to this new merged vertex. Each degree $0$ square vertex in $V(\beta) \setminus Q$ must have arisen from a collapse, and hence must have had at least an additive quantity of $2$ allocated to it. This proves that $|U_\alpha| + |V_\alpha| - |U_\beta| - |V_\beta| \ge 2s_0$.
	\end{proof}

	We will now prove a structural lemma.
	\begin{lemma}\label{lem:structure}
		Any vertex $\circle{u}$ that has degree at least $2$ in $V(\beta) \setminus S_{min}$ is adjacent to at most $1$ vertex of $P'$.
	\end{lemma}
	
	\begin{proof}
		Observe that $\circle{u}$ cannot be adjacent to $3$ vertices in $P'$ because otherwise, at least $2$ of them would be in $U_\beta \setminus Q$ or in $V_\beta \setminus Q$ which means $\beta$ would be a spider which is a contradiction. If $\circle{u}$ is adjacent to $2$ vertices in $P'$, then one of them is in $U_\beta \setminus Q$ and the other is in $V_\beta \setminus Q$ respectively. Since both of these vertices are not in $S_{min}$, it follows that $\circle{u}$ is in $S_{min}$ since there is no path from $U_\beta$ to $V_\beta$ that doesn't pass through $S_{min}$. This is a contradiction. Therefore, $\circle{u}$ is adjacent to at most $1$ vertex in $P'$.
	\end{proof}
	This lemma immediately implies $|P'| \le c_2 + c_{\ge 4}$.

	To account for edges of $\alpha$ that are not in $\beta$, we let 
	$\widetilde{\beta}$ be the result of improperly collapsing $\alpha$ to 
	$\beta$; note that $\abs{E(\alpha)} = \abs{E(\widetilde{\beta})}$.
	We call the edges that disappeared when properly collapsing ``phantom'' edges.
	Let $\deg_{\widetilde{\beta}}(\square{i})$ (resp. $\deg_{\widetilde{\beta}}(\circle{u})$) denote the degree of vertex $\square{i}$ (resp. $\circle{u}$) in $\widetilde{\beta}$. Observe that any circle vertex $\circle{u}$ in $V(\beta)$ has $deg_{\widetilde{\beta}}(\circle{u}) \ge 4$.
	
	\begin{lemma}\label{lem:phantom_edge}
		$\abs{E(\alpha)} \ge |P'| + s_{\ge 2} + (1.5 - \eps)(c_2 + c_{\ge 4}) + 2(1.5 - \epsilon)c_0 + \Omega(\eps\abs{E(\alpha)})$
	\end{lemma}
	
	\begin{proof}
		We will use the following charging scheme. Each edge of $\beta$ incident on $P'$ allocates $1$ to the incident square vertex, which is in $P'$. Every other edge of $\beta$ allocates $\frac{1}{2}$ to the incident square vertex and $\frac{1}{2} - \frac{\epsilon}{10}$ to the incident circle vertex.	Each phantom edge allocates $1 - \frac{\eps}{10}$ to the incident circle vertex $\circle{u}$. So, a total of $\frac{\eps}{10}(\abs{E(\alpha)} - |P'|)$ has not been allocated.
	
		All square vertices in $P'$ have been allocated a value of $1$. And observe that all square vertices of degree at least $2$ in $\beta$ have been allocated at least $1$ from the incident edges of $\beta$, for a total value of $s_{\ge 2}$. So, the square vertices get a total allocation of at least $|P'| + s_{\ge 2}$.
	
		Consider any degree-$0$ circle vertex $\circle{u}$ in $V(\beta)$. It must be incident to at least $4$ phantom edges and hence, must be allocated at least a value of $4(1 - \frac{\eps}{10}) > 2(1.5 - \eps)$. Hence, the degree-$0$ circle vertices in $V(\beta$) have a total allocation of at least $2(1.5 - \eps)c_0$.
		
		Suppose the degree of $\circle{u}$ in $V(\beta)$ is $2$. Then, it is incident on at least $2$ phantom edges. 
		By \cref{lem:structure}, it is also adjacent to at most one vertex of $P'$ and so, must have been allocated a value of at least $2(1 - \frac{\epsilon}{10}) + (deg_{\widetilde{\beta}}(\circle{u}) - 3)(\frac{1}{2} - \frac{\eps}{10})$. This is at least $1.5 - \epsilon + \frac{\eps}{10}$.
		
		Suppose the degree of $\circle{u}$ in $V(\beta)$ is at least $4$. By \cref{lem:structure}, it is adjacent to at most one vertex of $P'$. Then it must have been allocated a value of at least $(deg_{\widetilde{\beta}}(\circle{u}) - 1)(\frac{1}{2} - \frac{\eps}{10})$.  Using $deg_{\widetilde{\beta}}(\circle{u}) \ge 4$, this is at least $1.5 - \epsilon + \frac{\eps}{10}$.
		
		This implies
		\[\abs{E(\alpha)} \ge |P'| + s_{\ge 2} + 2(1.5 - \epsilon)c_0 + (1.5 - \eps + \frac{\eps}{10})(c_2 + c_{\ge 4}) + \frac{\eps}{10}(\abs{E(\alpha)} - |P'|)\]
		Using $|P'| \le c_2 + c_{\ge 4}$ completes the proof.
	\end{proof}
	
	Adding \cref{lem:phantom_vertex} and \cref{lem:phantom_edge}, we get the result.
\end{proof}

\begin{corollary}\label{cor:non-spider-norm-bound}
	If $\beta$ is a nontrivial non-spider and $\beta \in W(\alpha)$ for some spider $\alpha \in \calL$, then
	\[\eta^{\abs{U_\alpha} + \abs{V_\alpha}} \cdot \frac{1}{n^{\abs{E(\alpha)}/2}}\norm{M_\beta} \leq \eta^{\abs{U_\beta} + \abs{V_\beta}}\cdot \frac{1}{n^{\Omega(\eps\abs{E(\alpha)})}}\]
\end{corollary}

\begin{proof}
	
	From \cref{lem:gaussian-norm-bounds}, we have 
	\[ \norm{M_\beta} \leq 2\cdot\left(\abs{V(\beta)} \cdot (1+\abs{E(\beta)}) \cdot \log(n)\right)^{C\cdot (\abs{V_{rel}(\beta)} + \abs{E(\beta)})} \cdot n^{\frac{w(V(\beta)) - w(S_{\min}) + w(W_{iso})}{2}}\]
	
	We have $\abs{V(\beta)}\cdot (1+\abs{E(\beta)}) \cdot \log(n) \le n^{O(\tau)}$. Also, $|V_{rel}(\beta)| \le 2(|E(\alpha)| + |E(\beta)|)$ since all the degree $0$ vertices in $V_{rel}(\beta)$ would have had vertices of $V_{rel}(\alpha)$ collapse into it in the chain of collapses and there are no degree $0$ vertices in $V_{rel}(\alpha)$. Finally, since $|E(\alpha)| \ge |E(\beta)|$, the factor 
	$2\cdot(\abs{V(\beta)} \cdot (1+\abs{E(\beta)}) \cdot \log(n))^{C\cdot (\abs{V_{rel}(\beta)} + \abs{E(\beta)})}$ can be absorbed into $\frac{1}{n^{\Omega(\eps\abs{E(\alpha)})}}$. The result follows from \cref{lem:advanced-charging}.
\end{proof}

\begin{proposition}\label{prop:m+-diag}
If $\beta$ is a trivial shape, $\lambda_\beta^+ = \lambda_\beta$.
\begin{proof}
    A trivial shape cannot appear in $W(\alpha)$ for any $\alpha$, since every collapse of a spider always keeps its circle vertices around.
\end{proof}
\end{proposition}

\begin{lemma}\label{lem:non-spider-psd}
	For $k,l \in \{0, 1, \dots, D/2\}$, let $\calB_{k,l}$ denote the set of nontrivial non-spiders on block $(k, l)$. Then
	\[\displaystyle\sum_{\beta \in \calB_{k,l}} \abs{\lambda_\beta^+}\norm{M_\beta} \leq \eta^{k+l} \cdot \frac{1}{n^{\Omega(\eps)}} \]
\end{lemma}
\begin{proof}
 
\begin{align*}
\displaystyle\sum_{\beta \in \calB_{k,l}}\norm{\lambda_\beta^+ M_\beta}
\leq & \sum_{\beta \in \calB_{k,l}} \abs{\lambda_\beta} \norm{M_\beta} + \sum_{\beta \in \calB_{k,l}} \sum_{\substack{\text{spiders }\alpha:\\ \beta \in W(\alpha)}} \abs{v_\beta} \abs{\lambda_\alpha} \norm{M_\beta}
\end{align*}
To bound the first term, we checked previously in~\cref{cor:non-spider-sum} that the total norm of nontrivial non-spiders appearing in the pseudocalibration (i.e. this term) is $\eta^{k+l}o_n(1)$. For the second term, via~\cref{lem:web-leaves} we have a bound on the accumulations $v_\gamma$ of one spider on one non-spider, so it is at most
\[\leq \sum_{\beta \in \calB_{k,l}}\displaystyle\sum_{\substack{\text{spiders }\alpha:\\ \beta \in W(\alpha)}}(C_1\abs{V(\alpha)} \cdot \abs{E(\alpha)})^{C_2 \abs{E(\alpha)}} \cdot \abs{\lambda_\alpha} \norm{M_\beta}.\]
Use the bound on the coefficients $\abs{\lambda_\alpha}$,~\cref{prop:coefficient-bound},
\begin{align*}
\leq  \sum_{\beta \in \calB_{k,l}}\displaystyle\sum_{\substack{\text{spiders }\alpha:\\ \beta \in W(\alpha)}}(C_1\abs{V(\alpha)} \cdot \abs{E(\alpha)})^{C_2 \abs{E(\alpha)}} \cdot\eta^{\abs{U_\alpha} + \abs{V_\alpha}}\cdot  \frac{\abs{E(\alpha)}^{3\abs{E(\alpha)}}}{n^{\abs{E(\alpha)}/2}} \cdot \norm{M_\beta}
\end{align*}
Invoking the norm bound for non-spiders which are collapses, \cref{cor:non-spider-norm-bound},
\begin{align*}
& \leq  \eta^{k+l} \cdot \sum_{\beta \in \calB_{k,l}}\displaystyle\sum_{\substack{\text{spiders }\alpha:\\ \beta \in W(\alpha)}} \left(\frac{C_1\abs{V(\alpha)} \cdot \abs{E(\alpha)}}{n^{\Omega(\eps)}}\right)^{C_2' \abs{E(\alpha)}}\\
& \leq  \eta^{k+l} \cdot \sum_{\beta \in \calB_{k,l}}\displaystyle\sum_{\substack{\text{spiders }\alpha:\\ \beta \in W(\alpha)}} \left(\frac{C_1n^\tau \cdot n^\tau}{n^{\Omega(\eps)}}\right)^{C_2' \abs{E(\alpha)}}.
\end{align*}

Bound the sum over all spiders by the sum over all shapes. By~\cref{prop:edge-shape-count}, the number of shapes with $i$ edges is $n^{O(\tau (i+1))}$. Summing by the number of edges, observe that $\abs{E(\alpha)} \geq \max(\abs{E(\beta)}, 2)$ since spiders always have at least $2$ edges.
\begin{align*}
    & \leq \eta^{k+l}\sum_{\beta \in \calB_{k,l}}\displaystyle\sum_{i=\max(\abs{E(\beta)}, 2)}^\infty 
    n^{O(\tau (i+1))} \cdot \left(\frac{C_1n^\tau \cdot n^{\tau}}{n^{\Omega(\eps)}}\right)^{C_2' i} \\
    &\leq \eta^{k+l}\sum_{\beta \in \calB_{k,l}} \frac{1}{n^{\Omega(\eps \max(\abs{E(\beta)}, 2))}}\\
    & \leq \eta^{k+l}\sum_{i=0}^\infty \frac{n^{O(\delta (i+1))}}{n^{\Omega(\eps \max(i, 2))}}\\
    & = \eta^{k+l} \cdot \frac{1}{n^{\Omega(\eps)}} \qedhere
\end{align*}
\end{proof}

\begin{corollary}\label{cor:m+-diag}
For $k \in \{0, \dots, D/2\}$, the $(k,k)$ block of $\calM^+$ has minimum singular value at least $\eta^{2k}(1 - \frac{1}{n^{\Omega(\eps)}})$, and for $k, l \in \{0, \dots, D/2\}, l \neq k$, the $(k,l)$ off-diagonal block has norm at most $\eta^{k+l} \cdot \frac{1}{n^{\Omega(\eps)}}$.
\end{corollary}
\begin{proof}
    By~\cref{prop:m+-diag} the identity matrix appears on the $(k,k)$ blocks with coefficient $\eta^{2k}$. By construction, $\calM^+$ has no spider shapes. By~\cref{lem:non-spider-psd}, the total norm of the non-spider shapes on the $(k,l)$ block is at most $\eta^{k+l}\cdot \frac{1}{n^{\Omega(\eps)}}$.
\end{proof}

\begin{theorem}
    W.h.p. $\calM_{fix} \psdgeq 0$.
\end{theorem}
\begin{proof}
    For any $x \in \nullspace(\calM_{fix})$, we of course have $x^\T \calM_{fix} x = 0$. 
    For any $x \perp \nullspace(\calM_{fix})$ with $\norm{x}_2 = 1$,
    \begin{align*}
        x^\T \calM_{fix} x & = x^\T (\calM + \calE) x\\
        &= x^\T \calM^+ x + x^\T \left(\displaystyle\sum_{\text{spiders }\alpha} \lambda_\alpha\left( \calM_\alpha - \sum_{\text{leaves }\gamma\text{ of }W(\alpha)}v_\gamma M_\gamma \right)\right)x \\
        & \qquad + x^\T \calE x\\
        &= x^\T (\calM^+ +  \calE) x && \text{(\cref{prop:web-sum})}
    \end{align*}
    Because the norm bound on $\calE$ in~\cref{lem:constraint-fixing} is significantly less than $\eta^{D} = n^{-n^{\delta}}$, the bound on the norm of each block of $\calM^+$
    in~\cref{cor:m+-diag} also applies to the blocks of $\calM^+ + \calE$. Therefore, 
    we use~\cref{lem:block-psd} to conclude $\calM^+ + \calE \psdgeq 0$ and the above expression is nonnegative.
\end{proof}

\section{Sherrington-Kirkpatrick Lower Bounds}\label{sec:sk}

Here, we prove~\cref{theo:boolean-subspace} and~\cref{theo:sk-bounds}.

Recall that in the Planted Boolean Vector problem, we wish to optimize
\[
\OPT(V) \defeq  \frac{1}{n}\max_{b \in \{\pm 1\}^n} b^\T \Pi_V b,
\]
where $V$ is a uniformly random $p$-dimensional subspace of $\mathbb{R}^n$.

\booleanSubspace*

\begin{proof}
	We wish to produce an SoS solution $\pE$ on boolean variables $b_1, \ldots, b_n$ such that $\pE[b^\T \Pi_V b] = n$.
        Instead of sampling a uniformly random $p$-dimensional subspace $V$ of $\mathbb{R}^n$, we first sample $d_1,\ldots, d_n$ i.i.d. $p$-dimensional
        Gaussian vectors from $\mathcal{N}(0,I)$, then form an $n$-by-$p$ matrix $A$ with rows $d_1,\dots,d_n$, and finally take
        $V$ to be the span of the columns of $A$. Since the columns of $A$ are isotropic i.i.d. random Gaussian vectors, we have
        that $V$ is a uniform $p$-dimension subspace\footnote{Except for a zero measure event. %\fnote{It might be good to cite some fact here.}
        } of $\mathbb{R}^n$.
	
        We will consider $V$ as the input for the Planted Boolean Vector problem
        while the vectors $d_1,\dots,d_n$ will be used to construct a pseudoexpectation operator for the Planted Affine Planes
        problem\footnote{Note that the vectors $d_u$ are not ``given" in the Planted Boolean Vector problem, though the construction of $\pE$ is not required to be algorithmic in any sense anyway.}.
        Since $n \le p^{3/2 - \Omega(\epsilon)}$, by~\cref{theo:sos-bounds}, for all $\delta \le c\epsilon$ for a constant $c > 0$, \text{w.h.p.}, there exists a degree-$n^{\delta}$ pseudoexpectation operator $\pE'$ on formal variables $v=(v_1,\dots,v_p)$ such that $\pE'[\ip{v}{d_u}^2] = 1$
        for every $u \in [n]$.
	
	Define $\pE$ by $\pE[b_u] \defeq \pE'[\ip{v}{d_u}]$ for all $u \in [n]$ and extending it to all polynomials on $\{b_u\}$ by 
	multilinearity. This is well defined because $\pE'[\ip{v}{d_u}^2] = 1$. Note that $\pE$ is a valid pseudoexpectation operator
        of the same degree as $\pE'$. Finally, observe that 
	\begin{align*}
          \frac{1}{n}\pE [b^\T \Pi_V b] = \frac{1}{n}\pE'[v^\T A^\T \Pi_V Av] = \frac{1}{n}\pE'[v^\T A^\T Av] = 1.
      \end{align*}
\end{proof}

Now we prove lower bounds for the Sherrington-Kirkpatrick problem,
using a reduction and proof due to \cite{MRX20}.  We include it here
for completeness. Recall that the SK problem is to
compute
\[
\OPT(W) \defeq \max_{x \in \{\pm 1\}^n} x^\T W x,
\]
where $W$ is sampled from $\GOE(n)$.

\SKbounds*

We will use the following standard results from random matrix theory of $\GOE(n)$.

\begin{fact}\label{fact:goe}
  Let $\lambda_1 \ge \ldots\ge \lambda_n$ be the eigenvalues of $W \sim \GOE(n)$ with corresponding normalized eigenvectors $w_1, \ldots, w_n$.
  Then,
  \begin{enumerate}
    \item  For every $p \in [n]$, the span of $w_1, \ldots, w_p$ is a uniformly random $p$-dimensional subspace of $\RR^n$ (see e.g.~\cite[Section~2]{OVW16}).\label{fact:goe:1}
    \item  W.h.p., $\lambda_{n^{0.67}} \ge (2 - \littleoh(1))\sqrt{n}$ (Corollary of Wigner's semicircle law~\cite{Wig93})
  \end{enumerate}
\end{fact}

\begin{proofof}{\cref{theo:sk-bounds}}
	Let $p = n^{0.67}$ and $W \sim \GOE(n)$. Let $\lambda_1\ge \ldots \ge \lambda_n$ be the eigenvalues of $W$ with corresponding orthonormal set of
        eigenvectors $w_1, \ldots, w_n$. By~\cref{fact:goe}, we have that $\lambda_p \ge (2 - \littleoh(1))\sqrt{n}$ and that $w_1, \ldots, w_p$ span a
        uniformly random $p$-dimensional subspace $V$ of $\RR^n$.
	
	We consider $V$ as the input of the Boolean Planted Vector problem and by~\cref{theo:boolean-subspace}, for some constant $\delta > 0$,
        \text{w.h.p.} there exists a degree-$n^{\delta}$ pseudoexpectation operator $\pE$ such that $\pE[x_i^2] = 1$ and
        $\pE[\sum_{i = 1}^p\ip{x}{w_i}^2] = \pE[x^\T \Pi_V x] = n$. Now,
	\begin{align*}
	\pE[x^\T Wx] = \pE[\sum_{i = 1}^n \lambda_i \ip{x}{w_i}^2] &\ge \lambda_p\pE[x^\T\Pi_V x] - \abs{\lambda_n} \pE[\sum_{i = p + 1}^n\ip{x}{w_i}^2]\\
	&\ge (2 - \littleoh(1))n^{3/2} - \abs{\lambda_n}\pE[\ip{x}{x} - \sum_{i = 1}^p\ip{x}{w_i}^2] = (2 - \littleoh(1))n^{3/2}.
	\end{align*}
\end{proofof}

\begin{remk}
  Using the same proof as above, we can obtain~\cref{theo:sk-bounds} even if we were only able to prove SoS lower
  bounds for Planted Affine Planes for some $m = \omega(n)$. So, pushing the value of $m$ up to $n^{3/2 - \epsilon}$, which
  is~\cref{theo:sos-bounds}, offers only a modest improvement.
\end{remk}

\section{Satisfying the Constraints Exactly}
\label{sec:exact-constraints}

After pseudocalibration, the PAP constraints ``$\ip{v}{d_u}^2 = 1$'' are not exactly
satisfied by the pseudocalibration,
but they are satisfied up to truncation error $\pE[\ip{v}{d_u}^2 -1] = n^{-\Omega(n^\tau)}$. This is enough to produce a Sherrington-Kirkpatrick solution that is \textit{almost} boolean, meaning $\pE[x_i^2] = 1 \pm n^{-\Omega(n^\tau)}$ where the pseudocalibration is truncated to degree $n^\tau$. To satisfy the constraints exactly, and produce an SK solution which is \textit{exactly} boolean, we can project the pseudocalibration operator. The goal of this section is to prove the following lemma for the PAP problem,

\begin{lemma}\label{lem:pseudoexpectation-rounding}
W.h.p. for the PAP problem there is $\pE' \in \RR^{\binom{[n]}{\leq D}}$ such that $\norm{\pE - \pE'}_2 \leq \frac{1}{n^{\Omega(n^\tau)}}$ and $\pE'$ exactly satisfies the constraints ``$\ip{v}{d_u}^2 = 1"$.
\end{lemma}
\begin{remark}
Note that $\pE'$ is syntactically guaranteed to still satisfy the constraints ``$v_i^2 = \frac{1}{n}$".
\end{remark}

\begin{corollary}\label{lem:constraint-fixing}
There is an $\binom{[n]}{\leq D/2} \times \binom{[n]}{\leq D/2}$ matrix $\calE$ with $\norm{\calE} \leq \frac{1}{n^{\Omega(n^\tau)}}$ such that the matrix $M_{fix}~\defeq~M~+~\calE$ is SoS-symmetric and exactly satisfies the constraints ``$\ip{v}{d_u}^2 = 1$".
\end{corollary}

We view the operators $\pE$ as vectors in $\RR^{\binom{[n]}{\leq D}}$. The approach we take is to define a ``check matrix" $Q$ such that $\pE$ satisfies the necessary constraints iff $\pE \in \nullspace(Q)$. When the constraints are functions of $v$ only, the matrix $Q$ would be filled with constants. Since the constraints depend on the inputs $d_u$, the matrix $Q$ is also a function of the $d_u$. This allows us to deconstruct it as a sum of graph matrices -- and in fact it is made out of graph matrices which we have seen already.

\begin{definition}
We let $Q$ be the matrix
\[Q \defeq \displaystyle\sum_{k=2}^{D}L_k^\T \]
where the matrices $L_k$ are defined in~\cref{sec:psd}.
\end{definition}

\begin{lemma}
$Q\pE = 0$ iff $\pE$ exactly satisfies the constraints ``$\ip{v}{d_u}^2 = 1$".
\end{lemma}
\begin{proof}
One can see in the proof of~\cref{lem:completed-left-side} that the entries of $Q\pE$ measure exactly the error in the constraints.
\end{proof}

The natural choice of $\pE'$ is therefore the projection of $\pE$ to the nullspace. This is defined by
\[\pE' \defeq \pE - Q^\T(QQ^\T)^+Q\pE \]
where we take the pseudo-inverse of $QQ^\T$ as it will turn out not to be invertible. 
%Note also that, for any constraint which is purely a function of $v$ and was satisfied by $\pE$, the constraint will also be satisfied by $\pE'$. This can seen by writing the constraint as a vector and taking the inner product on the right.

To prove~\cref{lem:pseudoexpectation-rounding}, we must decompose the second 
term in terms of graph matrices and show it has small norm.

As a warm-up, we end this outline by 
showing a simpler projection argument in the Planted Boolean Vector domain is 
sufficient if one just wants to satisfy the boolean constraints in the Planted Boolean Vector problem rather than the constraints of the PAP problem.\footnote{Using the translation between the two problems in~\cref{sec:sk}, this would allow us to exactly satisfy ``$\ip{v}{d_u}^2 =1"$ for the PAP problem. Unfortunately, the constraints ``$v_i^2 =\frac{1}{n}$" might be broken.}

Let $\pE_{\text{PBV}}$ be a candidate, not-yet-boolean, degree-$D$ pseudoexpectation operator for the Planted Boolean Vector problem, $D = 2\cdot n^\delta$. $\pE_{\text{PBV}}$ has an entry for each monomial $b^\alpha$, therefore it 
is $\multiset{n}{\leq D}$-dimensional. Let $Q_{bool}$ be the ``check  
matrix'' for the boolean constraints. $Q_{bool}$ has $n~\cdot~\multiset{n}{\leq D-2}$ rows. The $(i, \alpha)$ row checks $\pE[b^\alpha \cdot b_i^2] = \pE[b^\alpha]$. It has entry 1 in column $\alpha$ and entry $-1$ in column $\alpha \cup \{i, i\}$. 

\begin{lemma}
Assume that $\pE_{\text{PBV}}$ approximately satisfies the boolean constraints:
\[\pE_{\text{PBV}} [b^\alpha \cdot (b_i^2 - 1)] \leq n^{-\Omega(n^\tau)}\]
for any $b^\alpha$ with degree at most $D-2$. Then letting $\pE_{\text{PBV}}'$ be the projection to $\nullspace(Q_{bool})$, we have
\[\norm{\pE_{\text{PBV}} - \pE_{\text{PBV}}'}_2 \leq n^{-\Omega(n^\tau)}.\]
\end{lemma}
\begin{proof}
The effect of projecting $\pE$ to $\nullspace(Q_{bool})$ is to symmetrize $\pE[b^{\alpha + 2\beta}]$ across all $\beta$; average all entries $\pE[1], \pE[b_1^2], \pE[b_2^2], \pE[b_1^6b_7^4b_{10}^2]$ etc, average $\pE[b_1], \pE[b_1b_3^2], \pE[b_1b_3^4b_4^4]$ etc, and so on. One can see this because this is a linear map which fixes $\nullspace(Q_{bool})$ and takes all vectors into $\nullspace(Q_{bool})$.

By assumption, there is additive error $n^{-\Omega(n^\tau)}$ between $\pE_{\text{PBV}}[b^\alpha]$ and 
$\pE_{\text{PBV}}[b^\alpha \cdot b_i^2]$. As the size of $\beta$ is at most $D \ll n^\tau$, we still easily have $\pE_{\text{PBV}}[b^{\alpha + 2\beta}] = \pE_{\text{PBV}}[b^\alpha] \pm n^{-\Omega(n^\tau)}$ for all $\beta$. 
Therefore 
averaging these entries changes each of them by at most 
$n^{-\Omega(n^\tau)}$. Thus,
\[\norm{\pE_{\text{PBV}} - \pE_{\text{PBV}}'}_2 \leq \multichoose{n}{\leq D} \cdot \norm{\pE_{\text{PBV}} - \pE_{\text{PBV}}'}_\infty\]
\[\leq n^{O(n^{\delta})} \cdot n^{-\Omega(n^\tau)} = n^{-\Omega(n^\tau)}\]
\end{proof}

\subsection{Truncation error in the pseudocalibration}

The constraint ``$\ip{v}{d_u}^2 = 1$'' isn't exactly satisfied, but a general property of pseudocalibration is that it's satisfied up to truncation error, which is small w.h.p. We show a quantitative version of this bound.

We introduce the notation
\[\mu_{I, \alpha} \defeq \E_{\text{pl}} [v^I \chi_{\alpha}(d)] \]
where $\chi_\alpha(d) = h_\alpha(d)$ in the Gaussian case and $\chi_\alpha(d) = d^\alpha$ in the boolean case.

\begin{lemma}\label{lem:boolean-approximate-constraints}
    Let $p(d,v)$ such that $p$ is uniformly zero on the planted distribution. Let $\deg_d(p) = D$. For any $I \subseteq [n]$, the only nonzero Fourier coefficients of $\pE[v^I p]$ are those with size between $n^\tau \pm D$.

    Furthermore, the nonzero coefficients are bounded in absolute value by
    \[ M \cdot L \cdot 2^D  e^{mn}\cdot \max_I\max_{\abs{\alpha} \in n^\tau \pm 2D} \abs{\mu_{I, \alpha}} \]
    where $M$ is the number of nonzero monomials of $p$ and $L$ is the largest coefficient of $p$ (in absolute value).
\end{lemma}

\begin{proof} We divide the calculations into boolean and Gaussian cases. For each case we compute that Fourier coefficients below the truncation threshold neatly cancel and bound the coefficients at the threshold.

\noindent\textbf{(Boolean case)} 
Expand $p(d,v) = \displaystyle\sum_{\abs{J} \leq D} d^J p_J(v)$. By linearity,
\[\pE[v^Ip] =  \displaystyle\sum_{\abs{J} \leq D} d^J \pE[v^Ip_J(v)].\]
The $\alpha$-th Fourier coefficient gets a contribution from the $J$-th term equal to the $(\alpha \oplus J)$-th Fourier coefficient of $\pE[v^Ip_J(v)]$. Expand the polynomial $p_J$ in the $J$-th term,
\[ \pE[v^Ip_J(v)] = \sum_{K} c_{J,K} \pE[v^Iv^K]\]
The $(\alpha \oplus J)$-th coefficient of $\pE[v^Iv^K]$ is defined by pseudocalibration to be
\begin{equation}\label{eq:boolean-fourier-cases}
\left\{
\begin{array}{lr}
    \mu_{I+K,\alpha\oplus J} & \abs{\alpha \oplus J} \leq n^\tau\\
    0 & \abs{\alpha \oplus J} > n^\tau
\end{array}\right.
\end{equation}
For $\abs{\alpha} \leq n^\tau - D$ we are guaranteed to be in the first case. For this case the total $\alpha$-th Fourier coefficient is
\begin{align*}
    \displaystyle\sum_{\abs{J} \leq D} \sum_K c_{J,K}\mu_{I+K, \alpha \oplus J} &= \sum_{\abs{J} \leq D} \sum_K c_{J,K}\E_{\text{pl}}[v^{I}v^K d^{J \oplus \alpha}]\\
    &= \sum_{\abs{J} \leq D} \sum_K c_{J,K}\E_{\text{pl}}[v^{I}v^K d^\alpha d^{J}]\\
    &= \E_{\text{pl}}[v^{I}d^{\alpha}p(d,v)]\\
    & = 0.
\end{align*}
For $\abs{\alpha} > n^\tau + D$, we are guaranteed to be in the second case of~\cref{eq:boolean-fourier-cases}, in which case the total Fourier coefficient will also be zero. For $\abs{\alpha}$ within $D$ of the truncation parameter, some terms $J$ will not contribute their coefficients towards cancellation. We bound the Fourier coefficient for these $\alpha$,
\begin{align*}
    \abs{\displaystyle\sum_{\substack{J:\abs{J} \leq D,\\ \abs{\alpha \oplus J} \leq n^\tau}} \sum_K c_{J,K}\cdot\mu_{I+K, \alpha \oplus J}} &\leq \displaystyle\sum_{\abs{J} \leq D} \sum_K \abs{c_{J,K}\cdot\mu_{I+K, \alpha \oplus J}}\\
    &\leq M \cdot L \cdot \max_{I} \max_{\abs{\alpha} \in n^\tau \pm 2D} \abs{\mu_{I,\alpha}}.
\end{align*}

\noindent\textbf{(Gaussian case)} Expand $p(d, v) = \displaystyle\sum_{\abs{\beta} \leq D}h_\beta(d) p_\beta(v) = \sum_{\abs{\beta} \leq D}h_\beta(d) \sum_K c_{\beta, K} v^K$. The pseudoexpectation is
\begin{align*}
\pE[v^I p(d, v)] &= \displaystyle\sum_{\abs{\beta} \leq D}h_\beta(d) \pE[v^Ip_\beta(v)]\\
&= \sum_{\abs{\beta} \leq D}h_\beta(d) \sum_{K}c_{\beta,K}\pE[v^Iv^K]\\
&= \sum_{\abs{\beta} \leq D}h_\beta(d) \sum_K c_{\beta, K}\sum_{\abs{\alpha} \leq n^\tau} \mu_{I+K, \alpha}\frac{h_\alpha(d)}{\alpha!}.
\end{align*}
Let $l_{\alpha,\beta,\gamma}$ be the coefficient of $h_\gamma$ in the Hermite product $h_\alpha \cdot h_\beta$.
\begin{align*}
   \pE[v^I p(d, v)] &= \displaystyle\sum_{\abs{\beta} \leq D} \sum_K c_{\beta, K}\sum_{\abs{\alpha} \leq n^\tau} \mu_{I+K, \alpha}\sum_{\gamma} l_{\alpha,\beta,\gamma}\frac{h_\gamma(d)}{\alpha!} 
\end{align*}
In the case $\abs{\gamma} > n^\tau + D$, the coefficient of $h_\gamma(d)$ is zero because the max degree of a Hermite polynomial appearing in $h_{\alpha}\cdot h_{\beta}$ is at most $\abs{\alpha}+\abs{\beta} \leq n^\tau + D$. We show cancellations occur when $\abs{\gamma} \leq n^\tau - D$.
Moving the summations around, the coefficient of $h_\gamma$ is,
\begin{align*}
    &\displaystyle\sum_{\abs{\beta} \leq D} \sum_K c_{\beta, K}\sum_{\abs{\alpha} \leq n^\tau} \mu_{I+K, \alpha} \cdot  l_{\alpha,\beta,\gamma}\frac{1}{\alpha!}\\
    = & \sum_{\abs{\beta} \leq D} \sum_K c_{\beta, K}\sum_{\abs{\alpha} \leq n^\tau} \E_{\text{pl}}[v^Iv^K h_{\alpha}(d)] \cdot l_{\alpha,\beta,\gamma}\frac{1}{\alpha!}\\
    = & \E_{\text{pl}}v^I\sum_{\abs{\beta} \leq D} \sum_K c_{\beta, K}v^K\sum_{\abs{\alpha} \leq n^\tau} l_{\alpha,\beta,\gamma}\frac{h_{\alpha}(d) }{\alpha!}.
\end{align*}

We need an explicit formula for $l_{\alpha, \beta, \gamma}$ from~\cite[p.~92]{roman2005umbral},

\begin{proposition}\label{prop:multiply-coefficients}
\[l_{\alpha,\beta,\alpha+\beta - 2\delta} = \displaystyle\prod_{u,i}\binom{\alpha_{ui}}{\delta_{ui}}\binom{\beta_{ui}}{\delta_{ui}} \delta_{ui}! \]
\end{proposition}

\begin{proposition}\label{prop:hermite-product}
\[\sum_{\alpha} l_{\alpha,\beta,\gamma}\frac{h_{\alpha}(d) }{\alpha!} = h_{\beta}(d)\cdot \frac{h_\gamma(d)}{\gamma!}\]
\end{proposition}
\begin{proof}
Compute using~\cref{prop:multiply-coefficients}.
\end{proof}

In~\cref{prop:hermite-product}, the summation is actually finite. The largest $\alpha$ with $l_{\alpha,\beta,\gamma}$ nonzero has $\abs{\alpha}~\leq~\abs{\beta}~+~\abs{\gamma}$. Since we have $\abs{\beta}\leq D$ (the constraint only has degree $D$), as long as $\abs{\gamma} \leq n^\tau - D$, the above equality applies, in which case continuing the calculation for this case,
\begin{align*}
   \pE[v^Ip] &= \E_{\text{pl}}v^I\sum_{\abs{\beta} \leq D} \sum_K c_{\beta, K}v^K \cdot h_{\beta}(d)\cdot \frac{h_\gamma(d)}{\gamma!}\\
   &=  \E_{\text{pl}}v^I \cdot \frac{h_\gamma(d)}{\gamma!} \cdot \sum_{\abs{\beta} \leq D} \sum_K c_{\beta, K}v^K \cdot h_{\beta}(d)\\
   &= \E_{\text{pl}}v^I \cdot \frac{h_\gamma(d)}{\gamma!} \cdot p(d,v)\\
   &= 0.
\end{align*}
We now bound the coefficients that appear in the remaining case when $n^\tau - D < \abs{\gamma} \leq n^\tau + D$.
\begin{align*}
    \abs{\displaystyle\sum_{\abs{\beta} \leq D} \sum_K c_{\beta, K}\sum_{\abs{\alpha} \leq n^\tau} \mu_{I+K, \alpha} \cdot  l_{\alpha,\beta,\gamma}\frac{1}{\alpha!}} & \leq \displaystyle\sum_{\abs{\beta} \leq D} \sum_K \abs{c_{\beta, K}}\sum_{\abs{\alpha} \leq n^\tau} \abs{\mu_{I+K, \alpha}} \cdot  l_{\alpha,\beta,\gamma}\frac{1}{\alpha!}
\end{align*}
If $l_{\alpha,\beta,\gamma}> 0$ then we must have $\abs{\alpha} \geq \abs{\gamma} - \abs{\beta} \geq n^\tau - 2D$.
\begin{align*}
    &\leq \sum_{\abs{\beta} \leq D} \sum_K \abs{c_{\beta, K}} \cdot \left(\max_I \max_{\abs{\alpha} \in n^\tau \pm 2D}\abs{\mu_{I, \alpha}} \right) \sum_{\alpha}   l_{\alpha,\beta,\gamma}\frac{1}{\alpha!}
\end{align*}

\begin{proposition}
   \[\sum_{\alpha} l_{\alpha,\beta,\gamma}\frac{1}{\alpha!} = e^{mn} \prod_{u,i}\binom{\beta_{ui}}{\frac{\alpha_{ui} + \beta_{ui} -\gamma_{ui}}{2}}\]
\end{proposition}
\begin{proof}
Compute using~\cref{prop:multiply-coefficients}.
\end{proof}
Using the proposition,
\[\leq \sum_{\abs{\beta} \leq D} \sum_K \abs{c_{\beta, K}} \cdot \left(\max_I \max_{\abs{\alpha} \in n^\tau \pm 2D}\abs{\mu_{I, \alpha}} \right)  e^{mn}\prod_{u,i}\binom{\beta_{ui}}{\frac{\alpha_{ui} + \beta_{ui} -\gamma_{ui}}{2}}  \]
We can bound 
\[\prod_{u,i}\binom{\beta_{ui}}{k_{ui}} \leq \prod_{u,i} 2^{\beta_{ui}} = 2^{\abs{\beta}} \leq 2^D.\]
In total, letting $M$ be the number of nonzero coefficients in the constraint $p$ and $L$ be the largest coefficient, this Fourier coefficient is at most,
\[ M \cdot L \cdot 2^D  e^{mn}\cdot \max_I\max_{\abs{\alpha}\in n^\tau \pm 2D} \abs{\mu_{I+K,\alpha}} .\]
\end{proof}

\begin{lemma}\label{lem:approximate-constraints}
W.h.p. $\norm{Q\pE} \leq \frac{1}{n^{\Omega(n^\tau)}}$
\end{lemma}
\begin{proof}
    Via~\cref{lem:boolean-approximate-constraints} the only nonzero Fourier characters that appear in $Q\pE$ are those of size $n^\tau \pm 2$. Their coefficient in the lemma is at most 
    \begin{align*}
        & C \cdot  e^{mn}\cdot \max_I\max_{\abs{\alpha}\in n^\tau \pm 4} \abs{\mu_{I+K,\alpha}}\\
        \leq & C \cdot e^{mn} \cdot \frac{(n^{\tau} - 4)^{3(n^\tau - 4)}}{n^{(n^\tau - 4)/2}} && (\text{\cref{prop:coefficient-bound}})\\
        \leq & \frac{n^{3\tau n^\tau}}{n^{(\frac{1}{2} + o(1)) n^\tau }}
    \end{align*}
    Therefore we can express $Q\pE$ as a sum of graph matrices\footnote{Graph vectors, since $Q\pE$ is a vector.} of this size, with coefficients bounded by the above quantity. Now we bound the total norm by summing over all graphs.
    
    The number of graph matrices of this size is at most $n^{O(\tau) \cdot n^\tau}$.
    
    The norm of each term can be bounded using norm bounds by $n^{\frac{w(V(\alpha)) - w(S_{\min}) + w(W_{iso})}{2}}$ w.h.p. Note that there is no minimum vertex separator since $V = \emptyset$, and there are $O(1)$ isolated vertices when multiplying graphs in $Q$ with graphs in the pseudocalibration (which have no isolated vertices). The number of circle vertices can be bounded by $\frac{1}{4}\abs{E(\alpha)} \leq \frac{1}{4}n^\tau$. The number of square vertices can be bounded by $O(n^\delta) + \frac{1}{2}\abs{E(\alpha)} \leq 0.52n^\tau$. Therefore the norms are at most $m^{\frac{1}{8}n^\tau}~n^{0.26n^\tau + O(1)}~\leq~n^{0.49n^\tau}$. Notably this is significantly less than the denominator of the graph matrix coefficient, which is $n^{(0.5+o(1))n^\tau}$. Assuming $\delta$ and $\tau$ are small enough, the denominator is enough to overpower all terms multiplied together.
\end{proof}

\subsection{Analyzing $Q{Q^T}$}
The main theorem of this subsection is that the minimum nonzero eigenvalue of $QQ^T$ is large.

\begin{theorem}\label{lem:pseudoinverse-lower-bound}
The minimum nonzero eigenvalue of $QQ^T$ is $\frac{n^2}{2} - \tilde{O}(n\sqrt{m})$. 
\end{theorem}

\begin{proof}[Proof of~\cref{lem:pseudoexpectation-rounding} assuming~\cref{lem:pseudoinverse-lower-bound}]
\begin{align*}
    \norm{\pE - \pE'} & = \norm{Q^\T (QQ^\T)^+ Q\pE} \leq \norm{Q} \cdot \norm{(QQ^\T)^+} \cdot\norm{Q\pE}\\
    & \leq n^{O(1)} \cdot \frac{1}{n^2} \cdot \frac{1}{n^{\Omega(n^\tau)}}\\
    & = \frac{1}{n^{\Omega(n^\tau)}}
\end{align*}
\end{proof}

% Before we prove the theorem, we describe the strategy. Interestingly, the high-level strategy is similar to the overall proof of PSD-ness. Using the expression for $Q$ in terms of graph matrices, we can express 
% $QQ^\T$ as a sum of graph matrices. One of these terms will be an identity matrix which has coefficient $2n^2$. It would be ideal if the 
% remaining terms had norms $\littleoh_n(n^2)$ w.h.p., but as $QQ^\T$ has a nontrivial null space, this is not the case.

% We explicitly give a basis for the null space of $QQ^\T$; it will be a collection of vectors $\{N_t\}_{t \in \calT}$. Then we appeal 
% to~\cref{fact:null-space}: for $x \perp \nullspace(QQ^\T)$, it is equivalent to consider the matrix $QQ^\T~+~\sum_{t \in \calT}~N_tN_t^\T$. 
% This will essentially cancel the non-identity terms with $\Omega(n^2)$ norm in $QQ^\T$, showing that $QQ^\T$ is dominated by an identity matrix with coefficient $\Theta(n^2)$ on $\nullspace(QQ^\T)$.

Recall that $Q = \sum_{k}{L^T_k}$. Let us refer to the five shapes in~\cref{def:lk} as $\alpha_1$ through $\alpha_5$, and their coefficients as $c_{\alpha_i}$. Observe that the dominant part of $L_k$ is $2M_{\alpha_1}$ which has norm $\tilde{O}(n)$. The norm bounds for the other components of $L_k$ are as follows:
\begin{enumerate}
\item $\norm{c_{\alpha_2}M_{\alpha_2}}$ is $\tilde{O}\left(\frac{1}{n} \cdot \sqrt{mn}\right) = \tilde{O}\left(\frac{\sqrt{m}}{\sqrt{n}}\right)$
\item $\norm{c_{\alpha_3}M_{\alpha_3}}$ is $\tilde{O}\left(\frac{1}{n^2} \cdot n\sqrt{m}\right) = \tilde{O}\left(\frac{\sqrt{m}}{n}\right)$
\item $\norm{c_{\alpha_4}M_{\alpha_4}}$ is $\tilde{O}\left(\frac{1}{n} \cdot \sqrt{mn}\right) = \tilde{O}\left(\frac{\sqrt{m}}{\sqrt{n}}\right)$
\item $\norm{c_{\alpha_5}M_{\alpha_5}}$ is $\tilde{O}\left(\frac{1}{n} \cdot \sqrt{m}\right) = \tilde{O}\left(\frac{\sqrt{m}}{n}\right)$
\end{enumerate}

We start by analyzing $QQ^T = \sum_{k}{L^T_k{L_k}}$.

From the above, taking $\alpha = \alpha_1$, the dominant term of $L_k$ is $2M_{\alpha}$ where $U_{\alpha} = \{j_1,\ldots,j_{k}\}$, $V_{\alpha} = \{j_3,\ldots,j_{k}\} \cup \{u\}$, and $E(\alpha) = \{(j_1,u),(j_2,u)\}$. Since $\norm{M_{\alpha}}$ is $\tilde{O}(n)$ and $\norm{L_k - 2M_{\alpha}}$ is $\tilde{O}(\frac{\sqrt{m}}{\sqrt{n}})$, this implies that for each $k$, $\norm{L_{k}^T{L_k} - 4M^T_{\alpha}M_{\alpha}}$ is $\tilde{O}(\sqrt{mn})$. Thus, it is sufficient to analyze $M^T_{\alpha}M_{\alpha}$.

\begin{lemma}\label{QQTdecompositionlemma}
Taking $\alpha$ to be the shape such that $U_{\alpha} \setminus V_{\alpha} = \{j_1,j_2\}$, $U_{\alpha} \cap V_{\alpha} = \{j_3,j_4,\ldots,j_{k}\}$, $V_{\alpha} \setminus U_{\alpha} = \{i_{circ}\}$, and $E(\alpha) = \{(j_1,i),(j_2,i)\}$,
\[
M^{T}_{\alpha}M_{\alpha} = M_{\alpha_1} + M_{\alpha_2} + M_{\alpha_3} + M_{\alpha_4} + M_{\alpha_5} + M_{\alpha_6}
\] where $\alpha_1,\alpha_2,\alpha_3,\alpha_4,\alpha_5,\alpha_6$ are the following shapes. Note that $\alpha_1$ and $\alpha_2$ are improper shapes.
\begin{enumerate}
\item $U_{\alpha_1} \setminus V_{\alpha_1} = \emptyset$, $U_{\alpha_1} \cap V_{\alpha_1} = \{j_3,j_4,\ldots,j_{k}\} \cup \{i_{circ}\}$, $V_{\alpha_1} \setminus U_{\alpha_1} = \emptyset$, $V(\alpha_1) \setminus (U_{\alpha_1} \cup V_{\alpha_1}) = \{j_1,j_2\}$, and $E(\alpha_1) = \{(i,j_1),(i,j_1),(i,j_2),(i,j_2)\}$.
\item $U_{\alpha_2} \setminus V_{\alpha_2} = \{j_1\}$, $U_{\alpha_2} \cap V_{\alpha_2} = \{j_4,\ldots,j_{k}\} \cup \{i_{circ}\}$, $V_{\alpha_2} \setminus U_{\alpha_2} = \{j'_1\}$, $V(\alpha_2) \setminus (U_{\alpha_2} \cup V_{\alpha_2}) = \{j_2\}$, and $E(\alpha_2) = \{(i,j_1),(i,j'_1),(i,j_2),(i,j_2)\}$.
\item $U_{\alpha_3} \setminus V_{\alpha_3} = \{j_1,j_2\}$, $U_{\alpha_3} \cap V_{\alpha_3} = \{j_5,\ldots,j_{k}\} \cup \{i_{circ}\}$, $V_{\alpha_3} \setminus U_{\alpha_3} = \{j'_1,j'_2\}$, and $E(\alpha_3) = \{(i,j_1),(i,j_2),(i,j'_1),(i,j'_2)\}$.
\item $U_{\alpha_4} \setminus V_{\alpha_4} = \{i_{circ}\}$, $U_{\alpha_4} \cap V_{\alpha_4} = \{j_3,j_4,\ldots,j_{k}\}$, $V_{\alpha_4} \setminus U_{\alpha_4} = \{i'_{circ}\}$, $V(\alpha_4) \setminus (U_{\alpha_4} \cup V_{\alpha_4}) = \{j_1,j_2\}$, and $E(\alpha_4) = \{(i,j'_1),(i,j'_2),(i',j_1),(i',j_2)\}$.
\item $U_{\alpha_5} \setminus V_{\alpha_5} = \{j_1\} \cup \{i_{circ}\}$, $U_{\alpha_5} \cap V_{\alpha_5} = \{j_4,\ldots,j_{k}\}$, $V_{\alpha_5} \setminus U_{\alpha_5} = \{j'_1\} \cup \{i'_{circ}\}$, $V(\alpha_5) \setminus (U_{\alpha_5} \cup V_{\alpha_5}) = \{j_2\}$, and $E(\alpha_5) = \{(i,j'_1),(i,j_2),(i',j_1),(i',j_2)\}$.
\item $U_{\alpha_6} \setminus V_{\alpha_6} = \{j_1,j_2\} \cup \{i_{circ}\}$, $U_{\alpha_6} \cap V_{\alpha_6} = \{j_5,\ldots,j_{k}\}$, $V_{\alpha_6} \setminus U_{\alpha_6} = \{j'_1,j'_2\}  \cup \{i'_{circ}\}$, and $E(\alpha_6) = \{(i',j_1),(i',j_2),(i,j'_1),(i,j'_2)\}$.
\end{enumerate}
\end{lemma}
For pictures of these shapes, see~\cref{fig:QQTdominantfigure} below.
\begin{figure}[ht]
\centerline{\includegraphics[height=10cm]{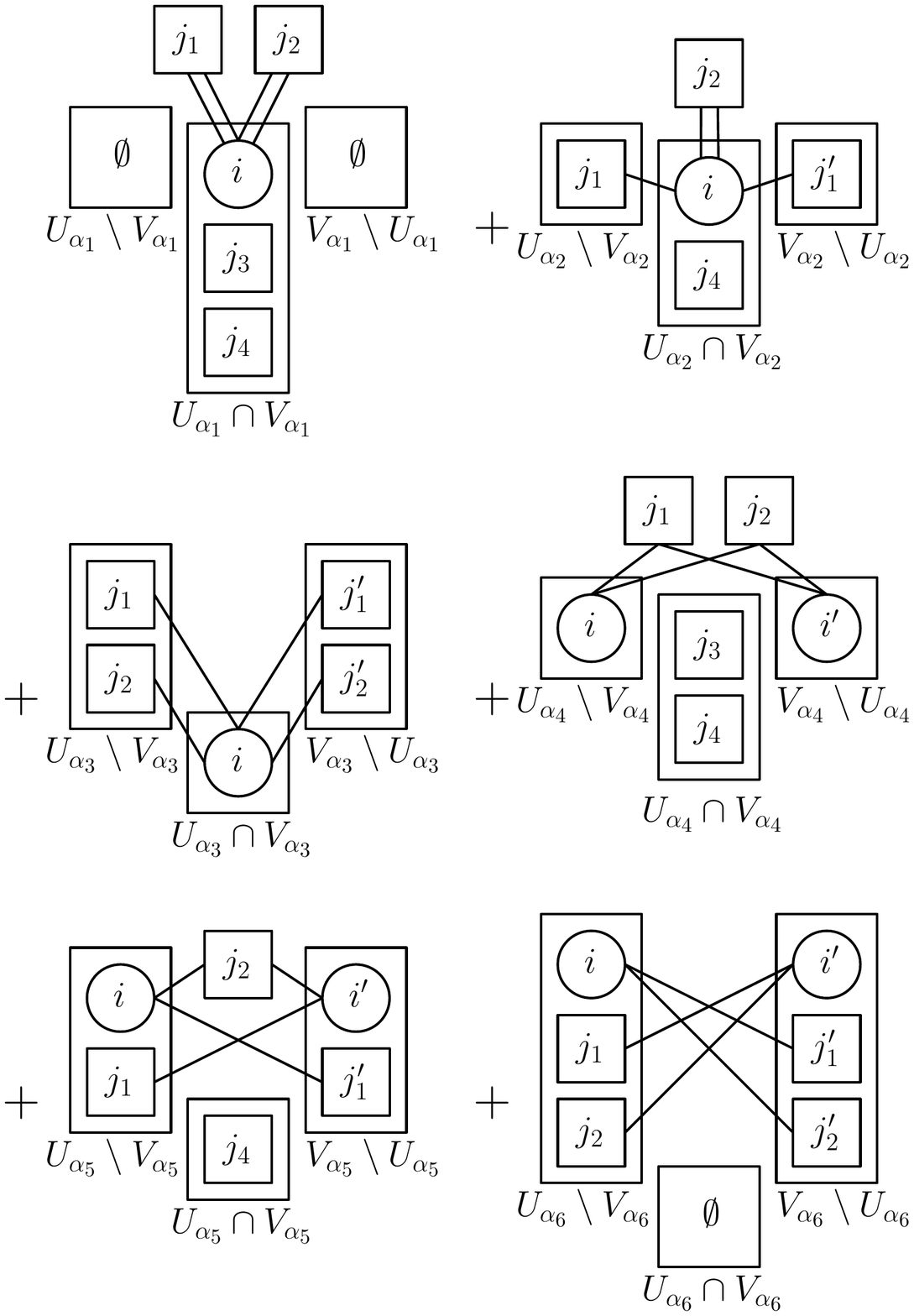}}
\caption{This figure shows the decomposition of $M^T_{\alpha}M_{\alpha}$.}
\label{fig:QQTdominantfigure}
\end{figure}
\begin{remark}
For $k = 2$, only shapes $\alpha_1$ and $\alpha_4$ are present and for $k = 3$, only shapes $\alpha_1,\alpha_2,\alpha_4,\alpha_5$ are present.
\end{remark}
\begin{proof}[Proof of Lemma \ref{QQTdecompositionlemma}]
We compute $M_{\alpha^T}M_{\alpha}$ by considering the ribbons which appear in $M_{\alpha^T}M_{\alpha}$.
\begin{enumerate}
\item Each ribbon $R$ with $A_R \setminus B_R = \emptyset$, $A_R \cap B_R = \{j_3,j_4,\ldots,j_{k}\} \cup \{i_{circ}\}$, $B_R \setminus A_R = \emptyset$, $V(R) \setminus (A_R \cup B_R) = \{j_1,j_2\}$, and $E(R) = \{(i,j_1),(i,j_1),(i,j_2),(i,j_2)\}$ appears in exactly one way as the composition of the ribbons $R_1$ and $R_2$ where 
\begin{enumerate}
\item $A_{R_1} \setminus B_{R_1} = \{i_{circ}\}$, $A_{R_1} \cap B_{R_1} = \{j_3,j_4,\ldots,j_{k}\}$, $B_{R_1} \setminus A_{R_1} = \{j_1,j_2\}$, and $E(R_1) = \{(i,j_1),(i,j_2)\}$.
\item $A_{R_2} \setminus B_{R_2} = \{j_1,j_2\}$, $A_{R_2} \cap B_{R_2} = \{j_3,j_4,\ldots,j_{k}\}$, $B_{R_2} \setminus A_{R_2} = \{i_{circ}\}$, and $E(R_2) = \{(i,j_1),(i,j_2)\}$.
\end{enumerate}
\item Each ribbon $R$ with $A_R \setminus B_R = \{j_1\}$, $A_R \cap B_R = \{j_4,\ldots,j_{k}\} \cup \{i_{circ}\}$, $B_R \setminus A_R = \{j'_1\}$, $V(R) \setminus (A_R \cup B_R) = \{j_2\}$, and $E(R) = \{(i,j_1),(i,j'_1),(i,j_2),(i,j_2)\}$ appears in exactly one way as the composition of the ribbons $R_1$ and $R_2$ where 
\begin{enumerate}
\item $A_{R_1} \setminus B_{R_1} = \{i_{circ}\}$, $A_{R_1} \cap B_{R_1} = \{j_1,j_4,\ldots,j_{k}\}$, $B_{R_1} \setminus A_{R_1} = \{j'_1,j_2\}$, and $E(R_1) = \{(i,j'_1),(i,j_2)\}$.
\item $A_{R_2} \setminus B_{R_2} = \{j_1,j_2\}$, $A_{R_2} \cap B_{R_2} = \{j'_1,j_4,\ldots,j_{k}\}$, $B_{R_2} \setminus A_{R_2} = \{i_{circ}\}$, and $E(R_2) = \{(i,j_1),(i,j_2)\}$.
\end{enumerate}
\item Each ribbon $R$ with $A_R \setminus B_R = \{j_1,j_2\}$, $A_R \cap B_R = \{j_5,\ldots,j_{k}\} \cup \{i_{circ}\}$, $B_R \setminus A_R = \{j'_1,j'_2\}$, $V(R) \setminus (A_R \cup B_R) = \emptyset$, and $E(R) = \{(i,j_1),(i,j'_1),(i,j_2),(i,j'_2)\}$ appears in exactly one way as the composition of the ribbons $R_1$ and $R_2$ where 
\begin{enumerate}
\item $A_{R_1} \setminus B_{R_1} = \{i_{circ}\}$, $A_{R_1} \cap B_{R_1} = \{j_1,j_2,j_5,\ldots,j_{k}\}$, $B_{R_1} \setminus A_{R_1} = \{j'_1,j'_2\}$, and $E(R_1) = \{(i,j'_1),(i,j'_2)\}$.
\item $A_{R_2} \setminus B_{R_2} = \{j_1,j_2\}$, $A_{R_2} \cap B_{R_2} = \{j'_1,j'_2,j_5,\ldots,j_{k}\}$, $B_{R_2} \setminus A_{R_2} = \{i_{circ}\}$, and $E(R_2) = \{(i,j_1),(i,j_2)\}$.
\end{enumerate}
\item Each ribbon $R$ with $A_R \setminus B_R = \{j_1,j_2\} \cup  \{i_{circ}\}$, $A_R \cap B_R = \{j_5,\ldots,j_{k}\}$, $B_R \setminus A_R = \{j'_1,j'_2\}  \cup  \{i'_{circ}\}$, $V(R) \setminus (A_R \cup B_R) = \emptyset$, and $E(R) = \{(i,j'_1),(i,j'_2),(i',j_1),(i',j_2)\}$ appears in exactly one way as the composition of the ribbons $R_1$ and $R_2$ where 
\begin{enumerate}
\item $A_{R_1} \setminus B_{R_1} = \{i_{circ}\}$, $A_{R_1} \cap B_{R_1} = \{j_1,j_2,j_5,\ldots,j_{k}\}$, $B_{R_1} \setminus A_{R_1} = \{j'_1,j'_2\}$, and $E(R_1) = \{(i,j'_1),(i,j'_2)\}$.
\item $A_{R_2} \setminus B_{R_2} = \{j_1,j_2\}$, $A_{R_2} \cap B_{R_2} = \{j'_1,j'_2,j_5,\ldots,j_{k}\}$, $B_{R_2} \setminus A_{R_2} = \{i'_{circ}\}$, and $E(R_2) = \{(i,j_1),(i,j_2)\}$.
\end{enumerate}
\item Each ribbon $R$ with $A_R \setminus B_R = \{i_{circ}\}$, $A_R \cap B_R = \{j_3,j_4,\ldots,j_{k}\}$, $B_R \setminus A_R = \{i'_{circ}\}$, $V(R) \setminus (A_R \cup B_R) = \{j_1,j_2\}$, and $E(R) = \{(i,j_1),(i,j_2),(i',j_1),(i',j_2)\}$ appears in exactly one way as the composition of the ribbons $R_1$ and $R_2$ where 
\begin{enumerate}
\item $A_{R_1} \setminus B_{R_1} = \{i_{circ}\}$, $A_{R_1} \cap B_{R_1} = \{j_3,j_4,\ldots,j_{k}\}$, $B_{R_1} \setminus A_{R_1} = \{j_1,j_2\}$, and $E(R_1) = \{(i,j_1),(i,j_2)\}$.
\item $A_{R_2} \setminus B_{R_2} = \{j_1,j_2\}$, $A_{R_2} \cap B_{R_2} = \{j_3,j_4,\ldots,j_{k}\}$, $B_{R_2} \setminus A_{R_2} = \{i'_{circ}\}$, and $E(R_2) = \{(i',j_1),(i',j_2)\}$.
\end{enumerate}
\item Each ribbon $R$ with $A_R \setminus B_R = \{j_1\}$, $A_R \cap B_R = \{j_4,\ldots,j_{k}\} \cup \{i_{circ}\}$, $B_R \setminus A_R = \{j'_1\}$, $V(R) \setminus (A_R \cup B_R) = \{j_2\}$, and $E(R) = \{(i,j_1),(i,j'_1),(i,j_2),(i,j_2)\}$ appears in exactly one way as the composition of the ribbons $R_1$ and $R_2$ where 
\begin{enumerate}
\item $A_{R_1} \setminus B_{R_1} = \{i_{circ}\}$, $A_{R_1} \cap B_{R_1} = \{j_1,j_4,\ldots,j_{k}\}$, $B_{R_1} \setminus A_{R_1} = \{j'_1,j_2\}$, and $E(R_1) = \{(i,j'_1),(i,j_2)\}$.
\item $A_{R_2} \setminus B_{R_2} = \{j_1,j_2\}$, $A_{R_2} \cap B_{R_2} = \{j'_1,j_4,\ldots,j_{k}\}$, $B_{R_2} \setminus A_{R_2} = \{i_{circ}\}$, and $E(R_2) = \{(i,j_1),(i,j_2)\}$.
\end{enumerate}
\end{enumerate}
Based on these cases, we have that $M_{\alpha^T}M_{\alpha} = M_{\alpha_1} + M_{\alpha_2} + M_{\alpha_3} + M_{\alpha_4} + M_{\alpha_5} + M_{\alpha_6}$.
\end{proof}
We now analyze each of the matrices $M_{\alpha_1}, M_{\alpha_2}, M_{\alpha_3}, M_{\alpha_4}, M_{\alpha_5}, M_{\alpha_6}$.
\begin{lemma}
Taking $Id_{k-2,1}$ to be the shape where $U_{Id_{k-2,1}} \setminus V_{Id_{k-2,1}} = \emptyset$, $U_{Id_{k-2,1}} \cap V_{Id_{k-2,1}} = \{j_3,j_4,\ldots,j_{k}\} \cup \{i_{circ}\}$, $V_{Id_{k-2,1}} \setminus U_{Id_{k-2,1}} = \emptyset$, and $E(Id_{k-2,1}) = \emptyset$, $\norm{M_{\alpha_1} - \binom{n-k+2}{2}M_{Id_{k-2,1}}}$ is $\tilde{O}\left(n^{\frac{3}{2}}\right)$.
\end{lemma}
\begin{proof}
To convert an improper shape $\alpha_1$ to a sum of proper shapes, we take each ribbon of shape $\alpha_1$ and decompose it into a sum of proper ribbons. Decomposing $M_{\alpha_1}$ in this way, each ribbon $R$ with $A_R \setminus B_R = \emptyset$, $A_R \cap B_R = \{j_3,j_4,\ldots,j_{k}\} \cup \{i_{circ}\}$, $B_R \setminus A_R = \emptyset$, and $E(R) = \emptyset$ appears $\binom{n-k+2}{2}$ times, once for each pair $j_1,j_2$ such that $j_1 < j_2$ and $j_1,j_2 \notin \{j_3,\ldots,j_{k}\}$. The other ribbons which arise all have an edge with label $2$ incident with $j_1$ or $j_2$ and thus the resulting terms have norm $\tilde{O}\left(n^{\frac{3}{2}}\right)$.
\end{proof}
\begin{definition}
Define $\beta_2$ to be the shape such that $U_{\beta_2} \setminus V_{\beta_2} = \{j_1\}$, $U_{\beta_2} \cap V_{\beta_2} = \{j_4,\ldots,j_{k}\} \cup \{i_{circ}\}$, $V_{\beta_2} \setminus U_{\beta_2} = \emptyset$, and $E(\beta_1) = \{(i,j_1)\}$.
\end{definition}
\begin{lemma}
$\norm{M_{\alpha_2} - (n-k+1)M_{\beta_2}M^T_{\beta_2}}$ is $\tilde{O}\left(n^{\frac{3}{2}}\right)$.
\end{lemma}
\begin{proof}
Again, to convert an improper shape $\alpha_2$ to a sum of proper shapes, we take each ribbon of shape $\alpha_2$ and decompose it into a sum of proper ribbons. Decomposing $M_{\alpha_2}$ in this way, each ribbon $R$ with $A_R \setminus B_R = \{j_1\}$, $A_R \cap B_R = \{j_4,\ldots,j_{k}\} \cup \{i_{circ}\}$, $B_R \setminus A_R = \emptyset$, and $E(R) = \{(i,j_1),(i,j'_1)\}$ appears $(n-k+1)$ times, once for each $j_2 \in [n] \setminus \{j_1,j'_1,j_4,\ldots,j_k\}$. The other ribbons which arise have an edge with label $2$ incident with $j_2$ and thus the resulting terms have norm $\tilde{O}\left(n^{\frac{3}{2}}\right)$. This implies that if we take $\alpha'_2$ to be the shape where $U_{\alpha_2} \setminus V_{\alpha_2} = \{j_1\}$, $U_{\alpha_2} \cap V_{\alpha_2} = \{j_4,\ldots,j_{k}\} \cup \{i_{circ}\}$, $V_{\alpha_2} \setminus U_{\alpha_2} = \{j'_1\}$, and $E(\alpha_2) = \{(i,j_1),(i,j'_1)\}$ then $\norm{M_{\alpha_2} - (n-k+1)M_{\alpha'_2}}$ is $\tilde{O}\left(n^{\frac{3}{2}}\right)$.

$\norm{M_{\alpha'_2}}$ is $\tilde{O}(n)$, so this term cannot be ignored. To handle this, we observe that $M_{\alpha'_2}$ is approximately equal to a PSD matrix. More precisely, $\norm{M_{\alpha'_2} - M_{\beta_2}M^T_{\beta_2}}$ is $\tilde{O}(1)$. To see this, note that when we expand out $M_{\beta_2}M^T_{\beta_2}$, the ribbons which result when there are no collisions give $M_{\alpha'_2}$ and for each ribbon $R$ which results from a collision, $A_R \setminus B_R = B_R \setminus A_R = \emptyset$ so the resulting terms have norm $\tilde{O}(1)$. 
\end{proof}
\begin{figure}[ht]\label{BetaTwoFigure}
\centerline{\includegraphics[height=3cm]{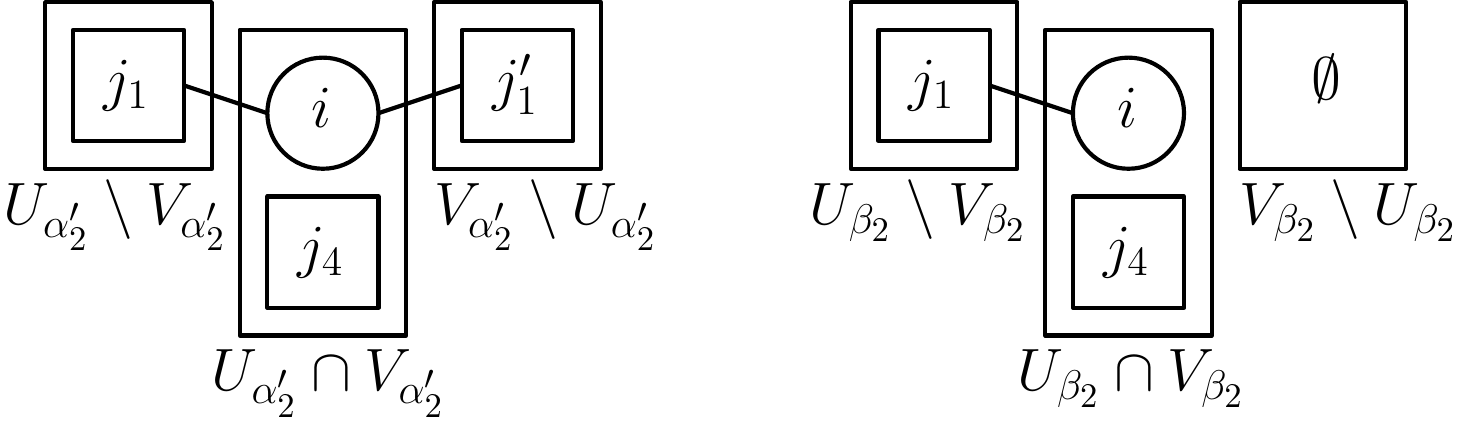}}
\caption{This figure shows $\alpha'_2$ and $\beta_2$ for $k = 4$.}
\end{figure}
$\norm{M_{\alpha_3}}$ is $\tilde{O}(n^2)$, so this term cannot be ignored. To handle this, we observe that $M_{\alpha_3}$ is approximately equal to a PSD matrix. More precisely, we have the following lemma.
\begin{definition}
Define $\beta_3$ to be the shape such that $U_{\beta_3} \setminus V_{\beta_3} = \{j_1,j_2\}$, $U_{\beta_3} \cap V_{\beta_3} = \{j_5,\ldots,j_{k}\} \cup \{i_{circ}\}$, $V_{\beta_3} \setminus U_{\beta_3} = \emptyset$, and $E(\beta_3) = \{(i,j_1),(i,j_2)\}$.
\end{definition}
\begin{lemma}
$\norm{M_{\alpha_3} - M_{\beta_3}M^T_{\beta_3}}$ is $\tilde{O}(n)$.
\end{lemma}
\begin{proof}
To see this, note that when we expand out $M_{\beta_3}M^T_{\beta_3}$, the ribbons which result when there are no collisions give $M_{\alpha_3}$ and for each ribbon $R$ which results from a collision, $|A_R \setminus B_R| = |B_R \setminus A_R| \leq 1$ so the resulting terms have norm $\tilde{O}(n)$. 
\end{proof}
\begin{figure}[ht]\label{BetaThreeFigure}
\centerline{\includegraphics[height=3cm]{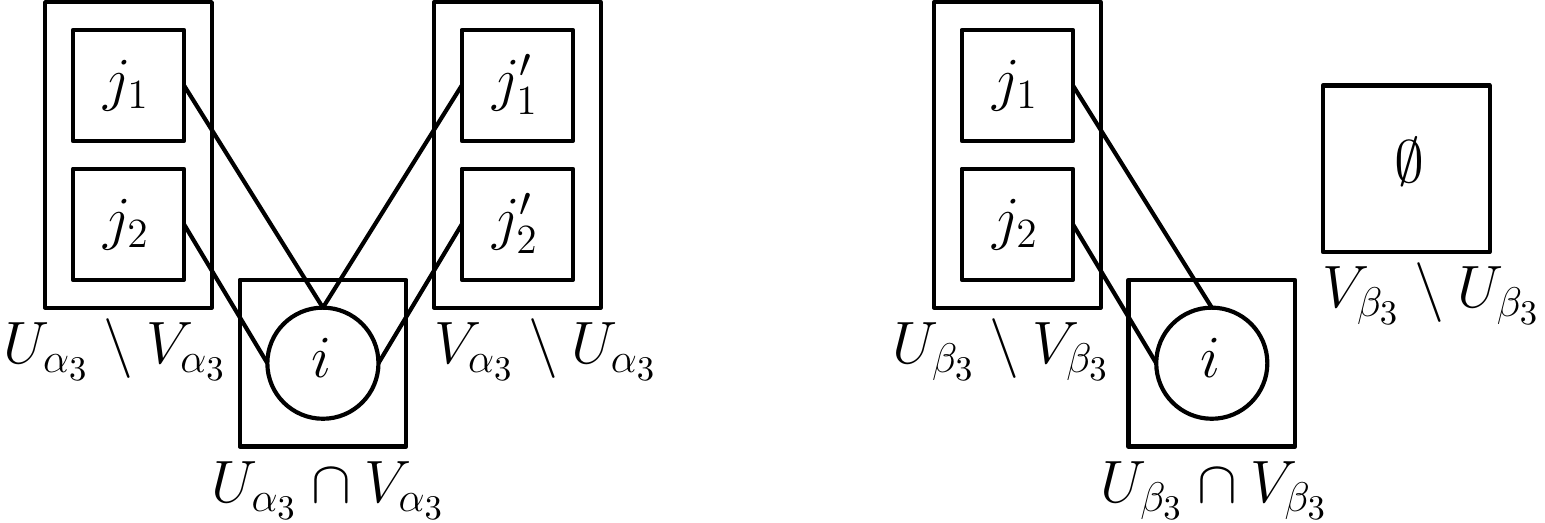}}
\caption{This figure shows $\alpha_3$ and $\beta_3$ for $k = 4$.}
\end{figure}
We now consider the norms of $M_{\alpha_4}$, $M_{\alpha_5}$, and $M_{\alpha_6}$.
\begin{enumerate}
\item $\norm{M_{\alpha_4}}$ is $\tilde{O}(n\sqrt{m})$.
\item $\norm{M_{\alpha_5}}$ is $\tilde{O}(n\sqrt{m})$.
\item $\norm{M_{\alpha_6}}$ is $\tilde{O}(n^2)$.
\end{enumerate}
This means that $M_{\alpha_4}$ and $M_{\alpha_5}$ can be ignored but $M_{\alpha_6}$ cannot be ignored. In fact, there is a very good reason for this. In particular, for $k \geq 4$, $L_k$ has a non-trivial nullspace $N_k$, so we cannot show that the minimum nonzero eigenvalue of ${L_k^T}{L_k}$ is large without taking this nullspace into account. We handle this nullspace $N_k$ in the next two subsubsections.

Putting everything together, we have the following corollary:
\begin{corollary}\label{QQTapproximationcorollary} \ 
\begin{enumerate}
\item For $k = 2$, $\norm{{L_k^T}L_k - 2{n^2}M_{Id_{k-2,1}}}$ is $\tilde{O}(n\sqrt{m})$
\item For $k = 3$, $\norm{{L_k^T}L_k - 2{n^2}M_{Id_{k-2,1}} - 4nM_{\beta_2}M^T_{\beta_2}}$ is $\tilde{O}(n\sqrt{m})$
\item For $k \geq 4$, $\norm{{L_k^T}L_k - 2{n^2}M_{Id_{k-2,1}} - 4nM_{\beta_2}M^T_{\beta_2} - 4M_{\beta_3}M^T_{\beta_3} - 4M_{\alpha_6}}$ is $\tilde{O}(n\sqrt{m})$
\end{enumerate}
\end{corollary}
\begin{remark}
We replaced $\binom{n-k+2}{2}$ with $\frac{n^2}{2}$ as $\norm{M_{Id_{k-2,1}}} = 1$ and $|\frac{n^2}{2} - \binom{n-k+2}{2}|$ is $\tilde{O}(n)$. Similarly, we replaced $(n-k+1)$ with $n$ as $\norm{M_{\alpha'_2}}$ is $\tilde{O}(n)$ and $|n - (n-k+1)|$ is $\tilde{O}(1)$
\end{remark}
\subsubsection{The Null Space $N_k$}
We now construct a matrix $N_k$ for each $k \geq 4$ such that ${L_k}{N_k} = 0$ and the columns of ${N_k}$ span the nullspace of $L_k$. To do this, we construct $N_k$ so that the entries of each column of $N_k$ is indexed by a subset $S = \{j_3,\ldots,j_{k}\} \subseteq [n]$ and an ordered tuple of circle indices $(i,i')$ where $i < i'$. We then want that if we view $\tilde{E}$ as a vector,
\[
({\tilde{E}^T}{L_k}N_k)_{S,(i,i')} = \tilde{E}\left[v^{S}\left({\ip{v}{d_i}}^2 - 1\right)\left({\ip{v}{d_{i'}}}^2 - 1\right)\right] - \tilde{E}\left[v^{S}\left({\ip{v}{d_{i'}}}^2 - 1\right)\left({\ip{v}{d_i}}^2 - 1\right)\right] = 0
\]
\begin{lemma}
$N_k = c_{\alpha_1}(M_{\alpha_1^+} - M_{\alpha_1^-}) + c_{\alpha_2}(M_{\alpha_2^+} - M_{\alpha_2^-}) + c_{\alpha_3}(M_{\alpha_3^+} - M_{\alpha_3^-}) + c_{\alpha_4}(M_{\alpha_4^+} - M_{\alpha_4^-}) + c_{\alpha_5}(M_{\alpha_5^+} - M_{\alpha_5^-})$ for the following shapes $\alpha_1^{+},\ldots,\alpha_5^{+}, \alpha_1^{-},\ldots,\alpha_5^{-}$ and coefficients $c_{\alpha_1},\ldots,c_{\alpha_5}$. Unless stated otherwise, all of these shapes have no middle vertices.
\begin{enumerate}
\item $U_{\alpha_1^+} \setminus V_{\alpha_1^+} = \{j_1,j_2\}$, $U_{\alpha_1^+} \cap V_{\alpha_1^+} = \{j_3,j_4,\ldots,j_{k}\} \cup \{i'_{circ}\}$, $V_{\alpha_1^+} \setminus U_{\alpha_1^+} = \{i_{circ}\}$, $E(\alpha_1^+) = \{(j_1,i),(j_2,i)\}$, and $c_{\alpha_1} = 2$.
\item $U_{\alpha_2^+} \setminus V_{\alpha_2^+} = \{j_2\}$, $U_{\alpha_2^+} \cap V_{\alpha_2^+} = \{j_4,\ldots,j_{k}\} \cup \{i'_{circ}\}$, $V_{\alpha_2^+} \setminus U_{\alpha_2^+} = \{j_3\} \cup \{i_{circ}\}$, $E(\alpha_2^+) = \{(j_3,i),(j_2,i)\}$, and $c_{\alpha_2} = \frac{2}{n}$.
\item $U_{\alpha_3^+} \setminus V_{\alpha_3^+} = \emptyset$, $U_{\alpha_3^+} \cap V_{\alpha_3^+} = \{j_5,\ldots,j_{k}\} \cup \{i'_{circ}\}$, $V_{\alpha_3^+} \setminus U_{\alpha_3^+} = \{j_3,j_4\} \cup \{i_{circ}\}$, $E(\alpha_3^+) = \{(j_3,i),(j_4,i)\}$, and $c_{\alpha_3} = \frac{2}{n^2}$.
\item $U_{\alpha_4^+} \setminus V_{\alpha_4^+} = \emptyset$, $U_{\alpha_4^+} \cap V_{\alpha_4^+} = \{j_3,j_4,\ldots,j_{k}\} \cup \{i'_{circ}\}$, $V_{\alpha_4^+} \setminus U_{\alpha_4^+} = \{i_{circ}\}$, $V(\alpha_4^+) \setminus (U_{\alpha_4^+} \cup V_{\alpha_4^+}) = \{j_1\}$ $E(\alpha_4^+) = \{(j_1,i)_2\}$, and $c_{\alpha_4} = \frac{1}{n}$.
\item $U_{\alpha_5^+} \setminus V_{\alpha_5^+} = \emptyset$, $U_{\alpha_5^+} \cap V_{\alpha_5^+} = \{j_3,j_4,\ldots,j_{k}\} \cup \{i'_{circ}\}$, $V_{\alpha_5^+} \setminus U_{\alpha_5^+} = \{i_{circ}\}$, $E(\alpha_5^+) = \{(j_3,i)_2\}$, and $c_{\alpha_5} = \frac{1}{n}$.
\end{enumerate}
where for all of these shapes, $(i_{circ},i'_{circ})$ is a tuple in the right side and $i < i'$. $\alpha_1^{-},\ldots,\alpha_5^{-}$ are the same as $\alpha_1^{+},\ldots,\alpha_5^{+}$ except that $i$ and $i'$ are swapped.
\end{lemma}
\begin{remark}
Note that $\alpha_1^{+},\ldots,\alpha_5^{+}$ are the same shapes which appear in the decomposition of $L_k$ except that the intersection of $U$ and $V$ now contains $i'_{circ}$ and we require that $i < i'$.
\end{remark}
For pictures of these shapes, see Figure \ref{NDecompositionFigure} below.
\begin{figure}[ht]
\centerline{\includegraphics[height=10cm]{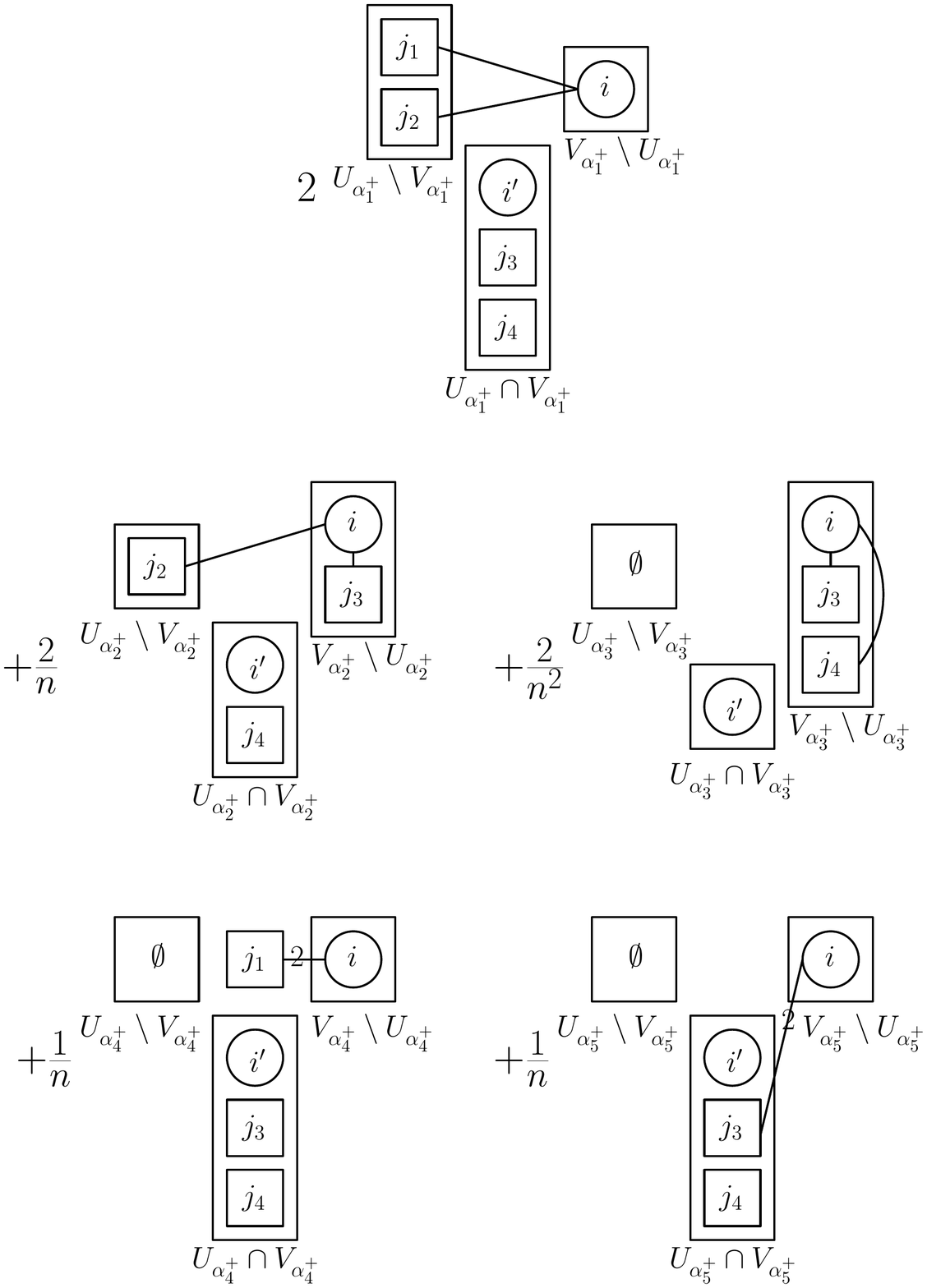}}
\caption{This figure shows the decomposition of $N_k$ for $k = 4$. Here we always have that $i < i'$. If $i$ and $i'$ are swapped then this flips the signs but these parts are not shown to save space.}
\label{NDecompositionFigure}
\end{figure}
\begin{proof}
To determine $N_k$, we analyze the ribbons which $N_k$ is composed of. Let $S = \{j_3,j_4,\ldots,j_k\}$.
\begin{enumerate}
\item If we take a ribbon $R$ with $A_R = \{j_1,\ldots,j_{k}\} \cup \{i'_{circ}\}$, $B_R = \{j_3,\ldots,j_{k}\} \cup (i_{circ},i'_{circ})$, and $E(R) = \{(j_1,i),(j_2,i)\}$ where $j_1 \neq j_2$ and $j_1,j_2 \notin S$ then 
\[
({\tilde{E}^T}{L_k}M_R)_{S,i,i'} = \tilde{E}\left[v^{S}v_{j_1}v_{j_2}(d_{i})_{j_1}(d_i)_{j_2}\left({\ip{v}{d_{i'}}}^2 - 1\right)\right]
\]
Each such term appears with a coefficient of $2$ in $\tilde{E}\left[v^{S}\left({\ip{v}{d_i}}^2 - 1\right)\left({\ip{v}{d_{i'}}}^2 - 1\right)\right]$, so we want each such ribbon $R$ to appear with a coefficient of $2$ in $N_k$.

Similarly, we want each ribbon $R$ with $A_R = \{j_1,\ldots,j_{k}\} \cup \{i_{circ}\}$, $B_R = \{j_3,\ldots,j_{k}\} \cup (i_{circ},i'_{circ})$, and $E(R) = \{(j_1,i'),(j_2,i')\}$ where $j_1 \neq j_2$ and $j_1,j_2 \notin S$ to appear with a coefficient of $-2$ in $N_k$.
\item If we take a ribbon $R$ with $A_R = (\{j_1,\ldots,j_{k}\} \setminus \{j_1,j_3\}) \cup \{i'_{circ}\}$, $B_R = \{j_3,\ldots,j_{k}\} \cup (i_{circ},i'_{circ})$, and $E(R) = \{(j_3,i),(j_2,i)\}$ where $j_1 = j_3 \in S$ and $j_2 \notin S$ then 
\begin{align*}
({\tilde{E}^T}{L_k}M_R)_{S,i,i'} &= \tilde{E}\left[v^{S \setminus \{j_1,j_3\}}v_{j_2}(d_{i})_{j_3}(d_i)_{j_2}\left({\ip{v}{d_{i'}}}^2 - 1\right)\right]\\
&=n\tilde{E}\left[v^{S}v_{j_1}v_{j_2}(d_{i})_{j_1}(d_i)_{j_2}\left({\ip{v}{d_{i'}}}^2 - 1\right)\right]
\end{align*}
Each such term appears with a coefficient of $\frac{2}{n}$ in $\tilde{E}\left[v^{S}\left({\ip{v}{d_i}}^2 - 1\right)\left({\ip{v}{d_{i'}}}^2 - 1\right)\right]$, so we want each such ribbon $R$ to appear with a coefficient of $\frac{2}{n}$ in $N_k$.

Similarly, we want each ribbon $R$ with $A_R = (\{j_1,\ldots,j_{k}\} \setminus \{j_1,j_3\}) \cup \{i_{circ}\}$, $B_R = \{j_3,\ldots,j_{k}\} \cup (i_{circ},i'_{circ})$, and $E(R) = \{(j_3,i'),(j_2,i')\}$ where $j_1 = j_3 \in S$ and $j_2 \notin S$ to appear with a coefficient of $-\frac{2}{n}$ in $N_k$.
\item If we take a ribbon $R$ with $A_R = (\{j_1,\ldots,j_{k}\} \setminus \{j_1,j_2,j_3,j_4\}) \cup \{i'_{circ}\}$, $B_R = \{j_3,\ldots,j_{k}\} \cup (i_{circ},i'_{circ})$, and $E(R) = \{(j_3,i),(j_4,i)\}$ where $j_1 = j_3 \in S$ and $j_2 = j_4 \in S$ then 
\begin{align*}
({\tilde{E}^T}{L_k}M_R)_{S,i,i'} &= \tilde{E}\left[v^{S \setminus \{j_1,j_2,j_3,j_4\}}(d_{i})_{j_3}(d_i)_{j_4}\left({\ip{v}{d_{i'}}}^2 - 1\right)\right]\\
&=n^2\tilde{E}\left[v^{S}v_{j_1}v_{j_2}(d_{i})_{j_1}(d_i)_{j_2}\left({\ip{v}{d_{i'}}}^2 - 1\right)\right]
\end{align*}
Each such term appears with a coefficient of $\frac{2}{n^2}$ in $\tilde{E}\left[v^{S}\left({\ip{v}{d_i}}^2 - 1\right)\left({\ip{v}{d_{i'}}}^2 - 1\right)\right]$, so we want each such ribbon $R$ to appear with a coefficient of $\frac{2}{n^2}$ in $N_k$.

Similarly, we want each ribbon $R$ with $A_R = (\{j_1,\ldots,j_{k}\} \setminus \{j_1,j_2,j_3,j_4\}) \cup \{i_{circ}\}$, $B_R = \{j_3,\ldots,j_{k}\} \cup (i_{circ},i'_{circ})$, and $E(R) = \{(j_3,i'),(j_4,i')\}$ where $j_1 = j_3 \in S$ and $j_2 = j_4 \in S$ to appear with a coefficient of $-\frac{2}{n^2}$ in $N_k$.
\item If we take a ribbon $R$ with $A_R = (\{j_1,\ldots,j_{k}\} \setminus \{j_1,j_2\})  \cup \{i'_{circ}\}$, $B_R = \{j_3,\ldots,j_{k}\} \cup (i_{circ},i'_{circ})$, and $E(R) = \{(j_1,i)_2\}$ where $j_1 = j_2 \notin S$ then 
\begin{align*}
({\tilde{E}^T}{L_k}M_R)_{S,i,i'} &= \tilde{E}\left[v^{S \setminus \{j_1,j_2\}}((d_{i})^2_{j_1} - 1)\left({\ip{v}{d_{i'}}}^2 - 1\right)\right]\\
&=n\tilde{E}\left[v^{S}v_{j_1}^2((d_{i})^2_{j_1} - 1)\left({\ip{v}{d_{i'}}}^2 - 1\right)\right]
\end{align*}
Each such term appears with a coefficient of $\frac{1}{n}$ in $\tilde{E}\left[v^{S}\left({\ip{v}{d_i}}^2 - 1\right)\left({\ip{v}{d_{i'}}}^2 - 1\right)\right]$ so we want each such ribbon $R$ to appear with a coefficient of $\frac{1}{n}$ in $N_k$.

Similarly, we want each ribbon $R$ with $A_R = (\{j_1,\ldots,j_{k}\} \setminus \{j_1,j_2\})  \cup \{i'_{circ}\}$, $B_R = \{j_3,\ldots,j_{k}\} \cup (i_{circ},i'_{circ})$, and $E(R) = \{(j_1,i)_2\}$ where $j_1 = j_2 \notin S$ to appear with a coefficient of $-\frac{1}{n}$ in $N_k$.
\item If we take a ribbon $R$ with $A_R = (\{j_1,\ldots,j_{k}\} \setminus \{j_1,j_2\}) \cup \{i'_{circ}\}$, $B_R = \{j_3,\ldots,j_{k}\} \cup (i_{circ},i'_{circ})$, and $E(R) = \{(j_3,i)_2\}$ where $j_1 = j_2 = j_3 \in S$ then 
\begin{align*}
({\tilde{E}^T}{L_k}M_R)_{S,i,i'} &= \tilde{E}\left[v^{S \setminus \{j_1,j_2\}}((d_{i})^2_{j_3} - 1)\left({\ip{v}{d_{i'}}}^2 - 1\right)\right]\\
&=n\tilde{E}\left[v^{S}v_{j_1}^2((d_{i})^2_{j_1} - 1)\left({\ip{v}{d_{i'}}}^2 - 1\right)\right]
\end{align*}
Each such term appears with a coefficient of $\frac{1}{n}$ in $\tilde{E}\left[v^{S}\left({\ip{v}{d_i}}^2 - 1\right)\left({\ip{v}{d_{i'}}}^2 - 1\right)\right]$ so we want each such ribbon $R$ to appear with a coefficient of $\frac{1}{n}$ in $N_k$.

Similarly, we want each ribbon $R$ with $A_R = (\{j_1,\ldots,j_{k}\} \setminus \{j_1,j_2\}) \cup \{i'_{circ}\}$, $B_R = \{j_3,\ldots,j_{k}\} \cup (i_{circ},i'_{circ})$, and $E(R) = \{(j_3,i)_2\}$ where $j_1 = j_2 = j_3 \in S$ to appear with a coefficient of $-\frac{1}{n}$ in $N_k$.
\end{enumerate}
\end{proof}
Observe that the dominant part of $N_k$ is $2(M_{\alpha_1^+} - M_{\alpha_1^-})$ which has norm $\tilde{O}(n)$. The norm bounds for the other components of $N_k$ are as follows:
\begin{enumerate}
\item $\norm{c_{\alpha_2}(M_{\alpha_2^+} - M_{\alpha_2^-})}$ is $\tilde{O}\left(\frac{1}{n} \cdot \sqrt{mn}\right) = \tilde{O}\left(\frac{\sqrt{m}}{\sqrt{n}}\right)$
\item $\norm{c_{\alpha_3}(M_{\alpha_3^+} - M_{\alpha_3^-})}$ is $\tilde{O}\left(\frac{1}{n^2} \cdot n\sqrt{m}\right) = \tilde{O}\left(\frac{\sqrt{m}}{n}\right)$
\item $\norm{c_{\alpha_4}(M_{\alpha_4^+} - M_{\alpha_4^-})}$ is $\tilde{O}\left(\frac{1}{n} \cdot \sqrt{mn}\right) = \tilde{O}\left(\frac{\sqrt{m}}{\sqrt{n}}\right)$
\item $\norm{c_{\alpha_5}(M_{\alpha_5^+} - M_{\alpha_5^-})}$ is $\tilde{O}\left(\frac{1}{n} \cdot \sqrt{m}\right) = \tilde{O}\left(\frac{\sqrt{m}}{n}\right)$
\end{enumerate}
\subsubsection{Analyzing $N_k{N^T_k}$}
The dominant terms of $N_k$ are $2M_{\alpha^{+}} - 2M_{\alpha^{-}}$ where 
\begin{enumerate}
\item $U_{\alpha^{+}} \setminus V_{\alpha^{+}} = \{j_1,j_2\}$, $U_{\alpha^{+}} \cap V_{\alpha^{+}} = \{j_3,\ldots,j_{k}\} \cup \{i'_{circ}\}$, $V_{\alpha^{+}} \setminus U_{\alpha^{+}} = \{i_{circ}\}$, and $E(\alpha^{+}) = \{(j_1,i),(j_2,i)\}$. Note that here $i < i'$ and $(i_{circ},i'_{circ})$ appears as a tuple in $V_{\alpha^{+}}$.
\item $U_{\alpha^{-}} \setminus V_{\alpha^{-}} = \{j_1,j_2\}$, $U_{\alpha^{-}} \cap V_{\alpha^{-}} = \{j_3,\ldots,j_{k}\} \cup \{i_{circ}\}$, $V_{\alpha^{-}} \setminus U_{\alpha^{-}} = \{i'_{circ}\}$, and $E(\alpha^{-}) = \{(j_1,i'),(j_2,i')\}$. Note that here $i < i'$ and $(i'_{circ},i_{circ})$ appears as a tuple in $V_{\alpha^{-}}$.
\end{enumerate}
\begin{lemma}\label{NNTPluslemma}
$M_{\alpha^{+}}M_{\alpha^{+}}^T + M_{\alpha^{-}}M_{\alpha^{-}}^T = M_{\alpha_1} + M_{\alpha_2} + M_{\alpha_3}$ where $\alpha_1,\alpha_2,\alpha_3$ are the following shapes.
\begin{enumerate}
\item $U_{\alpha_1} \setminus V_{\alpha_1} = \emptyset$, $U_{\alpha_1} \cap V_{\alpha_1} = \{j_1,j_2,j_3,j_4,\ldots,j_{k}\} \cup \{i_{circ}\}$, $V_{\alpha_1} \setminus U_{\alpha_1} = \emptyset$, $V(\alpha_1) \setminus (U_{\alpha_1} \cup V_{\alpha_1}) = \{i'_{circ}\}$, and $E(\alpha_1) = \{(i',j_1),(i',j_1),(i',j_2),(i',j_2)\}$.
\item $U_{\alpha_2} \setminus V_{\alpha_2} = \{j_1\}$, $U_{\alpha_2} \cap V_{\alpha_2} = \{j_2,j_4,\ldots,j_{k}\} \cup \{i_{circ}\}$, $V_{\alpha_2} \setminus U_{\alpha_2} = \{j'_1\}$, $V(\alpha_2) \setminus (U_{\alpha_2} \cup V_{\alpha_2}) = \{i'_{circ}\}$, and $E(\alpha_2) = \{(i',j_1),(i',j'_1),(i',j_2),(i',j_2)\}$.
\item $U_{\alpha_3} \setminus V_{\alpha_3} = \{j_1,j_2\}$, $U_{\alpha_3} \cap V_{\alpha_3} = \{j_5,\ldots,j_{k}\} \cup \{i_{circ}\}$, $V_{\alpha_3} \setminus U_{\alpha_3} = \{j'_1,j'_2\}$, $V(\alpha_3) \setminus (U_{\alpha_3} \cup V_{\alpha_3}) = \{i'_{circ}\}$, and $E(\alpha_3) = \{(i',j_1),(i',j_2),(i',j'_1),(i',j'_2)\}$.
\end{enumerate}
Note that for these shapes, we do not assume that $i < i'$. Also note that $\alpha_1$ and $\alpha_2$ are improper shapes, though this does not matter for us.
\end{lemma}
\begin{remark}
Actually, we do not need to do this computation as we will just use that $M_{\alpha^{+}}M_{\alpha^{+}}^T + M_{\alpha^{-}}M_{\alpha^{-}}^T \succeq 0$, but we include it anyways to show the similarity with the decomposition of ${L_k^T}L_k$.
\end{remark}
For pictures of these shapes, see Figure \ref{NNTPlusfigure} below.
\begin{figure}[ht]\label{NNTPlusfigure}
\centerline{\includegraphics[height=6cm]{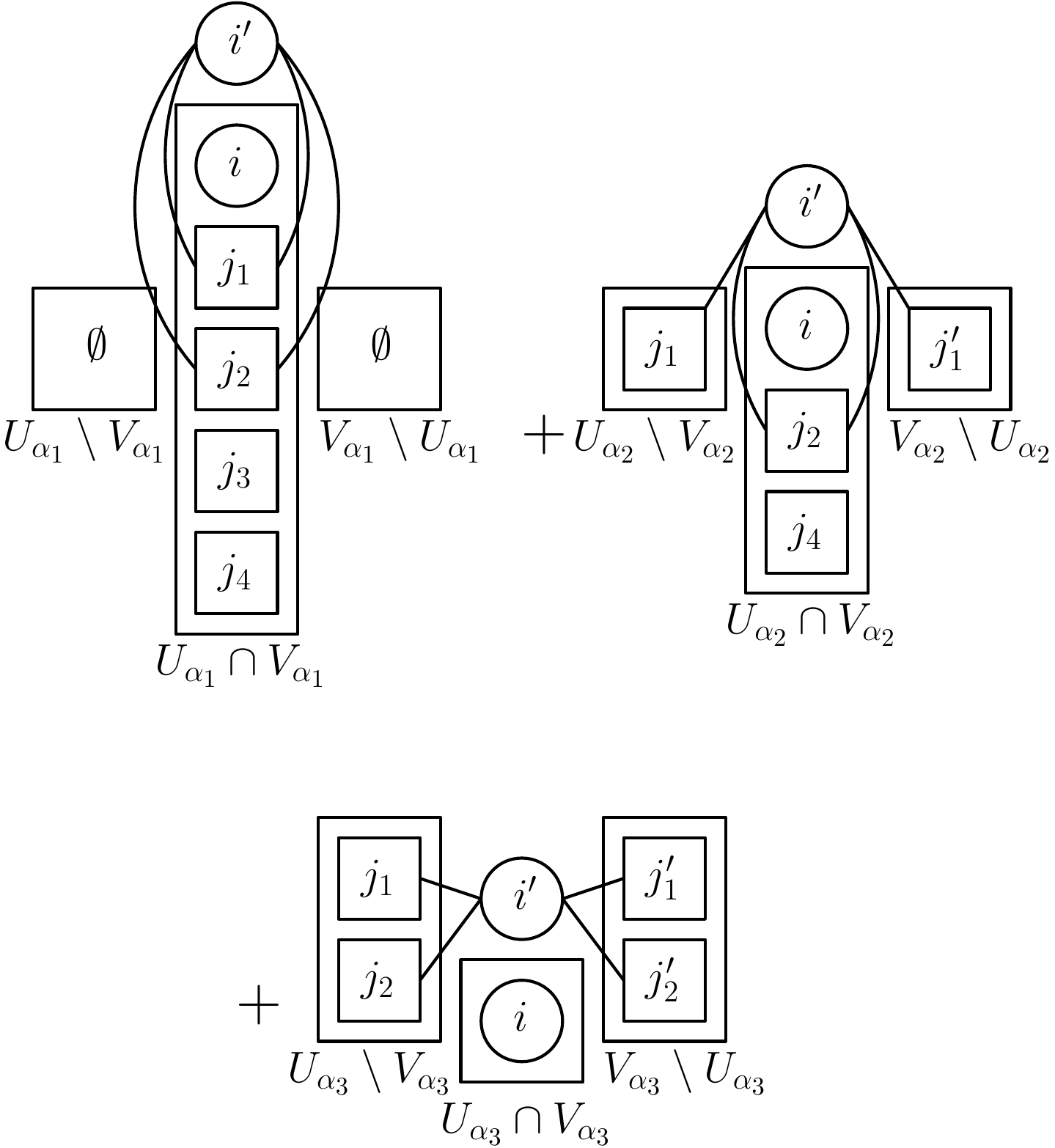}}
\caption{This figure shows the decomposition of $M_{\alpha^{+}}M_{\alpha^{+}}^T + M_{\alpha^{-}}M_{\alpha^{-}}^T$.}
\end{figure}
\begin{proof}[Proof of Lemma \ref{NNTPluslemma}]
We compute $M_{\alpha^{+}}M_{\alpha^{+}}^T + M_{\alpha^{-}}M_{\alpha^{-}}^T$ by considering the ribbons which appear in $M_{\alpha^{+}}M_{\alpha^{+}}^T + M_{\alpha^{-}}M_{\alpha^{-}}^T$. For these ribbons, we do not assume that $i < i'$.
\begin{enumerate}
\item Each ribbon $R$ with $A_R \setminus B_R = \emptyset$, $A_R \cap B_R = \{j_1,j_2,j_3,j_4,\ldots,j_{k}\} \cup \{i_{circ}\}$, $B_R \setminus A_R = \emptyset$, $V(R) \setminus (A_R \cup B_R) = \{i'_{circ}\}$, and $E(R) = \{(i',j_1),(i',j_1),(i',j_2),(i',j_2)\}$ appears in exactly one way as the composition of the ribbons $R_1$ and $R_2$ where 
\begin{enumerate}
\item $A_{R_1} \setminus B_{R_1} = \{i_{circ}\}$, $A_{R_1} \cap B_{R_1} = \{j_3,j_4,\ldots,j_{k}\}$, $B_{R_1} \setminus A_{R_1} = \{j_1,j_2\}$, and $E(R_1) = \{(i,j_1),(i,j_2)\}$.
\item $A_{R_2} \setminus B_{R_2} = \{j_1,j_2\}$, $A_{R_2} \cap B_{R_2} = \{j_3,j_4,\ldots,j_{k}\}$, $B_{R_2} \setminus A_{R_2} = \{i_{circ}\}$, and $E(R_2) = \{(i,j_1),(i,j_2)\}$.
\end{enumerate}
Whether $i < i'$ or $i' < i$ only affects whether this ribbon appears in $M_{\alpha^{+}}M_{\alpha^{+}}^T$ or $M_{\alpha^{-}}M_{\alpha^{-}}^T$.
\item Each ribbon $R$ with $A_R \setminus B_R = \{j_1\}$, $A_R \cap B_R = \{j_4,\ldots,j_{k}\} \cup \{i_{circ}\}$, $B_R \setminus A_R = \{j'_1\}$, $V(R) \setminus (A_R \cup B_R) = \{i'_{circ}\}$, and $E(R) = \{(i,j_1),(i,j'_1),(i,j_2),(i,j_2)\}$ appears in exactly one way as the composition of the ribbons $R_1$ and $R_2$ where 
\begin{enumerate}
\item $A_{R_1} \setminus B_{R_1} = \{i_{circ}\}$, $A_{R_1} \cap B_{R_1} = \{j_1,j_4,\ldots,j_{k}\}$, $B_{R_1} \setminus A_{R_1} = \{j'_1,j_2\}$, and $E(R_1) = \{(i,j'_1),(i,j_2)\}$.
\item $A_{R_2} \setminus B_{R_2} = \{j_1,j_2\}$, $A_{R_2} \cap B_{R_2} = \{j'_1,j_4,\ldots,j_{k}\}$, $B_{R_2} \setminus A_{R_2} = \{i_{circ}\}$, and $E(R_2) = \{(i,j_1),(i,j_2)\}$.
\end{enumerate}
Whether $i < i'$ or $i' < i$ only affects whether this ribbon appears in $M_{\alpha^{+}}M_{\alpha^{+}}^T$ or $M_{\alpha^{-}}M_{\alpha^{-}}^T$.
\item Each ribbon $R$ with $A_R \setminus B_R = \{j_1,j_2\}$, $A_R \cap B_R = \{j_5,\ldots,j_{k}\} \cup \{i_{circ}\}$, $B_R \setminus A_R = \{j'_1,j'_2\}$, $V(R) \setminus (A_R \cup B_R) = \{i'_{circ}\}$, and $E(R) = \{(i,j_1),(i,j'_1),(i,j_2),(i,j'_2)\}$ appears in exactly one way as the composition of the ribbons $R_1$ and $R_2$ where 
\begin{enumerate}
\item $A_{R_1} \setminus B_{R_1} = \{i_{circ}\}$, $A_{R_1} \cap B_{R_1} = \{j_1,j_2,j_5,\ldots,j_{k}\}$, $B_{R_1} \setminus A_{R_1} = \{j'_1,j'_2\}$, and $E(R_1) = \{(i,j'_1),(i,j'_2)\}$.
\item $A_{R_2} \setminus B_{R_2} = \{j_1,j_2\}$, $A_{R_2} \cap B_{R_2} = \{j'_1,j'_2,j_5,\ldots,j_{k}\}$, $B_{R_2} \setminus A_{R_2} = \{i_{circ}\}$, and $E(R_2) = \{(i,j_1),(i,j_2)\}$.
\end{enumerate}
Whether $i < i'$ or $i' < i$ only affects whether this ribbon appears in $M_{\alpha^{+}}M_{\alpha^{+}}^T$ or $M_{\alpha^{-}}M_{\alpha^{-}}^T$.
\end{enumerate}
\end{proof}
\begin{lemma}\label{NNTMinuslemma}
$M_{\alpha^{+}}M_{\alpha^{-}}^T + M_{\alpha^{-}}M_{\alpha^{+}}^T = M_{\alpha_4} + M_{\alpha_5} + M_{\alpha_6}$ where $\alpha_4,\alpha_5,\alpha_6$ are the following shapes.
\begin{enumerate}
\item $U_{\alpha_4} \setminus V_{\alpha_4} = \{i_{circ}\}$, $U_{\alpha_4} \cap V_{\alpha_4} = \{j_1,j_2,j_3,j_4,\ldots,j_{k}\}$, $V_{\alpha_4} \setminus U_{\alpha_4} = \{i'_{circ}\}$, and $E(\alpha_4) = \{(i,j_1),(i,j_2),$ $(i',j_1),(i',j_2)\}$.
\item $U_{\alpha_5} \setminus V_{\alpha_5} = \{j_1\} \cup \{i_{circ}\}$, $U_{\alpha_5} \cap V_{\alpha_5} = \{j_2,j_4,\ldots,j_{k}\}$, $V_{\alpha_5} \setminus U_{\alpha_5} = \{j'_1\} \cup \{i'_{circ}\}$, and $E(\alpha_5) = \{(i,j'_1),(i,j_2),(i',j_1),(i',j_2)\}$.
\item $U_{\alpha_6} \setminus V_{\alpha_6} = \{j_1,j_2\} \cup \{i_{circ}\}$, $U_{\alpha_6} \cap V_{\alpha_6} = \{j_5,\ldots,j_{k}\}$, $V_{\alpha_6} \setminus U_{\alpha_6} = \{j'_1,j'_2\}  \cup \{i'_{circ}\}$, and $E(\alpha_6) = \{(i',j_1),(i',j_2),(i,j'_1),(i,j'_2)\}$.
\end{enumerate}
\end{lemma}
For pictures of these shapes, see Figure \ref{NNTMinusfigure} below.
\begin{figure}[ht]\label{NNTMinusfigure}
\centerline{\includegraphics[height=6cm]{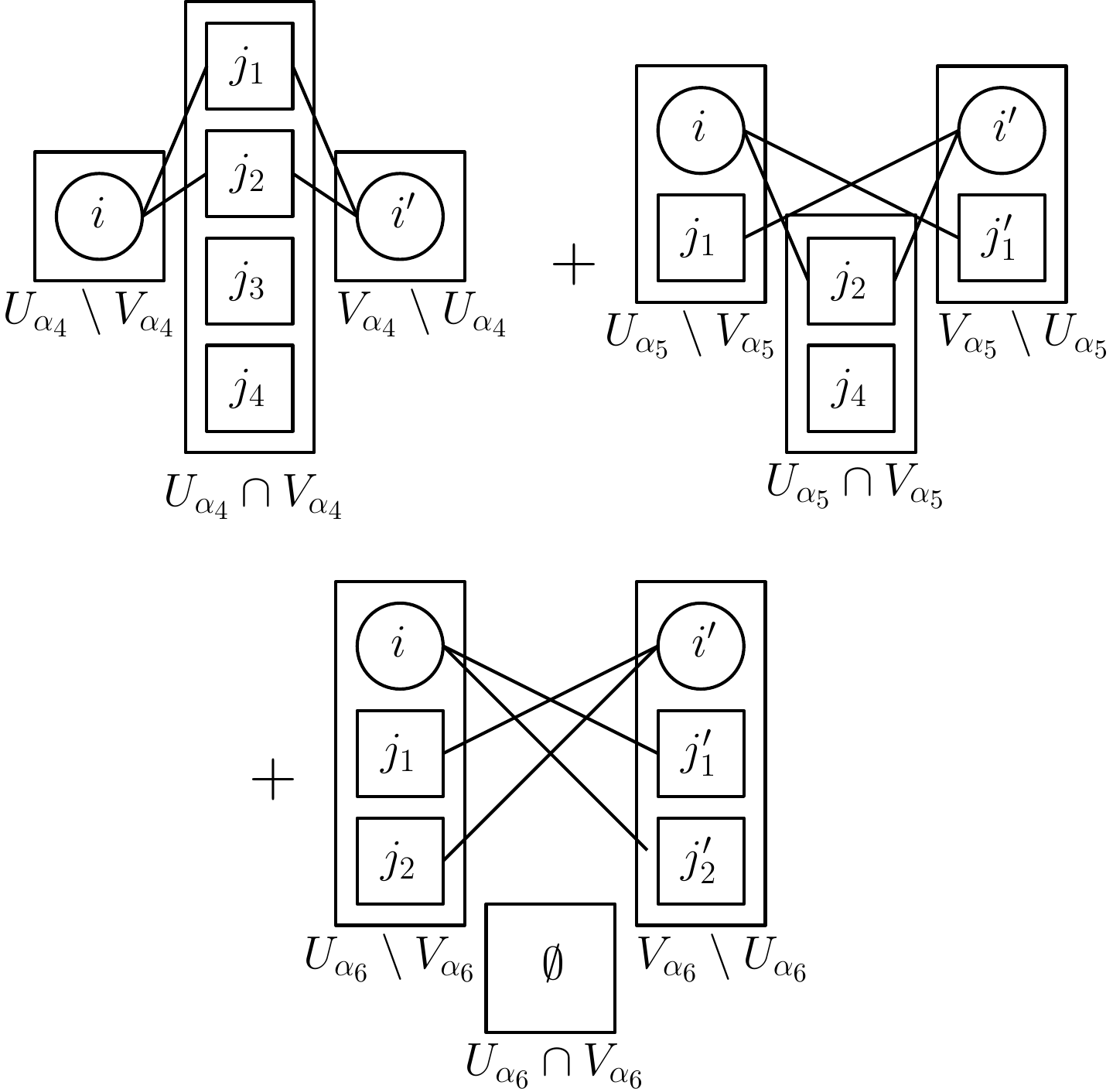}}
\caption{This figure shows the decomposition of $M_{\alpha^{+}}M_{\alpha^{+}}^T + M_{\alpha^{-}}M_{\alpha^{-}}^T$.}
\end{figure}
\begin{proof}[Proof of Lemma \ref{NNTMinuslemma}]
We compute $M_{\alpha^{+}}M_{\alpha^{-}}^T + M_{\alpha^{-}}M_{\alpha^{+}}^T$ by considering the ribbons which appear in $M_{\alpha^{+}}M_{\alpha^{-}}^T + M_{\alpha^{-}}M_{\alpha^{+}}^T$. For these ribbons, we do not assume that $i < i'$.
\begin{enumerate}
\item Each ribbon $R$ with $A_R \setminus B_R = \{i_{circ}\}$, $A_R \cap B_R = \{j_1,j_2,j_3,j_4,\ldots,j_{k}\}$, $B_R \setminus A_R = \{i'_{circ}\}$, $V(R) \setminus (A_R \cup B_R) = \emptyset$, and $E(R) = \{(i,j_1),(i,j_2),(i',j_1),(i',j_2)\}$ appears in exactly one way as the composition of the ribbons $R_1$ and $R_2$ where 
\begin{enumerate}
\item $A_{R_1} \setminus B_{R_1} = \{j_1,j_2\}$, $A_{R_1} \cap B_{R_1} = \{j_3,j_4,\ldots,j_{k}\} \cup \{i_{circ}\}$, $B_{R_1} \setminus A_{R_1} = \{i'_{circ}\}$, and $E(R_1) = \{(i',j_1),(i',j_2)\}$.
\item $A_{R_2} \setminus B_{R_2} =  \{i_{circ}\}$, $A_{R_2} \cap B_{R_2} = \{j_3,j_4,\ldots,j_{k}\} \cup \{i'_{circ}\}$, $B_{R_2} \setminus A_{R_2} = \{j_1,j_2\}$, and $E(R_2) = \{(i,j_1),(i,j_2)\}$.
\end{enumerate}
Whether $i < i'$ or $i' < i$ only affects whether this ribbon appears in $M_{\alpha^{+}}M_{\alpha^{-}}^T$ or $M_{\alpha^{-}}M_{\alpha^{+}}^T$.
\item Each ribbon $R$ with $A_R \setminus B_R = \{j_1\} \cup \{i_{circ}\}$, $A_R \cap B_R = \{j_2,j_4,\ldots,j_{k}\}$, $B_R \setminus A_R = \{j'_1\} \cup \{i'_{circ}\}$, $V(R) \setminus (A_R \cup B_R) = \emptyset$, and $E(R) = \{(i,j_1),(i,j_2),(i',j'_1),(i',j_2)\}$ appears in exactly one way as the composition of the ribbons $R_1$ and $R_2$ where 
\begin{enumerate}
\item $A_{R_1} \setminus B_{R_1} = \{j_1,j_2\}$, $A_{R_1} \cap B_{R_1} = \{j'_1,j_4,\ldots,j_{k}\} \cup \{i_{circ}\}$, $B_{R_1} \setminus A_{R_1} = \{i'_{circ}\}$, and $E(R_1) = \{(i',j_1),(i',j_2)\}$.
\item $A_{R_2} \setminus B_{R_2} =  \{i_{circ}\}$, $A_{R_2} \cap B_{R_2} = \{j_1,j_4,\ldots,j_{k}\} \cup \{i'_{circ}\}$, $B_{R_2} \setminus A_{R_2} = \{j'_1,j_2\}$, and $E(R_2) = \{(i,j'_1),(i,j_2)\}$.
\end{enumerate}
Whether $i < i'$ or $i' < i$ only affects whether this ribbon appears in $M_{\alpha^{+}}M_{\alpha^{-}}^T$ or $M_{\alpha^{-}}M_{\alpha^{+}}^T$.
\item Each ribbon $R$ with $A_R \setminus B_R = \{j_1,j_2\} \cup  \{i_{circ}\}$, $A_R \cap B_R = \{j_5,\ldots,j_{k}\}$, $B_R \setminus A_R = \{j'_1,j'_2\}  \cup  \{i'_{circ}\}$, $V(R) \setminus (A_R \cup B_R) = \emptyset$, and $E(R) = \{(i,j'_1),(i,j'_2),(i',j_1),(i',j_2)\}$ appears in exactly one way as the composition of the ribbons $R_1$ and $R_2$ where 
\begin{enumerate}
\item $A_{R_1} \setminus B_{R_1} = \{j_1,j_2\}$, $A_{R_1} \cap B_{R_1} = \{j_1,j_2,j_5,\ldots,j_{k}\} \cup \{i_{circ}\}$, $B_{R_1} \setminus A_{R_1} = \{i'_{circ}\}$, and $E(R_1) = \{(i,j'_1),(i,j'_2)\}$.
\item $A_{R_2} \setminus B_{R_2} = \{i_{circ}\}$, $A_{R_2} \cap B_{R_2} = \{j'_1,j'_2,j_5,\ldots,j_{k}\}  \cup \{i'_{circ}\}$, $B_{R_2} \setminus A_{R_2} = \{j'_1,j'_2\}$, and $E(R_2) = \{(i,j_1),(i,j_2)\}$.
\end{enumerate}
Whether $i < i'$ or $i' < i$ only affects whether this ribbon appears in $M_{\alpha^{+}}M_{\alpha^{-}}^T$ or $M_{\alpha^{-}}M_{\alpha^{+}}^T$.
\end{enumerate}
\end{proof}
We now consider the norm bounds for these terms
\begin{enumerate}
\item $\norm{M_{\alpha_4}}$ is $\tilde{O}(m)$.
\item $\norm{M_{\alpha_5}}$ is $\tilde{O}(m)$.
\item $\norm{M_{\alpha_6}}$ is $\tilde{O}(n^2)$.
\end{enumerate}
This means that $M_{\alpha_6}$ cannot be ignored, but this is fine. In fact, the $M_{\alpha_6}$ here will cancel with the $M_{\alpha_6}$ that appears in the decomposition of ${L_k^T}L_k$
\begin{corollary}\label{NNTapproximationcorollary}
For all $k \geq 4$, $\norm{{N_k}{N_k}^T - 4M_{\alpha^{+}}M_{\alpha^{+}}^T - 4M_{\alpha^{-}}M_{\alpha^{-}}^T + 4M_{\alpha_6}}$ is $\tilde{O}(m)$.
\end{corollary}
\subsubsection{Putting Everything Together}
We now put everything together to prove that for all $k$, the minimum nonzero eigenvalue of ${L_k^T}{L_k}$ is $2n^2 - \tilde{O}(n\sqrt{m})$.
\begin{enumerate}
\item For $k = 2$, by Corollary \ref{QQTapproximationcorollary}, $\norm{{L_k^T}L_k - 2{n^2}M_{Id_{k-2,1}}}$ is $\tilde{O}(n\sqrt{m})$ so the minimum eigenvalue of ${L_k^T}{L_k}$ is $2n^2 - \tilde{O}(n\sqrt{m})$.
\item For $k = 3$, by Corollary \ref{QQTapproximationcorollary}, $\norm{{L_k^T}L_k - 2{n^2}M_{Id_{k-2,1}} - 4nM_{\beta_2}M^T_{\beta_2}}$ is $\tilde{O}(n\sqrt{m})$ so the minimum eigenvalue of ${L_k^T}{L_k}$ is $2n^2 - \tilde{O}(n\sqrt{m})$.
\item For $k \geq 4$, by Corollary \ref{QQTapproximationcorollary}, $\norm{{L_k^T}L_k - 2{n^2}M_{Id_{k-2,1}} - 4nM_{\beta_2}M^T_{\beta_2} - 4M_{\beta_3}M^T_{\beta_3} - 4M_{\alpha_6}}$ is $\tilde{O}(n\sqrt{m})$. By Corollary \ref{NNTapproximationcorollary}, $\norm{{N_k}{N_k}^T - 4M_{\alpha^{+}}M_{\alpha^{+}}^T - 4M_{\alpha^{-}}M_{\alpha^{-}}^T + 4M_{\alpha_6}}$ is $\tilde{O}(m)$. Combining these equations, 
\[
\norm{{L_k^T}L_k + {N_k}{N_k}^T - 2{n^2}M_{Id_{k-2,1}} - 4nM_{\beta_2}M^T_{\beta_2} - 4M_{\beta_3}M^T_{\beta_3} - 4M_{\alpha^{+}}M_{\alpha^{+}}^T - 4M_{\alpha^{-}}M_{\alpha^{-}}^T}
\] is $\tilde{O}(n\sqrt{m})$ so the minimum eigenvalue of ${L_k^T}L_k + {N_k}{N_k}^T$ is $2n^2 - \tilde{O}(n\sqrt{m})$. Since the minimum nonzero eigenvalue of ${L_k^T}L_k$ is at least as large as the minimum eigenvalue of ${L_k^T}L_k + {N_k}{N_k}^T$, the nonzero eigenvalue of ${L_k^T}{L_k}$ is $2n^2 - \tilde{O}(n\sqrt{m})$.
\end{enumerate}
This implies that the minimum nonzero eigenvaue of $QQ^T$ is $2n^2 - \tilde{O}(n\sqrt{m})$, as needed.

\section{Open Problems}\label{sec:open-problems}

We conjecture that for the Planted Affine Planes problem, the problem remains difficult even with the number of vectors increased to $m = n^{2 - \epsilon}$.

\begin{conjecture}
	\cref{theo:sos-bounds} holds with the bound on the number of sampled vectors $m$ loosened to $m \leq n^{2-\eps}$. 
\end{conjecture}

The reason for the upper bound comes from \cref{rmk:pe-one}. Analyzing $\pE[1]$ is an established way to hypothesize about the power of SoS in hypothesis testing problems (see \cite{hop17}, \cite{hop18}).

Dual to the Planted Affine Planes problem, we conjecture a similar
bound for Planted Boolean Vector problem whenever $d \geq
n^{1/2+\eps}$.

\begin{conjecture}
	\cref{theo:boolean-subspace} holds with the bound on the dimension $p$ of a random subspace
  loosened to $p \geq n^{1/2+\eps}$.
\end{conjecture}

We conjecture that the Planted Boolean Vector problem/Planted Affine 
Planes problem is still hard for SoS if the input is no longer i.i.d. 
Gaussian or boolean entries, but is drawn from a ``random enough'' 
distribution. For example, if in the random instance of PAP the vectors 
$d_u$ are i.i.d. samples from $S^n$, or a random orthonormal system, 
degree $n^\delta$ SoS should still believe the instance is satisfiable 
(after appropriate normalization of $v$). Or, taking the view of Planted 
Boolean Vector, if the subspace is the eigenspace of the bottom 
eigenvectors of a random adjacency matrix, the instance should still be 
difficult. This last setting arises in MaxCut, for which we conjecture the following.

\begin{conjecture}
    Let $d \geq 3$, and let $G$ be a random $d$-regular graph on $n$ vertices. For some $\delta > 0$, w.h.p. there is a degree-$n^\delta$ pseudoexpectation operator $\pE$ on boolean variables $x_i$ with MaxCut value at least
    \[ \frac{1}{2} + \frac{\sqrt{d-1}}{d}(1 - \littleoh_{d,n}(1)) \]
\end{conjecture}

The above expression is w.h.p. the value of the spectral relaxation for MaxCut, therefore qualitatively this conjecture expresses that degree $n^\delta$ SoS cannot significantly tighten the basic spectral relaxation. 

We should remark that, with respect to the goal of showing SoS cannot significantly outperform the Goemans-Williamson relaxation, random instances are not integrality gap instances. The main difficulty in comparing (even degree 4) SoS to the Goemans-Williamson algorithm seems to be the lack of a candidate hard input distribution.

Evidence for this conjecture comes from the fact that the only property 
required of the random inputs $d_1, \dots, d_m$ was that norm bounds hold for 
the graph matrix with Hermite polynomial entries. When the variables 
$\{d_{u,i}\}$ are i.i.d from some other distribution, if we use graph matrices 
for the orthonormal polynomials under the distribution and assuming suitable 
bounds on the moments of the distribution, the same norm bounds 
hold~\cite{AMP20}.
When $d_u \unif S^n$ or another distribution for which the coordinates are not 
i.i.d, it seems likely that if we use e.g. the spherical harmonics then 
similar norm bounds hold, but this is not proven.

\subsection*{Acknowledgements}

We thank Madhur Tulsiani and Pravesh K. Kothari for several enlightening discussions in the initial phases of this work. G.R. thanks Sidhanth Mohanty for useful discussions regarding
the Sherrington-Kirkpatrick problem. We also thank the anonymous reviewers for their useful suggestions on improving the text.

\bibliographystyle{alphaurl}
\bibliography{macros,madhur}

\appendix

\section{Norm Bounds}\label{app:norm_bounds}

The norm bounds we use come from applying the trace power method
in~\cite{AMP20}. The paper~\cite{AMP20} uses a slightly different
definition of matrix index. They define a \textit{matrix index piece}
as a tuple of distinct elements from either $\calC_m$ or $\calS_n$
along with a fixed integer denoting multiplicity. A matrix index is
then a set of matrix index pieces. Our graph matrix $M_\alpha$ appears
as a submatrix of those matrices: for a given set of square vertices,
order the squares in increasing order in a tuple, and assign it
multiplicity 1. Hence the same norm bounds apply.

Boolean norm bounds:
\begin{lemma}\label{lem:norm-bounds}
Let $V_{rel}(\alpha) \defeq V(\alpha) \setminus (U_\alpha \cap V_\alpha)$. There is a universal constant $C$ such that the following norm bound holds for all proper shapes $\alpha$ w.h.p.:
\[\norm{M_\alpha} \leq 2\cdot\left(\abs{V(\alpha)} \cdot \log(n)\right)^{C\cdot \abs{V_{rel}(\alpha)}} \cdot n^{\frac{w(V(\alpha)) - w(S_{\min}) + w(W_{iso})}{2}} \]
\end{lemma}
\begin{proof}
From Corollary 8.13 of~\cite{AMP20}, with probability at least $1-\eps$ for a fixed shape $\alpha$,
\[\norm{M_\alpha} \leq 2 \abs{V(\alpha)}^{\abs{V_{rel}(\alpha)}}\cdot \left( 6e \ceil{\frac{\log\left(\frac{n^{w(S_{\min})}}{\eps}\right)}{6\abs{V_{rel}(\alpha)}}}\right)^{\abs{V_{rel}(\alpha)}} \cdot n^{\frac{w(V(\alpha)) - w(S_{\min}) + w(W_{iso})}{2}}\]
Letting $N_k$ be the number of distinct shapes on $k$ vertices (either
circles or squares), we apply the corollary with $\eps = 1/(mn
N_{\abs{V(\alpha)}})$. Union bounding, the failure probability across
all shapes of size $k$ is at most $1/mn$, and since the number of
vertices in a shape is at most $m + n \leq 2m$, we have a bound that
holds with high probability for all shapes. It remains to simplify the
exact bound.

\begin{proposition}\label{prop:boolean-shape-counting}
$N_k \leq 8^k 2^{k^2}$
\begin{proof}
The following process forms all shapes on $k$ vertices: starting from $k$ formal variables, assign each variable to be either a circle or a square, decide whether each variable is in $U_\alpha$ and/or $V_\alpha$, then among the $k^2$ variable pairs put any number of edges.
\end{proof}
\end{proposition}
We also bound $n^{w(S_{\min})} \leq (mn)^{\abs{V(\alpha)}}$.
\begin{align*}
    \norm{M_\alpha} & \leq 2 \abs{V(\alpha)}^{\abs{V_{rel}(\alpha)}}\cdot \left( 6e \ceil{\frac{\log\left(n^{w(S_{\min})} \cdot mn N_{\abs{V(\alpha)}}\right)}{6\abs{V_{rel}(\alpha)}}}\right)^{\abs{V_{rel}(\alpha)}} \cdot n^{\frac{w(V(\alpha)) - w(S_{\min}) + w(W_{iso})}{2}} \\
    & \leq 2 \abs{V(\alpha)}^{\abs{V_{rel}(\alpha)}}\cdot \left( 12e \log\left(n^{w(S_{\min})} \cdot mn N_{\abs{V(\alpha)}}\right)\right)^{\abs{V_{rel}(\alpha)}} \cdot n^{\frac{w(V(\alpha)) - w(S_{\min}) + w(W_{iso})}{2}} \\
    & \leq 2 \abs{V(\alpha)}^{\abs{V_{rel}(\alpha)}}\cdot \left( 12e \log\left((mn)^{\abs{V(\alpha)}} \cdot mn\cdot 8^{\abs{V(\alpha)}} 2^{\abs{V(\alpha)}^2}\right)\right)^{\abs{V_{rel}(\alpha)}} \cdot n^{\frac{w(V(\alpha)) - w(S_{\min}) + w(W_{iso})}{2}}\\
    &  \leq 2 \abs{V(\alpha)}^{\abs{V_{rel}(\alpha)}}\cdot \left( 100e \abs{V(\alpha)}^2 \log\left(mn\right)\right)^{\abs{V_{rel}(\alpha)}} \cdot n^{\frac{w(V(\alpha)) - w(S_{\min}) + w(W_{iso})}{2}}\\
    & \leq 2\cdot\left(\abs{V(\alpha)} \cdot \log(mn)\right)^{3\cdot \abs{V_{rel}(\alpha)}} \cdot n^{\frac{w(V(\alpha)) - w(S_{\min}) + w(W_{iso})}{2}}
\end{align*}
Note that we now assume $m \leq n^2$.
\end{proof}

We have the following norm bound for Hermite shapes. For a Hermite shape $\alpha$, define the \textit{total size} to be $\abs{U_\alpha} + \abs{V_\alpha} + \abs{W_\alpha} + \abs{E(\alpha)}$.
\begin{lemma}\label{lem:gaussian-norm-bounds}
Let $V_{rel}(\alpha) \defeq V(\alpha) \setminus (U_\alpha \cap V_\alpha)$ as sets. There is a universal constant $C$ such that the following norm bound holds for all proper shapes $\alpha$ with total size at most $n$ w.h.p.:
\[ \norm{M_\alpha} \leq 2\cdot\left(\abs{V(\alpha)} \cdot (1+\abs{E(\alpha)}) \cdot \log(n)\right)^{C\cdot (\abs{V_{rel}(\alpha)} + \abs{E(\alpha)})} \cdot n^{\frac{w(V(\alpha)) - w(S_{\min}) + w(W_{iso})}{2}}\]
\end{lemma}

The proof performs the same calculation starting from~\cite[Corollary 8.15]{AMP20}. Note that in our notation, $l(\alpha) = \abs{E(\alpha)}$. There is a further difference which is that~\cite{AMP20} uses normalized Hermite polynomials whereas we use unnormalized Hermite polynomials; this contributes the additional term $\prod_{e \in E(\alpha)} l(e)! \leq (1+\abs{E(\alpha)})^{\abs{E(\alpha)}}$. We must replace Proposition~\ref{prop:boolean-shape-counting} with the following:
\begin{proposition}\label{prop:gaussian-shape-counting}
The number of Hermite shapes with total size $k$ is at most $k2^k(k+1)^{2k+k^2}$.
\begin{proof}
Such a shape has at most $k$ distinct variable vertices. Each of these is either a circle or a square. Each variable can be in $U_\alpha$ with multiplicity between 0 and (at most) $k$, and also in $V_\alpha$ with multiplicity between 0 and $k$. The $k^2$ possible pairs of vertices can have edge multiplicity in $E(\alpha)$ between 0 and $k$. 
\end{proof}
\end{proposition}

\section{Properties of $e(k)$}

We establish some properties of the $e(k)$ used in the
analysis. Recall that $e(k) = \E_{x \in \mathcal{S}(\sqrt{n})}\left[x_1\dots x_k\right]$ where
$\mathcal{S}(\sqrt{n}) \coloneqq \set{x \in \set{\pm
1}^n \mid \sum_{i=1}^n x_i = \sqrt{n}}$.

\begin{claim}\label{claim:e2}
  $e(2)=0$.
\end{claim}

\begin{proof}
  Fix $y \in \mathcal{S}(\sqrt{n})$. Note that $(\sum_{i=1}^n y_i)^2 = n$ implying
  $\sum_{i < j} y_i y_j = 0$. Using this fact, we get
  $$
  \E_{x \in \mathcal{S}(\sqrt{n})}\left[x_1 x_2\right] = \E_{\sigma \in S_n} y_{\sigma(1)} y_{\sigma(2)} = 0,
  $$
  concluding the proof.
\end{proof}

\begin{definition}
   We say that a tuple $\lambda = (\lambda_1,\dots,\lambda_k)$ of non-negative integers is a partition of $k$
   provided $\sum_{i=1}^k \lambda_i = k$ and $\lambda_1 \ge \cdots \ge \lambda_k$. We use the notation $\lambda \vdash k$
   to denote a partition of $k$. We refer to $\lambda_i$ as a row/part of $\lambda$. 
\end{definition}

In the following, we will dealing polynomials that can be indexed by integer partitions.
For this reason, we now fix a notation for partitions and some associated objects.

\begin{definition}
  The transpose of partition $\lambda = (\lambda_1,\dots,\lambda_k)$ is denoted $\lambda^t$ and defined as
  $\lambda^t_i = \abs{\set{j \in [k] \mid \lambda_j \ge i}}$.
\end{definition}

\begin{remark}
  For a partition $\lambda \vdash k$, $\lambda^t_1$ is the number of rows/parts of $\lambda$.
\end{remark}

\begin{definition}
  The automorphism group of a partition $\aut(\lambda) \leq S_{\lambda^t_1}$ is the group generated by transpositions $(i,j)$
  of rows $\lambda_i = \lambda_j$.
\end{definition}

\begin{remark}
  Let $\lambda \vdash k$ and $p_1(\lambda),\dots,p_k(\lambda)$ be such that $p_i(\lambda) = \abs{\set{j \in [\lambda_1^t] \mid \lambda_j = i}}$.
  Then $\aut(\lambda) \simeq S_{p_1} \times \cdots \times S_{p_k}$.
\end{remark}

\begin{lemma}\label{lem:slice_inv_exact}
  We have
  \[
  \sum_{\lambda \vdash k} \frac{\lambda!}{\lambda_1!\cdots \lambda_k!} \cdot \frac{(n)_{\lambda^t_1}}{\abs{\aut(\lambda)}} \cdot \E_{x \in \mathcal{S}(\sqrt{n})}\left[x_1^{\lambda_1} \dots x_k^{\lambda_k} \right] = n^{k/2}.
  \]
\end{lemma}

\begin{proof}
  For $x \in \mathcal{S}(\sqrt{n})$, we have $(\sum_{i=1}^n x_i)^k =
  n^{k/2}$. Then expanding $(\sum_{i=1}^n x_i)^k$ in the previous equations and
  taking the expectation over $\mathcal{S}(\sqrt{n})$ on both sides
  yields the result of the lemma (after appropriately collecting
  terms).
\end{proof}

\begin{claim}\label{claim:bound_prod_exp_ff}
  Let $\lambda \vdash k$. We have
  $$
  (n)_{\lambda^t_1} \cdot \abs{\E_{x \in \mathcal{S}(\sqrt{n})}\left[x_1^{\lambda_1} \dots x_k^{\lambda_k} \right]} \le 3^{k^3} \cdot n^{k/2}.
  $$
\end{claim}

\begin{proof}
  We induct on $k$. For $k=1$, we have $n \cdot \abs{\E_{x \in \mathcal{S}(\sqrt{n})}\left[x_1\right]} = \sqrt{n} \le 3 \cdot n^{1/2}$.
  Now, suppose $k \ge 2$. We consider three cases:
  \begin{enumerate}
    \item Case $\lambda_1 \ge 3$: Let $\lambda'$ be the partition obtained from $\lambda$ by removing two boxes from $\lambda_1$.
          Note that $\lambda_1^t = (\lambda')^t_1 \le k-2$ and
          $\E_{x \in \mathcal{S}(\sqrt{n})}\left[x_1^{\lambda_1'} \dots x_{k-2}^{\lambda_{k-2}'} \right] = \E_{x \in \mathcal{S}(\sqrt{n})}\left[x_1^{\lambda_1} \dots x_{k-2}^{\lambda_{k-2}} \right]$.
          By the induction hypothesis, we have $(n)_{(\lambda')^t_1} \cdot \abs{\E_{x \in \mathcal{S}(\sqrt{n})}\left[x_1^{\lambda_1'} \dots x_{k-2}^{\lambda_{k-2}'} \right]} \le 3^{(k-2)^2} \cdot n^{(k-2)/2}$.
    \item Case $\lambda_1 = 2$: Let $\lambda'$ be the partition obtained from $\lambda$ by removing $\lambda_1$.
          Note that $\lambda_1^t = (\lambda')^t_1 + 1 \le k-2$. By the induction hypothesis, we have
          $$
         (n)_{\lambda^t_1} \cdot \abs{\E_{x \in \mathcal{S}(\sqrt{n})}\left[x_1^{\lambda_1} \dots x_{k-2}^{\lambda_{k-2}} \right]} \le n \cdot (n)_{(\lambda')^t_1} \cdot \abs{\E_{x \in \mathcal{S}(\sqrt{n})}\left[x_1^{\lambda_1'} \dots x_{k-2}^{\lambda_{k-2}'} \right]} \le 3^{(k-2)^3} \cdot n^{k/2}.
          $$
    \item Case  $\lambda_1 = 1$: To bound $(n)_k \cdot \E_{x \in \mathcal{S}(\sqrt{n})}\left[x_1^{\lambda_1} \dots x_k^{\lambda_k} \right]$, we use~\cref{lem:slice_inv_exact}
          and the two preceding cases. Let $p(k)$ be the partition function, i.e., $p(k) = \abs{\set{\lambda \vdash k}}$. We  deduce that
            \begin{align*}
             (n)_k \cdot \abs{\E_{x \in \mathcal{S}(\sqrt{n})}\left[x_1^{\lambda_1} \dots x_k^{\lambda_k}\right]}  & \le n^{k/2} + \sum_{\lambda \vdash k \colon \lambda_1 \ge 2} \frac{\lambda!}{\lambda_1!\cdots \lambda_k!} \cdot \frac{(n)_{\lambda^t_1}}{\abs{\aut(\lambda)}} \cdot \abs{\E_{x \in \mathcal{S}(\sqrt{n})}\left[x_1^{\lambda_1} \dots x_k^{\lambda_k} \right]} \\
                                                                                                 & \le n^{k/2} + k!  \sum_{\lambda \vdash k \colon \lambda_1 \ge 2} (n)_{\lambda^t_1} \cdot \abs{\E_{x \in \mathcal{S}(\sqrt{n})}\left[x_1^{\lambda_1} \dots x_k^{\lambda_k} \right]} \\
                                                                                                 & \le n^{k/2} + k! \sum_{\lambda \vdash k \colon \lambda_1 \ge 3} (n)_{\lambda^t_1} \cdot \abs{\E_{x \in \mathcal{S}(\sqrt{n})}\left[x_1^{\lambda_1} \dots x_k^{\lambda_k} \right]}  + \\
                                                                                                         & \qquad \qquad k !\sum_{\lambda \vdash k \colon \lambda_1 = 2} (n)_{\lambda^t_1} \cdot \abs{\E_{x \in \mathcal{S}(\sqrt{n})}\left[x_1^{\lambda_1} \dots x_k^{\lambda_k} \right]}\\                                                                                                 
                                                                                                 & \le 3^{(k-2)^3} \cdot k! \cdot (1 + p(k) + k) \cdot n^{k/2} \le 3^{k^3} \cdot n^{k/2},
          \end{align*}
          as desired.
  \end{enumerate}
\end{proof}

\begin{claim}\label{claim:crude_bound_e}
  Suppose $k < \sqrt{n}/2$. We have
  $$
  \abs{\E_{x \in \mathcal{S}(\sqrt{n})}\left[x_1\dots x_k\right]} \le 2 \cdot 3^{k^3} \cdot n^{-k/2}.
  $$
\end{claim}

\begin{proof}
  Follows from~\cref{claim:bound_prod_exp_ff} and the bound on $k$.
\end{proof}

\begin{remark}
  In~\cref{claim:crude_bound_e}, the factor $3^{k^3}$ is too lossy to
  allow a meaningful bound with $k = n^{\epsilon}$, where $\epsilon >
  0$ is a constant.
\end{remark}

Refining the ideas of~\cref{claim:bound_prod_exp_ff}, we prove a
stronger lemma below which will imply a tighter bound on $e(k)$
sufficient for our application.
\begin{lemma}\label{lem:slice_inv_exp}
   There exists an universal constant $C \ge 1$ such that
   \begin{equation}\label{eq:abs_e_k_sum}
     \sum_{\lambda \vdash k} \frac{\lambda!}{\lambda_1!\cdots \lambda_k!} \cdot \frac{(n)_{\lambda^t_1}}{\abs{\aut(\lambda)}} \cdot \abs{\E_{x \in \mathcal{S}(\sqrt{n})}\left[x_1^{\lambda_1} \dots x_k^{\lambda_k} \right]} \le  k^{C \cdot k} \cdot n^{k/2}.
  \end{equation}
  In particular, for $n \ge 6$,~\cref{eq:abs_e_k_sum} holds with $C=2$.
\end{lemma}

\begin{proof}
  We induct on $k$. For $k = 1$, we have $n \cdot \abs{\E_{x \in \slice(\sqrt{n})} x_1} \le \sqrt{n}$ as desired.
  Using $e(2) = 0$ from~\cref{claim:e2} and the case $k=1$ of~\cref{eq:abs_e_k_sum}, we get that~\cref{lem:slice_inv_exp}
  also holds for $k=2$. Now, consider $k \ge 3$. Let $\Lambda_1 = \set{\lambda \vdash k \mid \lambda_1 = 1}$,
  $\Lambda_2 = \set{\lambda \vdash k \mid \lambda_1 = 2}$ and $\Lambda_{\ge 3} = \set{\lambda \vdash k \mid \lambda_1 = 3}$.
  Note that $\Lambda_1 \sqcup \Lambda_2 \sqcup \Lambda_{\ge 3} = \set{\lambda \vdash k}$ and $\abs{\Lambda_1} = 1$.

  For convenience define $a_{\lambda}$ to be the term associated to $\lambda \vdash k$ on the LHS of~\cref{eq:abs_e_k_sum}, i.e.,
  $$
  a_{\lambda}
  = \frac{\lambda!}{\lambda_1!\cdots \lambda_k!} \cdot \frac{(n)_{\lambda^t_1}}{\abs{\aut(\lambda)}} \cdot \abs{\E_{x \in \mathcal{S}(\sqrt{n})}\left[x_1^{\lambda_1} \dots
  x_k^{\lambda_k} \right]}.
  $$ 

  First we bound the contribution of the terms associated to
  partitions from $\Lambda_{\ge 3}$ in the LHS
  of~\cref{eq:abs_e_k_sum}. Let $\lambda'$ be the partition obtained
  from $\lambda$ by removing two boxes from $\lambda_1$.  Note that
  $\lambda_1^t = (\lambda')^t_1 \le k-2$ and
  $\E_{x \in \mathcal{S}(\sqrt{n})}\left[x_1^{\lambda_1'} \dots
  x_{k-2}^{\lambda_{k-2}'} \right]
  = \E_{x \in \mathcal{S}(\sqrt{n})}\left[x_1^{\lambda_1} \dots
  x_{k-2}^{\lambda_{k-2}} \right]$. Thus,
  \begin{align*}
  &a_{\lambda} = \frac{\lambda!}{\lambda_1!\cdots \lambda_k!} \cdot \frac{(n)_{\lambda^t_1}}{\abs{\aut(\lambda)}} \cdot \abs{\E_{x \in \mathcal{S}(\sqrt{n})}\left[x_1^{\lambda_1} \dots x_k^{\lambda_k} \right]}  \\
  &  \qquad   = \frac{k(k-1)}{\lambda_1 (\lambda_1-1)} \cdot \frac{\abs{\aut(\lambda')}}{\abs{\aut(\lambda)}} \frac{\lambda'!}{\lambda_1'!\cdots \lambda_k'!} \cdot \frac{(n)_{(\lambda')^t_1}}{\abs{\aut(\lambda')}} \cdot \abs{\E_{x \in \mathcal{S}(\sqrt{n})}\left[x_1^{\lambda_1'} \dots x_{k-2}^{\lambda_{k-2}'} \right]}  \\
  & \qquad = k^2 \cdot \frac{\abs{\aut(\lambda')}}{\abs{\aut(\lambda)}} \cdot a_{\lambda'} \le k^3 \cdot a_{\lambda'},
  \end{align*}
  since $\abs{\aut(\lambda')}/\abs{\aut(\lambda)} \le k-2 \le k$.
  For each $\lambda' \vdash k-2$, we can form a partition $\lambda \vdash k$ in $k-2 \le k$
  ways by adding two blocks to a single row of $\lambda'$. Hence, we have
  \begin{equation}\label{eq:lamb_ge_3_contrib}
    \sum_{\lambda \in \Lambda_{\ge 3}} a_{\lambda} \le k \cdot \sum_{\lambda' \vdash k-2} k^3 \cdot a_{\lambda'} \le k^4 \cdot k^{C \cdot (k-2)} \cdot n^{(k-2)/2},
  \end{equation}
  where the last equality follows from the induction hypothesis.

  Now we bound the contribution of the terms $a_{\lambda}$ associated to partitions $\lambda$
  from $\Lambda_{2}$ in the LHS of~\cref{eq:abs_e_k_sum}. Let $i \ge 1$
  be the number of parts of size two of $\lambda$
  and let $\lambda'$ be the partition obtained
  from $\lambda$ by removing these $i$ parts of size two.  Note that $\lambda_1^t =
  (\lambda')^t_1 + i \le k-1$. We have
  \begin{align*}
  &a_{\lambda} = \frac{\lambda!}{\lambda_1!\cdots \lambda_k!} \cdot \frac{(n)_{\lambda^t_1}}{\abs{\aut(\lambda)}} \cdot \abs{\E_{x \in \mathcal{S}(\sqrt{n})}\left[x_1^{\lambda_1} \dots x_k^{\lambda_k} \right]}  \\
  &  \qquad   \le n^i \cdot \frac{(k)_i}{2^i} \cdot \frac{\abs{\aut(\lambda')}}{\abs{\aut(\lambda)}} \cdot \frac{\lambda'!}{\lambda_1'!\cdots \lambda_k'!} \cdot \frac{(n)_{(\lambda')^t_1}}{\abs{\aut(\lambda')}} \cdot \abs{\E_{x \in \mathcal{S}(\sqrt{n})}\left[x_1 \dots x_{k-2i} \right]} \\
  &  \qquad   = n^i \cdot \frac{(k)_i}{2^i} \cdot \frac{1}{i!} \cdot \frac{\lambda'!}{\lambda_1'!\cdots \lambda_k'!} \cdot \frac{(n)_{(\lambda')^t_1}}{\abs{\aut(\lambda')}} \cdot \abs{\E_{x \in \mathcal{S}(\sqrt{n})}\left[x_1 \dots x_{k-2i} \right]},\\
  \end{align*} 
  where in the last equality we used $\abs{\aut(\lambda')}/\abs{\aut(\lambda)} = 1/(i!)$.
  Since $\lambda \in \Lambda_2$ is uniquely specified by its number of parts of size two, applying the induction hypothesis we have
  \begin{align*}
    \sum_{\lambda \in \Lambda_2} a_{\lambda} & \le \sum_{i=1}^{\lfloor k/2 \rfloor} n^i \cdot \frac{(k)_i}{2^i} \cdot \frac{1}{i!} \cdot \left(\frac{\lambda'!}{\lambda_1'!\cdots \lambda_k'!} \cdot \frac{(n)_{(\lambda')^t_1}}{\abs{\aut(\lambda')}} \cdot \abs{\E_{x \in \mathcal{S}(\sqrt{n})}\left[x_1 \dots x_{k-2i} \right]} \right)\\
    &\le \sum_{i=1}^{\lfloor k/2 \rfloor} n^i \cdot \frac{(k)_i}{2^i} \cdot \frac{1}{i!} \cdot k^{C\cdot(k-2i)} \cdot n^{(k-2i)/2}\\
    &\le  k^{C \cdot (k-1) } \cdot n^{k/2} \cdot \sum_{i=0}^{\infty} k^{- C \cdot i} \le  \frac{3}{2} \cdot k^{C \cdot (k-1) } \cdot n^{k/2},
  \end{align*}
  where in the last inequality we used $k \ge 3$ and $C \ge 1$.

  Finally, we consider the case $\lambda_1 = 1$. To bound $a_{\lambda}$, we use~\cref{lem:slice_inv_exact}
  and the two preceding cases. We deduce that
  \begin{align*}
    a_{\lambda} & \le n^{k/2} + \sum_{\mu \in \Lambda_2} a_{\mu} + \sum_{\mu \in \Lambda_{\ge 3}} a_{\mu} \le n^{k/2} + k^4 \cdot k^{C \cdot (k-2)} \cdot n^{(k-2)/2} + \frac{3}{2} \cdot k^{C \cdot (k-1) } \cdot n^{k/2}\\
              & = k^{C \cdot k} \cdot n^{k/2} \left( \frac{1}{k^{C \cdot k}} + \frac{k^4}{n \cdot k^{2 \cdot C}} + \frac{3}{2 \cdot k^{C}} \right).
  \end{align*}
   We can bound the LHS of~\cref{eq:abs_e_k_sum} as
  \begin{align*}
    \sum_{\mu \in \Lambda_1} a_{\mu} +  \sum_{\mu \in \Lambda_2} a_{\mu} + \sum_{\mu \in \Lambda_{\ge 3}} a_{\mu} &\le 
               k^{C \cdot k} \cdot n^{k/2} \left( \frac{1}{k^{C \cdot k}} + \frac{2 \cdot k^4}{n \cdot k^{2 \cdot C}} + \frac{3}{k^{C}} \right)\\
               &\le k^{C \cdot k} \cdot n^{k/2},
  \end{align*}
  provided $C > 0$ is a sufficiently large constant. In particular, the constant $C$ can be taken to be $2$ for $n \ge 6$.
\end{proof}

\begin{corollary}\label{cor:bound_on_coeff_e_k}
  We have
  $$
  \abs{\E_{x \in \mathcal{S}(\sqrt{n})}\left[x_1\dots x_k\right]} \le k^{3\cdot k} \cdot n^{-k/2}.
  $$  
\end{corollary}

\begin{proof}
   Suppose $k \le \sqrt{n}$.
   Note that~\cref{lem:slice_inv_exp} implies that for $\lambda \vdash k$ with $\lambda_1$
   there exists a constant $C > 0$ such that
   \begin{align*}
     \frac{\lambda!}{\lambda_1!\cdots \lambda_k!} \cdot \frac{(n)_{\lambda^t_1}}{\abs{\aut(\lambda)}} \cdot \abs{\E_{x \in \mathcal{S}(\sqrt{n})}\left[x_1^{\lambda_1} \dots x_k^{\lambda_k} \right]} &= (n)_k \cdot \abs{\E_{x \in \mathcal{S}(\sqrt{n})}\left[x_1\dots x_k\right]}\\
      &\le  k^{C \cdot k} \cdot n^{k/2}.
   \end{align*}
   Simplifying and using the assumption $k \le \sqrt{n}$, we obtain
   $$
   \abs{\E_{x \in \mathcal{S}(\sqrt{n})}\left[x_1\dots x_k\right]} \le \frac{k^{C \cdot k} \cdot n^{-k/2}}{\prod_{i=1}^{k-1} \left(1 - \frac{i}{n}\right)} \le 2 \cdot k^{C \cdot k} \cdot n^{-k/2}.
   $$
   Furthermore, for $n \ge 6$,~\cref{lem:slice_inv_exp} allows us to choose $C=2$.
   Since $\abs{\E_{x \in \mathcal{S}(\sqrt{n})}\left[x_1 \right]} = 1/\sqrt{n}$, the simpler
   bound applies for all values of $k$
   $$
   \abs{\E_{x \in \mathcal{S}(\sqrt{n})}\left[x_1\dots x_k\right]} \le  k^{3 \cdot k} \cdot n^{-k/2},
   $$
   Now the assumption $n \ge 6$ can be removed since, for $k \ge 2$, we have$(k^{3}/\sqrt{n})^{k} \ge 1$,
   where $1$ is the trivial bound. Similarly, our initial assumption of $k \le \sqrt{n}$ can also be removed
   as the bound also becomes trivial in the regime $k > \sqrt{n}$.
\end{proof}

\section{Combinatorial Proof of Lemma \ref{lem:fixed-moments}}\label{app:combinatorialpseudocalibration}
In this appendix, we give a combinatorial proof of Lemma \ref{lem:fixed-moments}. We recall the statement of Lemma \ref{lem:fixed-moments} here.
\begin{lemma}
Let $\alpha \in \N^n$. When $v$ is fixed and $b$ is fixed (not necessarily +1 or -1) and $d \sim N(0, I)$ conditioned on $\ip{v}{d} = b\norm{v}$, 
\[\E_{d}[h_{\alpha}(d)] = \frac{v^\alpha}{\norm{v}^{\abs{\alpha}}} \cdot h_{\abs{\alpha}}(b).\]
\end{lemma}
\begin{proof}
Again, it is sufficient to prove this lemma when $\norm{v} = 1$. For this proof, we need the following description of Hermite polynomials in terms of matchings and Isserlis' Theorem/Wick's Theorem.
\begin{fact}
\[
h_k(x) = \sum_{M: M \text{ is a matching on } [k]}{(-1)^{|M|}x^{k - 2|M|}}
\]
\end{fact}
\begin{theorem}[Isserlis' Theorem/Wick's Theorem]
For any vectors $u_1,\ldots,u_k$,
\[
E_{x \sim N(0,I)}\left[\prod_{j=1}^{k}{\ip{x}{u_j}}\right] = \sum_{M: M \text{ is a perfect matching on } [k]}\prod_{(i,j) \in M}{\ip{u_i}{u_j}}
\]
\end{theorem}
The idea behind this proof is to break up each coordinate vector $e_i$ into a component which is parallel to $v$ and a component which is perpendicular to $v$.
\begin{definition}
For each coordinate $i$, define $e_{i}^{\perp} = e_i - {v_i}v$
\end{definition}
\begin{proposition}
For any coordinate $i$, $\ip{e_{i}^{\perp}}{e_{i}^{\perp}} = 1-{v_i}^2$. For any pair of distinct coordintes $i$ and $i'$, $\ip{e_{i}^{\perp}}{e_{i'}^{\perp}} = -{v_i}{v_{i'}}$
\end{proposition}
\begin{proof}
Observe that for all i,
\[
\ip{e_{i}^{\perp}}{e_{i}^{\perp}} = \ip{e_{i} - {v_i}v}{e_{i} - v_{i}v} = \ip{e_{i}}{e_{i}}  - 2{v_i}\ip{v}{e_{i}} + {v_i}^2\ip{v}{v} = 1-{v_i}^2
\]
and if $i$ and $i'$ are distinct then
\[
\langle{e_{i}^{\perp},e_{i'}^{\perp}}\rangle = \langle{e_{i} - {v_i}v,e_{i'} - v_{i'}v}\rangle = \langle{e_{i},e_{i'}}\rangle  - {v_i}\langle{v,e_{i'}}\rangle - {v_{i'}}\langle{e_i,v}\rangle + {v_i}{v_{i'}}\langle{v,v}\rangle = -{v_i}{v_{i'}}
\]
\end{proof}
To evaluate $\E_{d}[h_{\alpha}(d)]$, we proceed as follows:
\begin{enumerate}
\item Break up each $d_i = \ip{d}{e_i}$ as $d_i = \ip{bv}{e_{i}} + \ip{d^{\perp}}{e_{i}} = bv_{i} + \ip{d^{\perp}}{e_{i}^{\perp}}$
where $d^{\perp}$ is the component of $d$ which is orthogonal to $v$.
\item Observe that since each $e^{\perp}_{i}$ is orthogonal to $v$, we can replace $d^{\perp}$ by a random vector $d' \sim N(0,I)$.
\item Apply Isserlis' Theorem/Wick's Theorem to evaluate these terms.
\end{enumerate}
For this calculation, it is convenient to think of $\alpha$ as a tuple of $|\alpha|$ elements where each $i \in [n]$ appears $\alpha_i$ times.
\begin{definition}
For each $j \in [\abs{\alpha}]$, we define $\alpha(j)$ to be the index $i$ such that $\sum_{i' = 1}^{i-1}{\alpha_{i'}} < j$ and $\sum_{i'=1}^{i}{\alpha_{i'}} \geq j$. For example, if $\alpha = (2,1,0,3)$ then $\alpha(1) = \alpha(2) = 1$, $\alpha(3) = 2$, and $\alpha(4) = \alpha(5) = \alpha(6) = 4$.
\end{definition}
In the special case when $\alpha(1),\ldots,\alpha(\abs{\alpha})$ are all distinct, 
\[
\E_{d}[h_{\alpha}(d)] = \E_d\Big[\prod_{j=1}^{\abs{\alpha}}{\ip{d}{e_{\alpha(j)}}}\Big] = \E_{d' \sim N(0,I)}\Big[\prod_{j=1}^{\abs{\alpha}}{\left(bv_{\alpha(j)} + \ip{d'}{e_{\alpha(j)}^{\perp}}\right)}\Big]
\]
In this case, we can associate a matching $M$ to each term we get after applying Isserlis' Theorem/Wick's Theorem as follows:
\begin{enumerate}
\item For each $j \in \abs{\alpha}$ where we have the $bv_{\alpha(j)}$ term, we take $j$ to be isolated.
\item For each pair of distinct $j,j' \in \abs{\alpha}$ such that we have the term $\ip{e_{\alpha
(j)}^{\perp}}{e_{\alpha(j')}^{\perp}}$ (which only happens if we start with the $\ip{d'}{e_{\alpha(j)}^{\perp}}$ and $\ip{d'}{e_{\alpha(j')}^{\perp}}$ terms and $e_{\alpha(j)}^{\perp}$ and $e_{\alpha(j')}^{\perp}$ are paired together after applying Isserlis' Theorem/Wick's Theorem), we add an edge between $j$ and $j'$ in $M$.
\end{enumerate}
We now have that
\begin{align*}
\E_{d}[h_{\alpha}(d)] &= \sum_{M:M \text{ is a matching on [\abs{\alpha}]}}{\left(\prod_{(j,j') \in M}{-v_{\alpha(j)}v_{\alpha(j')}}\right)\left(\prod_{j: j \text{ is unmatched by } M}{bv_{\alpha(j)}}\right)}\\ 
&= \left(\sum_{M:M \text{ is a matching on [\abs{\alpha}]}}{(-1)^{|M|}b^{\abs{\alpha} - 2|M|}}\right)\prod_{j=1}^{\abs{\alpha}}{v_{\alpha(j)}} \\
&= h_{\abs{\alpha}}(b)v^{\alpha} 
\end{align*}
For the general case, we use a similar idea although it is somewhat more complicated. In particular, we associate a multi-colored matching $M = M_{blue} \cup M_{red} \cup M_{purple}$ to each term. The idea is that whenever we have a blue edge, we could have had a red edge instead and vice versa, so we can combine terms with red and blue edges to make purple edges which gives us an ordinary matching as before. More precisely, the idea is as follows.
\begin{enumerate}
\item When we expand out $h_{\alpha}(d)$ in terms of matchings, we take $M_{blue}$ to be the union of these matchings.
\item For each $j \in \abs{\alpha}$ where we have the $bv_{\alpha(j)}$ term, we take $j$ to be isolated.
\item For each pair of distinct $j,j' \in \abs{\alpha}$ such that we have the term $\ip{e_{\alpha
(j)}^{\perp}}{e_{\alpha(j')}^{\perp}}$ (which only happens if we start with the $\ip{d'}{e_{\alpha(j)}^{\perp}}$ and $\ip{d'}{e_{\alpha(j')}^{\perp}}$ terms and $e_{\alpha(j)}^{\perp}$ and $e_{\alpha(j')}^{\perp}$ are paired together after applying Isserlis' Theorem/Wick's Theorem), we add an edge between $j$ and $j'$. If $\alpha(j') = \alpha(j)$ then we take this edge to be red and add it to $M_{red}$. If $\alpha(j') \neq \alpha(j)$ then we take this edge to be purple and add it to $M_{purple}$.
\end{enumerate}
We now implement this idea. We have that
\begin{align*}
&\E_{d}[h_{\alpha}(d)] = \sum_{M_{blue}: M_{blue} \text{ is a matching on } [\abs{\alpha}], \atop \forall (j,j') \in M_{blue}, \alpha(j) = \alpha(j')}{
(-1)^{|M_{blue}|}\E_{d}\Big[\prod_{j \in \abs{\alpha}: j \text{ is unmatched by } M_{blue}}{\ip{d}{e_{\alpha(j)}}}\Big]} \\
&=  \sum_{M_{blue}: M_{blue} \text{ is a matching on } [\abs{\alpha}], \atop \forall (j,j') \in M_{blue}, \alpha(j) = \alpha(j')}{
(-1)^{|M_{blue}|}\E_{d' \sim N(0,I)}\Big[\prod_{j \in \abs{\alpha}: j \text{ is unmatched by } M_{blue}}{\left(bv_{\alpha(j)} + \ip{d'}{e_{\alpha(j)}^{\perp}}\right)}\Big]}
\end{align*}
Expanding out $\E_{d' \sim N(0,I)}\Big[\prod_{j \in \abs{\alpha}: j \text{ is unmatched by } M_{blue}}{\left(bv_{\alpha(j)} + \ip{d'}{e_{\alpha(j)}^{\perp}}\right)}\Big]$ and applying Isserlis' Theorem/Wick's Theorem, we have that 
\begin{align*}
\E_{d}[h_{\alpha}(d)] &= \sum_{M_{blue},M_{red},M_{purple}}{(-1)^{|M_{blue}|}\prod_{(j,j') \in M_{red}}{(1-v^{2}_{\alpha(j)})}\prod_{(j,j') \in M_{purple}}{(-v_{\alpha(j)}v_{\alpha(j')})}} \\
&\prod_{j: j \text{ is unmatched by } M = M_{blue} \cup M_{red} \cup M_{purple}}{bv_{\alpha(j)}}
\end{align*}
where the sum is taken over all $M_{blue},M_{red},M_{purple}$ such that
\begin{enumerate}
\item $M = M_{blue} \cup M_{red} \cup M_{purple}$ is a matching on $[\abs{\alpha}]$ and $M_{blue},M_{red},M_{purple}$ are disjoint.
\item $\forall (j,j') \in M_{blue}, \alpha(j) = \alpha(j')$.
\item $\forall (j,j') \in M_{red}, \alpha(j) = \alpha(j')$.
\item $\forall (j,j') \in M_{purple}, \alpha(j) \neq \alpha(j')$.
\end{enumerate}
Since whenever we have a blue edge, we could have instead had a red edge and vice versa, for each distinct $j,j'$ such that $\alpha(j') = \alpha(j)$, we can combine terms which have a blue edge between $j$ and $j'$ with terms which have a red edge between $j$ and $j'$. A blue edge between $j$ and $j'$ has a coefficient of $-1$ and a red edge between $j$ and $j'$ has a coefficient of $1 - v^2_{\alpha(j)}$, so this effectively gives a purple edge with coefficient $-v^2_{\alpha(j)} = v_{\alpha(j)}v_{\alpha(j')}$. Thus, 
\begin{align*}
\E_{d}[h_{\alpha}(d)] &= \sum_{M:M \text{ is a matching on [\abs{\alpha}]}}{\left(\prod_{(j,j') \in M}{-v_{\alpha(j)}v_{\alpha(j')}}\right)\left(\prod_{j: j \text{ is unmatched by } M}{bv_{\alpha(j)}}\right)}\\ 
&= h_{\abs{\alpha}}(b)v^{\alpha} 
\end{align*}
\end{proof}

\section{Importance of Scaling}\label{app:scaling}
We remark that somewhat surprisingly, the scaling of the problem is important for our arguments. The reason this is somewhat surprising is that for the purpose of determining whether or not a matrix $M$ is PSD, the scaling of the rows and columns of $M$ doesn't matter. More precisely, we have the following proposition.
\begin{proposition}
For any symmetric $N \times N$ matrix $M$ and any $N \times N$ diagonal matrix $D$ such that $\forall i \in [N], D_{ii} \neq 0$, $M \succeq 0$ if and only if $DMD \succeq 0$.
\end{proposition}
However, for our techniques, we also use the fact that if $x$ is in the nullspace of $M$ then for the purposes of determining whether $M$ is PSD, we can freely add a non-negative multiple of $xx^T$ to $M$.
\begin{proposition}\label{addingnullspaceprop}
For any symmetric $N \times N$ symmetric matrix $M$, any vector $x$ such that $Mx = 0$, and any constant $c   $, $M \succeq 0$ if and only if $M+cxx^T \succeq 0$.
\end{proposition}
As shown by the following example, the set of matrices that can be obtained using Proposition \ref{addingnullspaceprop} depends on the scaling of $M$.

If $M = \begin{pmatrix}
1 & 1 & 2 \\
1 & 2 & 3 \\
2 & 3 & 5
\end{pmatrix}$, $x = \begin{pmatrix}
1\\
1\\
-1
\end{pmatrix}$, and $D = \begin{pmatrix}
1 & 0 & 0 \\
0 & 1 & 0 \\
0 & 0 & \lambda
\end{pmatrix}$
then $DMD = \begin{pmatrix}
1 & 1 & 2\lambda \\
1 & 2 & 3\lambda \\
2\lambda & 3\lambda & 5{\lambda}^2
\end{pmatrix}$ and 
\[
DMD + cD^{-1}x{x^T}D^{-1} = \begin{pmatrix}
1 + c & 1 + c & 2\lambda - \frac{c}{\lambda} \\
1 + c & 2 + c & 3\lambda - \frac{c}{\lambda} \\
2\lambda - \frac{c}{\lambda} & 3\lambda - \frac{c}{\lambda} & 5{\lambda}^2 + \frac{c}{{\lambda}^2}
\end{pmatrix}
\]
Scaling this so that the diagonal entries are $1$ gives the matrix 
\[
\begin{pmatrix}
1 & \frac{\sqrt{1+c}}{\sqrt{2+c}} & \frac{2\lambda - \frac{c}{\lambda}}{\sqrt{(1+c)(5{\lambda}^2 + \frac{c}{{\lambda}^2})}} \\
\frac{\sqrt{1+c}}{\sqrt{2+c}} & 1 & \frac{3\lambda - \frac{c}{\lambda}}{\sqrt{(2+c)(5{\lambda}^2 + \frac{c}{{\lambda}^2})}} \\
\frac{2\lambda - \frac{c}{\lambda}}{\sqrt{(1+c)(5{\lambda}^2 + \frac{c}{{\lambda}^2})}} & \frac{3\lambda - \frac{c}{\lambda}}{\sqrt{(2+c)(5{\lambda}^2 + \frac{c}{{\lambda}^2})}} & 1
\end{pmatrix}
\]
Note that the entries in the upper left $2 \times 2$ block only depend on $c$ and are different for each $c$ while the other off-diagonal entries also depend on $\lambda$. Thus, different $\lambda$ give different sets of matrices.
\label{app:comment_on_scaling}

\end{document}